\documentclass[mnsc]{informs3_freeuse}

\OneAndAHalfSpacedXI

\usepackage{natbib}
 \bibpunct[, ]{(}{)}{,}{a}{}{,}%
 \def\bibfont{\small}%
\TheoremsNumberedThrough     
\ECRepeatTheorems

\EquationsNumberedThrough    

\usepackage{booktabs} 

\usepackage{multirow}
\usepackage{amsfonts,mathtools}
\usepackage[dvipsnames]{xcolor}
\usepackage[colorlinks=true, allcolors=teal]{hyperref}
\usepackage{cleveref}
\usepackage[export]{adjustbox}
\usepackage{float}
\usepackage{diagbox}

\usepackage{caption}
\usepackage{subcaption}

\renewcommand{\mc}[1]{\multicolumn{1}{c}{#1}}

\DeclareMathOperator{\gain}{gain}
\DeclareMathOperator{\Var}{Var}
\newcommand{\maxweight}{\mathsf{max\text{-}weight}}
\newcommand{\dels}{\mathsf{dels}}
\newcommand{\weight}{\mathsf{weight}}
\newcommand{\totweight}{\mathsf{total\text{-}weight}}
\newcommand{\suchthat}{\ | \ }
\newcommand{\E}{\mathbb{E}}
\newcommand{\R}{\mathbb{R}}
\renewcommand{\P}{\mathbb{P}}
\newcommand{\N}{\mathbb{N}}

\newcommand{\D}{\mathcal{D}}
\newcommand{\U}{\mathcal{U}}
\newcommand{\cE}{\mathcal{E}}
\newcommand{\set}[1]{{\left\{#1\right\}}}
\newcommand{\bern}{\text{Bern}}
\newcommand{\bin}{\text{Bin}}
\newcommand{\pois}{\text{Pois}}
\newcommand{\betadist}{\text{Beta}}
\newcommand{\e}{\varepsilon}
\newcommand{\p}{\varphi}
\renewcommand{\a}{\alpha}
\renewcommand{\d}{\delta}
\newcommand{\g}{\gamma}
\newcommand{\emdash}{\,---\,}

\newif\ifcomments
\commentstrue

\begin{document}

\RUNTITLE{Berinsky et al. Tracking Truth with Liquid Democracy}


\ARTICLEAUTHORS{%
\AUTHOR{Adam J. Berinsky}
\AFF{Department of Political Science, MIT} 
\AUTHOR{Daniel Halpern}
\AFF{\AFF{School of Engineering and Applied Sciences, Harvard University}}
\AUTHOR{Joseph Y. Halpern}
\AFF{Department of Computer Science, Cornell University}
\AUTHOR{Ali Jadbabaie}
\AFF{Institute for Data, Systems, and Society, MIT}
\AUTHOR{Elchanan Mossel}
\AFF{Department of Mathematics, MIT}
\AUTHOR{Ariel Procaccia}
\AFF{School of Engineering and Applied Sciences, Harvard University}
\AUTHOR{Manon Revel}
\AFF{Institute for Data, Systems, and Society, MIT\\ Corresponding author: mrevel@mit.edu}
} 
\TITLE{Tracking Truth with Liquid Democracy}
\ABSTRACT{%
The dynamics of random transitive delegations on a graph are of particular interest when viewed through the lens of an emerging voting paradigm, \emph{liquid democracy}. This paradigm allows voters to choose between directly voting and transitively delegating their votes to other voters, so that those selected cast a vote weighted by the number of delegations they received. In the epistemic setting, where voters decide on a binary issue for which there is a ground truth, previous work showed that a few voters may amass such a large amount of influence that liquid democracy is less likely to identify the ground truth than direct voting. We quantify the amount of permissible concentration of power and examine more realistic delegation models, showing they behave well by ensuring that (with high probability) there is a permissible limit on the maximum number of delegations received. Our theoretical results demonstrate that the delegation process is similar to well-known processes on random graphs that are sufficiently bounded for our purposes. Along the way, we prove new bounds on the size of the largest component in an infinite P\'olya urn process, which may be of independent interest. In addition, we empirically validate the theoretical results, running six experiments (for a total of $N=168$ participants, $62$ delegation graphs and over $11k$ votes collected). We find that empirical delegation behaviors meet the conditions for our positive theoretical guarantees. Overall, our work alleviates concerns raised about liquid democracy and bolsters the case for the applicability of this emerging paradigm.
}%


\KEYWORDS{liquid democracy, random graph theory, collective-decision making, social choice theory} 

\maketitle

\section{Introduction}
\label{sec:intro}
\emph{Liquid democracy} is a voting paradigm that is conceptually situated between \emph{direct democracy}, in which voters have direct influence over decisions, and \emph{representative democracy}, where voters choose delegates who represent them for a period of time. Under liquid democracy, voters have a choice: they can either vote directly on an issue like in direct democracy, or delegate their vote to another voter, entrusting them to vote on their behalf. The defining feature of liquid democracy is that these delegations are \emph{transitive}: if voter 1 delegates to voter 2 and voter 2 delegates to voter 3, then voter 3 votes (or delegates) on behalf of all three voters.

In recent years, liquid democracy has gained prominence around the world. The most impressive example is that of the German Pirate Party, which adopted the \emph{LiquidFeedback} platform in 2010~\citep{kling2015voting}. Other political parties, such as the Net Party in Argentina and Flux in Australia, have run on the wily promise 
that, once elected, their representatives would be essentially controlled by voters through a liquid democracy platform. 
Companies are also
exploring the use of liquid democracy for corporate governance; Google, for example, has run a proof-of-concept experiment~\citep{hardt2015google}. Blockchain systems have also been experimenting with related  weighted decentralized voting systems~\citep{benhaim2023scaling, li2023liquid}.

Practitioners, however, recognize that there is a potential flaw in liquid democracy, namely, the possibility of \emph{concentration of power}, in the sense that certain voters amass a relatively large number of delegations, giving them pivotal influence over the final decision. This scenario seems inherently undemocratic\emdash and it is not a mere thought experiment. Indeed, in the LiquidFeedback platform of the German Pirate Party, a linguistics professor at the University of Bamberg received so many delegations that, as noted by Der Spiegel,\footnote{ See \url{https://tinyurl.com/y52j6nfs}.} his ``vote was like a decree.'' 

\citet{kahng} examine liquid democracy's concentration-of-power phenomenon from a theoretical viewpoint and establish a troubling 
impossibility result in what has been called the \emph{epistemic} setting, that is, one where there is a ground truth.\footnote{The use of the term ``epistemic'' in this context is well-established in the social choice literature~\citep{LG01,Piv12}.}
Informally, they demonstrate that, even under the strong assumption that voters 
delegate only to more ``competent'' voters, any ``local mechanism'' satisfying 
minimal conditions will, in certain instances, be subject to  concentration of power, leading to relatively low accuracy. More specifically, \citeauthor{kahng} model the problem as a decision problem where  voters decide on an issue with two outcomes, $\{0,1\}$, where $1$ is correct (the ground truth) and $0$ is incorrect. Each of the voters $i\in \{1, \ldots, n\}$ is characterized by a \emph{competence} $p_i \in [0,1]$. The binary vote $V_i$ of each voter $i$ is drawn independently from a Bernoulli distribution, that is, each voter votes correctly with probability $p_i$. Under direct democracy, the outcome of the election is determined by a majority vote: the correct outcome is selected if and only if more than half of the voters vote for the correct outcome; 
that is, it is correct if and only if $\sum_{i=1}^n V_i \geq n/2$. Under liquid democracy, there exists a set of weights, $\weight_i$ for each $i\in[n],$ which represent the number of votes that  voter $i$ gathered transitively after delegation (if voter $i$ delegates, then $\weight_i=0$). The outcome of the election is then determined by a weighted majority; it is correct if and only if $\sum_{i=1}^n \weight_iV_i \geq n/2$. 

\citeauthor{kahng} also introduce the concept of a \emph{delegation mechanism}, which determines whether voters delegate and, if so,  to whom they delegate. They are especially interested in \emph{local mechanisms}, where the delegation 
decision of a voter depends only on their local neighborhood according to an 
underlying social network. They assume that voters delegate only to those with strictly higher competence, which excludes the possibility of cyclic delegations. To evaluate liquid democracy, \citeauthor{kahng} test the intuition that society makes more informed decisions under liquid democracy than under direct democracy (especially given the foregoing assumption about upward delegation). To that end, they define the \emph{gain} of a delegation mechanism to be the difference between the probability the correct outcome is selected under liquid democracy and the probability the correct outcome
is selected under direct democracy. A delegation mechanism satisfies 
\emph{positive gain} if its gain is strictly positive in some cases, and it satisfies \emph{do 
no harm} if, for all $\e > 0$, its gain 
is at least $-\e$ for sufficiently large instances. 
Assuming that competence after delegation remains
strictly above $1/2$, this will follow from the law of large number that applies to the weighted majority with weights relatively spread out~\citep{haggstrom2006law}. 
The main result of \citeauthor{kahng} is that local mechanisms can never satisfy these two requirements. 
\citet{caragiannis2019contribution} further strengthen this negative result by showing that there are degenerate instances where local mechanisms perform much worse than either direct democracy or dictatorship (the most extreme concentration of power).\footnote{The former constructs an instance where even with arbitrarily many voters, a constant number will receive a majority of the delegations. The group has an average competence above $1/2$. The probability liquid democracy 
gives the right answer
can be upper bounded by a constant strictly below $1$, while direct democracy is correct with probability approaching one. In the latter case, the voters' numbers and relative competence are chosen 
so
that liquid democracy 
almost always gives the incorrect answer  (as does direct voting), 
while dictatorship is correct with a constant probability.
}

These results undermine the case for liquid democracy: the
benefits of delegation 
appear to be reversed by concentration of power. However, the negative conclusion relies heavily on worst-case modeling assumptions. Our research represents a significant 
advance as it offers a comprehensive framework that not only captures 
the worst-case scenarios of previous works, 
but also provides insights into more intriguing "high probability" cases. In 
particular, in this paper, we provide
a new theoretical model and extensive experiments that show 
that liquid democracy will typically satisfy a probabilistic 
version of \emph{positive gain} and \emph{do no harm} under minimal assumptions.

\subsection{Our Contributions and Techniques}
Our contributions are the following. First, building on the work of \citet{kahng}, we provide a general framework to analyze the stochastic network dynamics of transitive delegations that captures 
liquid democracy's intricate interactions between local-delegation choices and global properties. Second, we identify large classes of delegation models where liquid democracy performs well, in that delegations induce a sufficiently small amount of concentration of power and liquid democracy almost surely results in correct outcomes. Along the way, we prove new high-probability bounds on the size of the largest component in an infinite P\'olya urn process;\footnote{An infinite 
P\'olya
urn process models an urn process where each new ball picks its urn with a probability proportional to the size of the urn or creates its own urn with constant probability.} this result may be of independent interest. Finally, we conduct the first series of lab experiments on liquid democracy that can test epistemic performance. This involved over $11,000$ votes from $168$ participants in six experimental groups, where each group had pre-existing social ties. Our novel experimental design allows us to compare the performance of direct and liquid democracy, as well as to analyze properties of real voter delegation behavior. Importantly, the behaviors we observe align with one of the models we introduce, 
thus lending support to this approach. Taken together, these results 
exhibit
a regime in which liquid democracy displays 
promising
performance. We next elaborate on some of our specific techniques. 


\subsubsection{Stochastic Delegations}
Our point of departure from the existing literature is the way we model delegation in liquid democracy. To emphasize these differences, instead of calling these delegation functions \emph{mechanisms}, we instead call them delegation 
\emph{models}, as they are intended to capture independent voter behavior rather than 
prescribing to each voter to whom they must delegate. Our delegation models are defined by $M=(q,\p),$ where $q: [0, 1] \to [0, 1]$ is a function that maps a voter's competence to the probability they delegate and $\p: [0, 1]^2 \to \R_{\ge 0}$ maps a pair of competencies to a weight. In this model, each voter $i$ votes directly with probability $1-q(p_i)$ and, conditioned on delegating with probability $q(p_i)$, delegates to voter $j \neq i$ with probability proportional to $\p(p_i,p_j).$ 
These delegation functions do not model explicit reasoning; rather,
they model behaviors that may be influenced by tacit knowledge captured by $q$ and $\p.$ A voter does not need to ``know'' the competence of another voter to decide whether to delegate; rather, the delegation \emph{probabilities} are 
influenced by competence, as captured by $\p$
note that delegation cycles are possible, and we take a worst-case approach to dealing with them: If the delegations form a cycle, then all voters in the cycle are assumed to be incorrect (vote $0$).\footnote{In LiquidFeedback, delegation cycles are, in fact, ignored.}

The most significant difference between our model of delegation and that of \citet{kahng} is that in our model, each voter has a chance of delegating to any other voter, whereas in their model, an underlying social network restricts delegation options. Our model captures a connected world where, in particular, voters may have heard of experts on various issues even if they do not know them personally. Although our model eschews an explicit social network, it can be seen as embedded into the delegation process, where the probability that $i$ delegates to $j$ takes into account the probability that $i$ is familiar with $j$ in the first place. Another difference between our model and that of \citet{kahng} is that we model the competencies $p_1,\ldots,p_n$ as being sampled independently from a distribution $\mathcal{D}$. While this assumption is made mainly for ease of exposition, it allows us to avoid edge cases and obtain robust results. 

\subsubsection{Delegation Models}
Our goal is to identify delegation models that satisfy (probabilistic versions of) positive gain and do no harm. Our first technical contribution, in \Cref{sec:corelemma}, is the formulation of general conditions on the model and competence distribution that are sufficient for these properties to hold (\Cref{lem:core}). In particular, to achieve the more difficult do no harm property, we present conditions that guarantee the maximum weight $\maxweight(G_n)$ accumulated by any voter is sub-linear with high probability and the expected increase in 
competence after
delegation is at least a positive constant times the population size. These conditions prevent extreme concentration of power and ensure that the representatives 
after delegation are
sufficiently better than the entire population to compensate for any concentration of power that does happen. 

Although the proof is straightforward, the benefit of this lemma is that it then suffices to identify models and distribution classes that verify these conditions. A delegation model $M$ and a competence distribution $\mathcal{D}$ induce a distribution over delegation instances that generates random graphs in ways that relate to well-known graph processes, which we leverage to analyze our models. Specifically, we introduce three models, all shown to satisfy do no harm and positive gain under \emph{any} continuous distribution over competence levels. The first models, \emph{upward delegation} and \emph{confidence-based delegation}, are interesting but restricted case studies 
that demonstrate the robustness of our approach. 
By contrast, the \emph{general continuous delegation} model is, as the name suggests, quite general. Moreover, it is realistic: its predictions are consistent with our experiments.

\paragraph{Upward 
Delegation:} 
In \Cref{sec:up}, we consider a model according to which the 
probability $p$ of delegation
is exogenous and constant across competencies, $q(p_i)=p$, and 
delegation can occur only to voters with strictly higher 
competence (the weight that any voter $i$ puts on another voter $j$ is $\p(p_i,p_j)=\mathbb{I}_{\{p_j-p_i>0\}}.$
This model captures the fact that there might be some reluctance to delegate regardless of the voter's competence but does assume that voters act in the interest of society by only delegating to voters that are more competent than they are. 

To generate a random graph induced by such a model, one can add a single voter at a time in order of decreasing competence and allow the voter to either not delegate (with probability $1 - p$) and create their own disconnected component, or delegate to the creator of any other component with probability proportional the size of the component. This works because delegating to any voter in the previous components is possible (since they have strictly higher competence) and would result in the votes being concentrated in the originator of that component by transitivity. Such a process exactly generates a preferential attachment graph with a positive probability of not attaching to the existing components, also called an 
infinite P\'olya urn process~\citep{simon1955class}. We can 
then show that, with high probability, no component grows too large so long as $p<1$ (see Section~\ref{subsubsec:polya} for an overview of this step). Further, continuity of the competence distribution ensures that enough lower competence voters delegate to higher competence voters to sufficiently increase the average. 

\paragraph{Confidence-Based Delegation.} In \Cref{sec:med}, we consider a model in which voters delegate with probability decreasing in their competencies and choose someone at random when they delegate. That is, the probability $q(p_i)$ that any voter $i$ delegates is decreasing in $p_i$ and the weight that any voter $i$ gives to any voter $j$ is
$\p(p_i,p_j)=1$. In other words, in this model, competence does not affect the probability of receiving delegations, only the probability of delegating.  

To generate a random graph induced by such a model, one can begin from a random vertex and study the delegation tree that starts at that vertex. A delegation tree is defined as a branching process, where a node $i$'s ``children'' are the nodes that delegated to node $i$. In contrast to classical branching processes, the probability for a child to be born increases as the number of people who already received delegations decreases. Nevertheless, we prove that, with high probability, as long as a delegation tree is no larger than $O(\log n)$, our heterogeneous branching process is dominated by a sub-critical graph branching process~\citep{alon2016probabilistic}. We can then conclude that no component has size larger than $O(\log n)$ with high probability. Next, we show that the expected competence among the voters that do not delegate is strictly higher than the average competence. 

\paragraph{General Continuous Delegation.} Finally, we consider a general model in \Cref{sec:general} where the likelihood of delegation is fixed and the weight assigned to each voter when delegating is increasing in their competence. That is, each voter $i$ delegates with probability $q(p_i)=p$ and the weight that voter $i$ places on voter $j$ is $\p(p_i,p_j),$ where $\p$ is continuous and increases in its second coordinate. Thus, in this model, the delegation distribution is slightly skewed towards more competent voters. 

To generate a random graph induced by such a model, we again consider a branching process, but now voters $j$ and $k$ place different weights on $i$ per $\p$. Therefore, voters have a \emph{type} that governs their delegation behavior; this allows us to define a multi-type branching process with types that are continuous in $[0,1]$. The major part of the analysis is a proof that, with high probability, as long as the delegation tree is no larger than $O(\log n)$, our heterogeneous branching process is dominated by a sub-critical Poisson multi-type 
branching process. In a manner similar to Confidence-Based 
Delegation, we also show that there is an expected increase in competence 
after delegation.


\subsubsection{Component Sizes in Infinite P\'olya Urn Processes}
\label{subsubsec:polya}


Recall that to prove that upward delegation satisfies do no harm, we show that the largest component in an 
infinite P\'olya urn process
is sub-linear with high probability (\Cref{lem:polya}). We briefly expand on the proof as this result was, to the best of our knowledge, not previously known in the random graph literature, and may be of independent interest. 
We begin by focusing on the first $t^\gamma$ bins (for a suitably chosen $\gamma$ depending on the attachment probability $p$) and derive an upper bound on the expected size of these bins. This allows us to use Markov's inequality and union bound over all bins to show that simultaneously all of them are sublinear in size with high probability. 

Second, we take care of the remaining bins by observing that each additional bins's growth is isomorphic to a classic P\'olya urn process with two bins, whose limiting dynamic follows a Beta distribution. We analyze the rate of convergence, which allows us to give sufficiently strong bounds using Chebyshev's inequality after exactly $t-t^\gamma$ steps, and union bound over all of these bins, concluding that all are sublinear with high probability. 


\subsubsection{Consistency With Experiments} Lastly, we conduct six experiments to statistically estimate the functions $q$ and $\p$, and test the overall effectiveness of liquid democracy.
Participants were presented with several yes or no questions on 
various topics.  We call the set of questions related to each topic a \emph{task}. For each  task, 
participants
could either choose to vote directly or delegate their vote (for all questions) to another participant. They only saw the questions in a task if they chose to vote directly. In a later phase, they were asked to answers the questions they had delegated (and not seen) to see how they would have voted. This setup allows us to do a few things. First, it induces a matched-pair
design where, for each task and experiment, we can compare the accuracy of voting under liquid and direct democracy. Second, we use the answers to all questions to estimate participants'
competencies. Using this information, we study how delegation behavior depends on competence and investigate 
whether it is 
consistent
with the theoretical findings.


Results suggest that (i) competence is inversely correlated with the chance of delegation, and (ii) the likelihood of delegating to another voter increases with their competence. The results, therefore, support the assumptions and predictions made by the confidence-based and general continuous-delegation models. 
Taken together, these results 
exhibit
a regime in which liquid democracy is overall more likely to pinpoint the 
truth than direct
democracy.

\subsection{Related work}
The most closely related paper is that of \citet{kahng}, which was discussed in detail above. It is worth noting, though, that they complement their negative result with a positive one: when the mechanism can restrict the maximum number of delegations (transitively) received by any voter to $o(\sqrt{\log n})$, do no harm and positive gain are satisfied. Imposing such a restriction would require a central planner that monitors and controls delegations. \citet{golz} build on this idea: they study liquid democracy systems where voters may nominate multiple delegates and a central planner 
chooses a single delegate in order to minimize the maximum weight of any voter. Similarly, \citet{brill2018pairwise} introduces a process that allows voters to specify ordinal preferences over delegation options and possibly restricting or modifying delegations in a centralized way.  \citet{caragiannis2019contribution}, and then \citet{becker2021can} also consider central planners; they show that, for given competencies, the problem of choosing among delegation options to maximize the probability of a correct decision is hard to approximate. In any case, implementing these proposals would require a fundamental rethinking of the practice of liquid democracy. By contrast, our positive results show that \emph{decentralized} delegation models can be inherently self-regulatory, which supports the effectiveness of the current practice of liquid democracy. 

More generally, there has been a significant amount of theoretical research on liquid democracy in recent years. To give a few examples: \citet{green} studies whether it is rational for voters to delegate their vote from a utilitarian viewpoint; \citet{christoff} examine a similar question but in the context of voting on logically interdependent propositions; \citet{bloembergen2019rational, zhang2021power} and \citet{dhillon2023information} study liquid democracy from a game-theoretic viewpoint.



Next, our work builds on the random graph literature, as our delegation processes are related to well-known stochastic graph processes. Upward delegation can be viewed as a generalization of the preferential attachment model where agents do not attach to the existing component(s) with a fixed probability. Classical preferential attachment models assume that a new node attaches to 
an
existing node $n_0$ with probability (parameterized by an \textit{attachment function}) depending on the degree of $n_0$ \citep{barabasi1999emergence,durrett2007random}. 
In our framework, a new component may be created with sfixed probability,
a setup introduced by \citet{simon1955class} and usually referred to as an infinite P\'olya urn process. Others have studied the distribution of degrees~\citep{drinea2001variations}, the distribution of the number of components with $k$ people at time $t$~\citep{chung2003generalizations}, and the conditions for the emergence of infinite components \citep{collevecchio2013preferential}. However, to the best of our knowledge, the existing results do not allow us to derive bounds on the size of the largest component with high probability after a finite amount of time.


  

In terms of our experiments on liquid democracy, ours is the first paper to conduct experiments with human subjects. Previous papers have studied different aspects of liquid democracy through experiments in corporate \citep{hardt2015google} and political environments 
Independent of and essentially cooncurrent with our work,
\citet{campbell2022liquid} tested a game-theoretic formulation of liquid democracy. Unlike our experiments, they used online platforms to gather participants who did not know each other. Participants were assigned a probability of being correct and asked whether they would want to delegate to others, with experts (those with the highest probability of being correct) being publicly known. The delegations were randomly assigned to the pre-determined experts in one set-up, and through the random dot kinematogram task in another one. The group sizes considered  are $5$ people with one expert, $15$ people with $3$ experts and $125$ people with $25$ experts. While this study 
reveals
interesting connections between individuals' perceived competence and delegation behavior, it cannot investigate how experts are (or are not) identified endogenously through interpersonal knowledge embedded in a social networks, since the participants do not know each other.

Last, our work relates to recent advances in managerial studies that consider novel forms of governance, such as corporate governance~\citep[e.g.,][]{huang2023thy}, blockchain technologies~\citep[e.g.,][]{benhaim2023scaling, li2023liquid}, and prediction markets~\citep[e.g.,][]{chen2008modeling,atanasov2017distilling}.

\section{Model}
\label{sec:model}
There is a set of $n$ voters, denoted $[n] = \{1, \ldots, n\}$. We assume
voters are making a decision on a binary issue 
with possible answers 0 and 1; there is a correct alternative ($1$)
and an incorrect alternative ($0$). Each voter $i$ has a \emph{competence
level} $p_i \in [0, 1]$ which is the probability that $i$ votes
correctly. We denote the vector of competencies by $\vec{p}_n = (p_1, \ldots, p_n)$. When $n$ is clear from the context, we sometimes
drop it from the notation.

\paragraph{Delegation graphs}
A \emph{delegation graph} $G_n = ([n], E)$ on $n$ voters is a directed
graph with voters as vertices and a directed edge $(i, j) \in E$
denoting that $i$ delegates their vote to $j$.
Again, if $n$ is clear from context, we occasionally drop it from the
notation.
The outdegree of a vertex in the delegation graph is
at most $1$ since each voter can delegate to at most one
person. Voters that do not delegate have no outgoing edges.
In a delegation graph $G_n$, the
\emph{delegations received} by a voter $i$, $\dels_i(G_n)$, is
defined as the total number of people that (transitively) delegated to
$i$ in $G_n$, (i.e., the total number of ancestors of $i$ in $G_n$). The
\emph{weight} 
of a voter $i$, $\weight_i(G_n)$, is $\dels_i(G_n)+1$ (the number of delegation they received plus their own weight) if $i$ votes directly,
and $0$ if $i$ delegates.  We define
$\maxweight(G_n) = \max_{i \in [n]} \weight_i(G_n)$ to be the largest weight of any voter and define
$\totweight(G_n) = \sum_{i = 1}^n \weight_i(G_n)$. Since each vote is counted at most once, we have that
$\totweight(G_n) \le n$. However, note that if delegation edges form a
cycle, then the weight of the voters on the cycle and voters
delegating into the cycle are all set to $0$ and hence will not be
counted. In particular, this means that $\totweight(G_n)$ may be
strictly less than $n$.\footnote{This is a worst-case approach where
cycles can only hurt the performance of liquid democracy, since this assumption is equivalent to assuming that all voters on the cycles
vote incorrectly.}  

\paragraph{Delegation instances}
We call the tuple $(\vec{p}_n, G_n)$ a \emph{delegation instance}, or
simply an instance, on $n$ voters. Let $V_i = 1$ if voter $i$ would vote correctly if $i$ did vote, and
$V_i=0$ otherwise.
Fixed competencies $\vec{p}_n$
induce a probability measure $\P_{\vec{p}_n}$ over the $n$ possible
binary votes $V_i$, where $V_i \sim \bern(p_i)$. Given votes 
$V_1, \ldots, V_n$, we let $X^D_n$ be the number of correct votes
under direct democracy, that is, $X^D_n = \sum_{i = 1}^n V_i$. We let
$X^F_{G_n}$ be the number of correct votes under liquid democracy with
delegation graph $G_n$, that is, $X^F_{G_n} = \sum_{i = 1}^n
\weight_i(G_n) \cdot V_i$. The probability that direct democracy and liquid democracy are correct
are $\P_{\vec{p}_n}[X^D_n > n/2]$ and  $\P_{\vec{p}_n}[X^F_{G_n} > n/2]$, respectively. 

\paragraph{Gain of a delegation instance} We define the \emph{gain} of an instance as
$$\gain(\vec{p}_n, G_n) = \P_{\vec{p}_n}[X^F_{G_n} > n/2] - \P_{\vec{p}_n}[X^D_n > n/2].$$ In words, it is the difference between the probability that liquid democracy is correct and the probability that majority is correct.

\paragraph{Randomization over delegation instances}
In general, we assume that both competencies and delegations are chosen randomly. Each voter's competence $p_i$ is sampled i.i.d.\@ from a fixed distribution $\D$ with support contained in $[0, 1]$. Delegations will be chosen according to a \emph{model} $M$. A model $M = (q, \p)$ is composed of two parts. The first $q: [0, 1] \to [0, 1]$ is a function that maps competencies to the probability that the voter delegates. The second $\p: [0, 1]^2 \to \R_{\ge 0}$ maps pairs of competencies to a weight.
A voter $i$ with competence $p_i$ will choose how to delegate as follows:
\begin{itemize}
    \item[--] With probability $1 - q(p_i)$ they do not delegate.
    \item[--] With probability $q(p_i)$, $i$ delegates; $i$ places
      weight $\p(p_i, p_j)$ on each voter $j \ne i$ and randomly
      sample another voter $j$ to delegate to proportional to these
      weights. In the degenerate case where $\p(p_i, p_j) = 0$ for all $j \ne i$, we assume that $i$ does not delegate. 
\end{itemize}

A competence distribution $\D$, a model $M$, and a number $n$ of voters induce a probability measure $\P_{\D, M,n}$ over all
instances $(\vec{p}_n, G_n)$ of size $n$.

We can now
redefine the \emph{do no harm (DNH)} and \emph{positive gain (PG)}
properties from \citet{kahng} in a probabilistic way.  

\begin{definition}[Probabilistic do no harm]
A model $M$ satisfies \emph{probabilistic do no harm} with respect to a  class $\mathfrak{D}$ of distributions if, for all distributions
$\D \in \mathfrak{D}$ and all $\e, \d > 0$, there exists $n_0 \in \N$
such that for all $n \ge n_0$, $$\P_{\D, M,n}[\gain(\vec{p}_n, G_n)
  \ge  - \e] > 1 - \d.$$ 
\end{definition}

\begin{definition}[Probabilistic positive gain]
A model $M$ satisfies \emph{probabilistic positive gain} with
respect to a class $\mathfrak{D}$ of distributions if there exists a distribution $\D \in \mathfrak{D}$ such that for all $\e, \d > 0$, there exists $n_0 \in N$ such that for all $n \ge n_0$, $$\P_{\D, M,n}[\gain(\vec{p}_n, G_n) \ge 1 - \e] > 1 - \d.$$
\end{definition}

\footnote{Note 
that positive gain and do no harm relate to the notion of concentration of the weighted sum $\sum_{i=1}^n \weight_iV_i$. Indeed, the probability of direct democracy being correct approaches $1$ as $n$ increases when the average competence is strictly above 
$1/2$.
As a result, do no harm is satisfied by a delegation model exactly when the probability that liquid democracy is correct also approaches $1$. This happens when the competence 
after delegation
remains strictly above
$1/2$, and the weighted sum $\sum_{i=1}^n \weight_iV_i$ concentrates. 
Positive gain also holds
if there exists a setup where the average group competence is strictly below 
$1/2$,
and the average competence 
after delegation
remains strictly above 
$1/2$
and the weighted sum $\sum_{i=1}^n \weight_iV_i$ concentrates. In turn, these established benchmarks are directly mapped to existing results in social choice theory on the convergence of weighted majorities~\cite{haggstrom2006law}.}

\subsection{Core Lemma}
\label{sec:corelemma}

Next, we give a key lemma, which provides
sufficient conditions for a model $M$ to satisfy probabilistic do
no harm and probabilistic  positive gain with respect to a class
$\mathfrak{D}$ of distributions.  This lemma will form the basis of
all of our later results. Since the result follows from relatively straightforward concentration inequalities, we defer the proof to Appendix~\ref{app:core-proof}.
\begin{lemma}
\label{lem:core}

If $M$ is a  model, $\mathfrak{D}$ a class of distributions, $n$ a number of persons, and
for all distributions $\D \in \mathfrak{D}$, there is an $\a \in (0,
1)$ and $C: \N \to \N$ with $C(n) \in o(n)$ such that 
    \begin{align}
   	 &\P_{\D, M,n}\left[\maxweight(G_n) \le C(n)\right] = 1 -
      o(1)\label{constr:1} \\ 
     &\P_{\D, M,n}\left[\sum_{i = 1}^n \weight_i(G_n) \cdot p_i - \sum_{i = 1}^n p_i \ge 2 \a n \right] = 1 - o(1) \label{constr:2},
    \end{align}
    then $M$ satisfies probabilistic do no harm.
    If in addition, there exists a distribution $\D \in \mathfrak{D}$ and an $\a \in (0, 1)$ such that
    \begin{equation}
    \label{constr:3}
    \P_{\D, M,n} \left[\sum_{i = 1}^n  p_i + \a n \le n/2 \le \sum_{i = 1}^n \weight_i(G_n) \cdot p_i -  \a  n\right] = 1- o(1),
\end{equation}
then $M$ satisfies probabilistic positive gain.
\end{lemma}

In words, condition~\eqref{constr:1} ensures that, as the number of voters grows large, the weighted number of correct votes under liquid democracy will concentrate around its expectation, $\sum_{i = 1}^n \weight_i(G_n) \cdot p_i$. Standard concentration results already imply this holds for direct democracy. Condition~\eqref{constr:2} ensures that these expectations are sufficiently separated. So with high probability, liquid democracy will have more correct votes than direct democracy, which is sufficient to guarantee DNH. Finally, Condition~\eqref{constr:3} ensures that in some cases, the expectations for direct and liquid votes will be below and 
over
half the voters, respectively, which after applying concentration means there will likely be a large gain.

In the following sections, we investigate natural delegation
models and identify conditions such that the models satisfy
probabilistic do no harm and probabilistic positive gain. In all instances, we will invoke \Cref{lem:core} after showing that its sufficient
conditions are satisfied. 

\section{Strictly Upward Delegation Model}\label{sec:up}

We now turn to the analysis of a simple model that assumes that
voters either do not delegate with fixed exogenous probability or delegate to voters that have a competence greater than
their own. 

Formally, for a fixed $p \in [0,1]$ we let $M^{U}_{p}=(q, \p)$ be the
model consisting of $q(p_i)=p$ for all $p_i \in [0, 1]$, and
$\p(p_i, p_j) = \mathbb{I}_{\{p_j > p_i\}}$ for all $i, j  
\in [n]$. That is, voter $i$ delegates with fixed probability $p$ and
puts equal weight on all the more competent voters. In other words, if voter $i$ delegates, then $i$ does so to a more
competent voter chosen uniformly at random. Note that a voter with maximal
competence will place $0$ weight on all other voters, and hence is
guaranteed not to delegate. We refer to $M_p^U$ as the \emph{Upward Delegation Model} parameterized by $p$.  

\begin{theorem}[Upward Delegation Model]
For all $p \in (0,1)$, $M^U_{p}$ satisfies probabilistic do no harm and probabilistic positive gain with respect to the class $\mathfrak{D}^C$ of all continuous distributions.
\label{thm:UDM}
\end{theorem}

The proof of the theorem relies on novel bounds we drive on the largest bin size in an infinite 
P\'olya urn 
process~\citep{simon1955class,chung2003generalizations}. We first formally define the process and present our bound in \Cref{lem:polya}. A 
\emph{P\'olya urn process with attachment probability 
$p$} begins at time $t = 1$ with one ball in one bin. At each timestep $t > 1$, a new ball arrives. With probability $1 - p$, a new bin is created and the new ball is placed in that bin; with probability $p$, the ball joins an existing bin, and it does so with probability proportional to the number of balls in the bins, i.e., if there are three bins containing $1$, $2$, and $3$ balls respectively, it joins each with probability $1/6, 2/6,$ and $3/6$ respectively. We then have the following.

\begin{lemma}
    \label{lem:polya}
    For all $p \in (0, 1)$ and $t \ge 1$, let $L^p_t$ be the random variable denoting the maximum number of balls in any bin after running the infinite 
P\'olya
urn process with new-bin probability $p$ for $t$ steps. Then, there exists $\delta < 1$ depending only on $p$ such that for all $T \ge 1$, $\Pr[L^p_T \le T^\delta] = 1 - o(1).$
\end{lemma}
\proof{Proof}
Fix the parameter $p \in (0, 1)$. Choose $\g$ to be a constant such that $3/4 < \g < 1$; note that $p + (1 - p)\g < p + (1 - p) = 1$. Choose $\d$ (for the lemma statement) such that $p + (1 - p)\g < \d < 1$. Notice that we can choose $\g$ and $\d$ such that $\d$ is arbitrarily close to $3/4 + p/4$.

Let $B^{(k)}$ denote the $k$-th bin.  Let $U_t^{(k)}$ be the size of $B^{(k)}$ at time
$t$. Since there are at most $t$ bins by time $t$, notice that $L^p_t = \max(U_t^{(1)}, \ldots, U_t^{(t)})$. In general, our approach will be to analyze bins separately and show that
$U^{(k)}_T$ remains below $T^\delta$ with high
enough probability so that we can union bound over all possible $k \le
T$. That is, we will show 
$$ \sum_{k = 1}^T
    \Pr[U_T^{(k)} > T^\delta] = o(1),
$$
which also implies $\Pr[L^p_T > T^\delta] = o(1)$.
Hence, it will be useful to consider this process more formally from the perspective of the $k$th bin, $B^{(k)}$. The $k$th
bin $B^{(k)}$ is ``born'' at some time $t \ge k$, the $k$th time in which a ball does not join a pre-existing bin, at which point $U_t^{(k)} = 1$ (prior
to this, $U_t^{(k)} = 0$). More specifically, the first bin $B^{(k)}$ is
guaranteed to be born at time $t = 1$ and for all other $k > 1$, $B^{(k)}$
will be born at time $t \ge k$ with probability $\binom{t - 1}{k -
  1}(1 - p)^{k}p^{t - k}$, although these exact probabilities will be
unimportant for our analysis. Once born, we have the following
recurrence on $U_t^{(k)}$ describing the probability $B^{(k)}$ will be
chosen at time $t$:
$$
U_t^{(k)} = 
\begin{cases}
    U_{t-1}^{(k)} + 1 & \text{with probability } \frac{p \cdot U_{t-1}^{(k)}}{t-1} \\
        U_{t-1}^{(k)} & \text{with probability } 1- \frac{p \cdot U_{t-1}^{(k)}}{t-1}.
\end{cases}
$$
Let $W_t^{(k)}$ be the process for the size of bin that is born at time $k$. That is, $W_k^{(k)} = 1$, and for $k > t$,
$W_t^{(k)}$ follows the exact same recurrence as
$U_t^{(k)}$. Note that since the $k$th bin $B^{(k)}$ can only be born at time $k$ or later, we have that $W_t^{(k)}$
stochastically dominates $U_t^{(k)}$ for all $k$ and $t$. Hence, it suffices to show that
\begin{equation}
    \label{eq:w-union}
    \sum_{k = 1}^T \Pr[W_T^{(k)} > T^\delta] = o(1).
\end{equation}

We split our analysis into two parts: the first consider the first
$T^\g$ bins, while the second considers the last $T-T^\g$ bins.

We first show that 
$\sum_{k = 1}^{T^\g} \P[W^{(k)}_T > T^\d] = o(1)$. Note that the expectation of $W_n^{(k)}$
\begin{equation}
    \E[W_n^{(k)}] = \frac{\Gamma(n+p)\Gamma(k)}{\Gamma(p+k)\Gamma(n)}
\label{eq:UDM-exp} 
\end{equation}
for all $k \le n$, where $\Gamma$ represents the Gamma function. We relegate the argument for \Cref{eq:UDM-exp} to Appendix~\ref{app:lemma_proof}. Using this along with Gautchi's inequality~\citep{gautschi1959some},
    $(t+p-1)^p \leq \frac{\Gamma(p+t)}{\Gamma(t)} \leq(t+p)^p$,
to approximate the $\Gamma$ terms, we can apply Markov's inequality and use algebra to get $\sum_{k = 1}^{n^\g} \P[W^k_n > n^\d] = o(1)$. We again relegate these arguments to Appendix~\ref{app:lemma_proof}. 

Now consider the final $T-T^\g$ components. We  will
prove that
$\Pr[W_T^{(T^\g + 1)} > T^\d] = o(1/T)$. Since $W_T^{(k)}$
stochastically dominates $W_T^{(k')}$ for all $k' \ge k$, this implies
that $\Pr[W_T^{(k)} > T^\d] = o(1/T)$ for all $k \ge T^\g + 1$. Hence, 
$$\sum_{k = T^\g + 1}^T \Pr\left[W_T^{(k)} > T^\d \right]= o(1).$$

To do this, we compare the $W_t^{(T^\g + 1)}$ process to another
process, $V_t$.
We define $V_0=1$, and for $t > 0$, take $V_t$ to satisfy the following
recurrence: 
$$
V_t = 
\begin{cases}
    V_{t-1} + 1 & \text{with probability } \frac{ V_{t-1}}{t + n^\g} \\
    V_{t-1} & \text{with probability } 1- \frac{ V_{t-1}}{t + n^\g}.
\end{cases}
$$
This is identical to the $W$ recurrence with $t$ shifted down by $n^{\gamma} + 1$ except without the $p$ factor. Hence, $V_{T - T^\g + 1}$ clearly stochastically dominates $W_T^{(T^\g + 1)}$. For convenience in calculation, we will instead focus on bounding $V_T$ which itself stochastically dominates $V_{T - T^\g + 1}$.

Next, note that the $V_t$ process is isomorphic to the
following classic P\'olya urn process. We begin with
two bins,
one with a single ball and the other with $n^\g$ balls. At each time,
a new ball is added to one of the two bins with probability
proportional to the bin size. The process $V_t$ is isomorphic to the
size of the one-ball urn after $t$ steps. Classic results tell us
that for fixed starting bin sizes $a$ and $b$, as the number of steps
grows large, the possible proportion of balls in the $a$-bin follows a
$\betadist(a, b)$ distribution~\citep{markoff1917quelques,
  eggenberger1923statistik, Polya1930quelques, munford1978urn,
  mahmoud2009Polya}. 

The mean and variance of such a Beta distribution would be sufficient to prove our necessary concentration bounds; however, for us, we need results after exactly $T - T^\g$ steps, not simply in the limit. Hence, we will be additionally concerned with the speed of convergence to this Beta distribution.

Let $X_T = \frac{V_T}{T}$ and $Z_T \sim \betadist(1, T^\g)$.
From \citet{janson2020rate}, we know that the rate of convergence is such that, for any $p \ge 1$ 
\begin{equation}
    \ell_p(X_T, Z_T) = \Theta(1/T)
    \label{eq:svante}
\end{equation}
where $\ell_p$ is the \emph{minimal $L_p$ metric}, defined as
$$\ell_p(X, Y) = \inf \set{\E[| X' - Y'|^p]^{1/p} \suchthat X'
  \overset{d}{=} X, Y' \overset{d}{=} Y},$$ which can be thought of as
the minimal $L_p$ norm over all possible couplings between $X$ and $Y$. For our purposes, the only fact about the $\ell_p$ metric we
will need is that $\ell_p(X, 0) = \E[|X|^p]^{1/p}$ where $0$ is the
identically $0$ random variable. Since $\ell_p$ is in fact a metric,
the triangle inequality tells us that $\ell_p(0, X_n) \le \ell_p(0, Z_n)
+ \ell_p(Z_n, X_n)$, so, combining with \eqref{eq:svante}, we have that 
\begin{equation}
    \E[|X_T|^p]^{1/p} \le \E[|Z_T|^p]^{1/p} + \Theta(1/T)
    \label{eq:p-distance}
\end{equation}
for all $p \ge 1$. 

Note that since $Z_T \sim \betadist(1, T^\g)$,
$$\E[Z_T] = \frac{1}{1 + T^\g} =\Theta(T^{-\g})$$
and
$$\Var[Z_T] = \frac{T^\g}{(2 + T^\g)(1 + T^\g)^2} = \Theta(T^{-2\g}).$$
Given these results, we are ready to prove that $V_T$ is smaller than $T^\d$ with probability $1-o(1/T).$ Precisely, we want to show that
$\Pr[X_T \geq T^{\d-1}] = o(1).$
By Chebyshev's inequality,
$$\Pr[X_T \geq T^{\d - 1}] \leq \frac{\Var[X_T]}{(T^{\d - 1} - \E[X_T])^2}.$$

Inequality~\eqref{eq:p-distance} with $p=1$ along with the fact that $X_T$ and $Z_T$
are always nonnegative implies that $\E[X_T] \leq \E[Z_T] + \Theta(1/T) =
O(T^{-\g})$. Hence, $T^{\d - 1} - \E[X_T] = \Omega(T^{\d - 1})$ since
$\d - 1 > -1/2 > -\g$. We can therefore write: 
\begin{equation}
    \left(T^{\d - 1} - \E[X_T]\right)^2 = \Omega(T^{-2(\d - 1)}).
    \label{eq:expbound}
\end{equation}

Inequality~\eqref{eq:p-distance} with $p=2$ implies that $\sqrt{\E[X_T^2]} \leq
\sqrt{\E[Z_T^2]} + \Theta(1/T).$ Hence, 
\begin{align*}
\E[X_T^2]
&\leq (\Theta(1/T) + \sqrt{\E[Z_T^2]})^2\\
&\le (\Theta(1/T) + \sqrt{\E[Z_T]^2 + \Var[Z_T]})^2\\
&\le (\Theta(1/T) + \sqrt{\Theta(T^{-2\g})})^2\\
&= (\Theta(1/T) + \Theta(T^{-\g}))^2\\
&= \Theta(T^{-\g})^2\\
&=  \Theta(T^{-2\g}).
\end{align*}

Next, note that $\Var[X_T] \leq \E[X_T^2],$ so
\begin{equation}
    \Var[X_T] = O(T^{-2\g})
    \label{eq:varbound}
\end{equation}
as well.
Combining \eqref{eq:expbound} and \eqref{eq:varbound}, we have that
\begin{equation*}
    \Pr[X_T \geq T^{\d - 1}]
    \leq \frac{\Var[X_T]}{(T^{\d - 1} - \E[X_T])^2} = O\left(T^{-2\g + 2(1-\d)}\right).
\end{equation*}

Since $-2\g + 2(1 - \d) < 1$, given our assumption that
$3/4 < \g < \d$, it follows that
$\Pr[X_T \geq T^{\d - 1}] = o(1/T)$, which allows us
to conclude that 
$$\sum_{k = T^\g + 1}^T \Pr[W_T^{(k)} > T^\d] = o(1).$$
 
Since we showed earlier that $\sum_{k = 1}^{T^\g} \Pr[W_T^{(k)} >
T^\d] = o(1)$, we have that 
$$
    \sum_{k = 1}^{T} \Pr[W_T^{(k)} > T^\d] = o(1),
$$
as needed.
\Halmos\endproof

We are now ready to prove the theorem about Upward Delegation.

\proof{Proof of \Cref{thm:UDM}}
  To prove the theorem, we will prove that the Upward Delegation
 Model with respect to $\mathfrak{D}^C$
  satisfies \eqref{constr:1}, \eqref{constr:2}, \eqref{constr:3}, which implies that
  the model 
    satisfies probabilistic do no harm and positive gain by \Cref{lem:core}. We show \eqref{constr:1} here and relegate \eqref{constr:2} and \eqref{constr:3} to Appendix~\ref{app:proof_udm}
\medskip

\emph{Upward Delegation satisfies \eqref{constr:1}}
\smallskip

To do this, we will simply show that the component sizes in $G_n$ sampled according to $\P_{D, M, n}$  have the 
same distribution as the bin sizes in a P\'olya  urn 
process with attachment probability $p$, and hence $\maxweight(G_n)$ follows the same distribution as $L^p_n$. Once we have shown this, \eqref{constr:1} follows immediately from \Cref{lem:polya} as $n^\delta \in o(n)$.

To that end, fix some sampled competencies $\vec{p}_n$. Recall that each entry
$p_i$ in $\vec{p}_n$ is sampled i.i.d.\@ from $\D$, a continuous distribution. Hence, almost surely, no two competencies are equal.  From
now on, we condition on this probability $1$ event.  Now
consider sampling the delegation graph $G_n$. By the design of the
model $M^U_p$, we can consider a random process for generating $G_n$ that is isomorphic to sampling according to $\P_{D,M,n}$
as follows: first, order  the
competencies $p_{(1)} > p_{(2)} > \cdots > p_{(n)}$ (note that such
strict order is possible by our assumption that all competencies are different)
and rename the voters such that voter $i$ has competence
$p_{(i)}$; then construct $G_n$ iteratively by adding the voters one
at a time in decreasing order of competencies, voter $1$ at time
$1$, voter $2$ at time $2$, and so on.  

We start with the voter with the highest competence, voter $1$.  By the choice of $\p$, voter $1$ places weight $0$ on every other voter and hence by definition does not delegate. This voter forms the first component in the graph $G_n$, which we call $C^{(1)}$. Then, we add voter $2$ who either delegates to voter $1$ joining component $C^{(1)}$ with probability $p$, or starts a new component $C^{(2)}$ with probability $1-p$. Next, we add voter $3$. If $2 \in C^{(1)}$ (that is, if $2$ delegated to $1$), $3$ either delegates to $1$ (either directly or through $2$ by transitivity) with probability $p$ or she starts a new component $C^{(2)}$. If $2 \in C^{(2)}$, then $3$ either delegates to $1$ with probability $p/2$ and is added to $C^{(1)}$, or delegates to $2$ with probability $p/2$ and is added to $C^{(2)}$, or starts a new component $C^{(3)}$. In general, at time $t$, if there are $k$ existing components $C^{(1)}, \ldots, C^{(k)}$, voter $t$ either joins each component $C(j)$ with probability $\frac{p |C(j)|}{t-1}$ or starts a new component with probability $1-p$. To construct $G_n$, we run this process for $n$ steps. Notice that this is identical to the 
P\'olya 
urn process with bins $B^{(k)}$ and balls replaced with components $C^{(k)}$ and voters being run for $n$ steps, as needed.
\Halmos\endproof

\section{Confidence-Based Delegation Model}\label{sec:med}

We now explore a model according to which voters delegate with probability that is strictly decreasing (or, monotonically decreasing, that is 
$x < y$ implies
$f(x) > f(y)$)
in their competence and when they do decide to delegate, they do so by picking a voter uniformly at random. This models the case where voters do not need to know anything about their peers'
competencies, but do have some sense of their own competence, and 
delegate accordingly. 

Formally, for any $q$, let $M^C_q=(q, \p^1)$ where $\p^1(p_i, p_j) = 1$ for all $i, j \in [n]$. Voter $i$ puts equal weight on all the voters and hence samples one uniformly at random when they delegate. We refer to $M^C_q$ as the Confidence-Based Delegation Model.

\begin{theorem}[Confidence-Based Delegation Model]
All models $M^C_q$ with monotonically decreasing $q$ satisfy probabilistic do no harm and probabilistic positive gain with respect to the class $\mathfrak{D}^C$ of all continuous distributions.
\label{thm:CBM}
\end{theorem}

\proof{Proof}
    We show that the Confidence-Based Model satisfy
\eqref{constr:1} and \eqref{constr:2} here, and relegate showing \eqref{constr:3} to Appendix~\ref{app:proof_cdm}.
\medskip

\emph{Confidence-Based Delegation satisfies \eqref{constr:1}}
\smallskip

Fix some distribution $\D \in \mathfrak{D}^C$. We show there exists
$C(n) \in O(\log n)$ such that \eqref{constr:1} holds. 
    
    Note that when sampling an instance $(\vec{p}_n, G_n)$, the
    probability an arbitrary voter $i$ chooses to delegate is
    precisely $p := \E_{\D}[q]$. To see this, consider how a voter $i$ chooses whether to delegate: they first sample a competence $p_i \sim \D$ and then sample whether or not to delegate from $\bern(q(p_i))$. Treating this as a single process, it is clear that the overall probability of choosing to delegate is exactly $\E_{\D}[q]$ by integrating out the competence.
    
    Further, since $\D$ is
    continuous and $q$ is monotonically decreasing, $p \in (0,
    1)$. When a voter does decide to delegate, they do so by picking
    another voter uniformly at random. Hence, we can consider the
    marginal distribution of delegation graphs directly (ignoring the
    competencies). We will show that when sampling a delegation graph,
    for any specific voter $i$, with probability $1 - o(1/n)$,
    $\dels_i(G_n) \le C(n)$, which implies $\weight_i(G_n) \le
    C(n)$. A union bound over all $n$ voters implies $\maxweight(G_n)
    \le C(n)$ with probability $1 - o(1)$. 
    
    To that end, we will describe a branching process similar to the
    well-known \emph{graph branching
    process}~\citep{alon2016probabilistic}, which has the property
    that the distribution of its size exactly matches the distribution
    of $\dels_i(G_n)$ for an arbitrary voter $i$. We will compare this
    process to a known graph branching process that has size at most
    $O(\log n)$ with high probability. We will show our process is
    sufficiently dominated such that it too has size at most $O(\log
    n)$ with high probability. The branching process works as
    follows. Fix our voter $i$. We sample which other
    voters end up in $i$'s ``delegation tree'' (i.e., its ancestors in
    $G_n$) dynamically over a sequence of time steps. As is standard
    for these processes, all voters $V$ will be one of three types, live, dead, or neutral. Dead voters are those whose
        ``children'' (i.e., voters who delegate to them) we have already
    sampled.  Live
    voters are voters who have decided to delegate, but whose children have
    not yet been sampled. 
   Neutral voters are still in the ``pool''
    and have yet to commit to a delegation. At time zero, $i$ is a
    live voter, there are no dead voters, and all other voters $V
    \setminus \set{i}$ are neutral. At each time step, we take some
    live voter $j$, sample which of the neutral voters choose to
    delegate to $j$, add these voters as live vertices, and update $j$ as dead. The procedure ends when there are no more live vertices,
    at which point the number of delegations received by $i$ is simply
    the total number of dead vertices. 
    
    Let us now describe this more formally. Following the notation of
    \citet{alon2016probabilistic}, let $Z_t$ denote the number of
    voters we sample to delegate at time $t$. Let $Y_t$ be the number
    of live vertices at time $t$; we have that $Y_0 = 1$. At time $t$,
    we remove one live vertex and add $Z_t$ more, so we have the recursion $Y_t = Y_{t - 1} - 1 + Z_t$. We let $N_t$ be the number
    of neutral vertices at time $t$. We have that $N_0 = n - 1$,
    and $N_t = N_{t - 1} - Z_t$. Note that after $t$ time steps, there
    are $t$ dead vertices and $Y_t$ live ones, so this is equivalent
to $N_t = n - 1 - t - Y_t$. To sample $Z_t$, we fix some live voter $j$ and ask how many of the neutral voters chose to delegate
    to $j$, conditioned on them not delegating to any of the dead
    voters. Note that when sampling at this step, there are $t - 1$
    dead voters and conditioned on the neutral voters not delegating
    to the dead ones, the probability they delegate to any of the
    other $n - t$ individuals (not including themselves) is exactly
    $\frac{p}{n - t}$, equally split between them for a total
    delegation probability of $p$. Hence $Z_t \sim \text{Bin}(N_{t-1},
    \frac{p}{n - t}) \sim \text{Bin}(n - t - Y_{t-1}, \frac{p}{n -
      t})$. We denote by $\mathfrak{X}^D_{n, p}$ the random variable that counts the size of this branching process, i.e., the
    number of time steps until there are no more live vertices. Note
    that the number of delegations received by any voter has the same
    distribution as $\mathfrak{X}^D_{n, p}$. 
    
    Choose some constant $p'$ such that $p < p' < 1$. We will be
    comparing the $\mathfrak{X}^D_{n, p}$ to a graph branching process
    $\mathfrak{X}^G_{n, p'}$. The graph branching process is nearly
    identical, except the probability each of the neutral vertex joins
    our component is independent of the number of dead vertices and is
    simply $\frac{p'}{n}$. In other words, $Z_t \sim \bin(N_{t-1},
    \frac{p'}{n})$. A key result about this branching process is the
    probability of seeing a component of a certain size $\ell$ decreases exponentially with $\ell$. In other words, there is some
    constant $c$ such that 
    $$
        \P_{\D,M_q^C,n}[\mathfrak{X}^G_{n, p'} \le c \log(n)] = 1 - o(1/n).
        $$

   Take $C(n) = c \log(n)$. Note that as long as $t$ is
        such that $\frac{p}{n - t} \le \frac{p'}{n}$, the sampling in the
    delegation branching process is dominated by the sampling in this graph branching process. Hence, as long as  $\frac{p}{n - C(n)} \le
    \frac{p'}{n}$, $\P[\mathfrak{X}^D_{n, p} \le c \log(n)] \ge
    \P[\mathfrak{X}^G_{n, p'} \le c \log(n)]$. Since $C(n) \in O(\log
    n)$, this is true for sufficiently large $n$, so for such $n$,
    $\P[\mathfrak{X}^D_{n, p} \le c \log(n)] = 1 - o(1/n)$. By a union
    bound over all $n$ voters, this implies the desired result. 
\medskip  

\emph{Confidence Based Delegation satisfies \eqref{constr:2}}
\smallskip

Let $\bar{q}$ be such that $\bar{q}(x) = 1 - q(x)$, so $\bar{q}$
represents the probability someone with competence $x$ does not delegate. Notice that $\E_{\D}[\bar{q}]$ is exactly the probability an arbitrary voter will not delegate. Let $q^+(x) = \bar{q}(x) x$ and let  $$\mu^* = \frac{\E_{\D}[q^+]}{\E_{\D}[\bar{q}]}.$$ Expanding the definition, we see that $\mu^*$ is exactly the expected value of a voter's competence, conditioned on them not voting. 
Let $\mu_\D$ the mean of
the competence distribution $\D$. We first show that 
$\mu^* > \mu_\D$.
Indeed, since both $x$ and $\bar{q}(x)$ are strictly increasing
functions of $x$, the Fortuin–Kasteleyn–Ginibre (FKG)
inequality~\citep{fortuin1971correlation} tells us that
$\E_{\D}[q^+] > \E_{\D}[\bar{q}] \cdot \E_{\D}[x] = \E_{\D}[\bar{q}]
\cdot \mu_{\D}.$ 
This implies that the expected competence conditioned on not delegating is strictly higher than the overall expected competence.

Next, we will show 
that for any constant $\g > 0$, with high probability, both $\sum_{i =
  1}^n p_i \le (\mu + \g) n$ and $\sum_{i = 1}^n \weight_i(G)p_i \ge
(\mu^* - \g) n$. If we choose $\g = (\mu^* - \mu)/3$ and $\a = \g /2$, it follows that, with high probability,  
$$\sum_{i = 1}^n \weight_i(G)p_i -  \sum_{i = 1}^n p_i \ge 2 \a n,$$
implying that \eqref{constr:2} is satisfied. 

Since the $p_i$s are bounded independent variables, it follows
directly from Heoffding's inequality that $\sum_{i = 1}^n p_i \le n(\mu
+ \g)$ with high probability, so we now focus on showing $\sum_{i =
  1}^n \weight_i(G) \cdot p_i \ge (\mu^* - \g)n$ with high
probability. To do this, we will first show that, with high probability, the
delegation graph $G$ satisfies $\dels_i(G) \le C(n)$ for all $i$ and
$\totweight(G) \ge n - C(n)\log^2 n$. 

We showed in the earlier part of this proof that $\dels_i(G) \le
C(n)$ with high probability. We will now prove that $\P_{\D,M_q^C,n}[\totweight(G)
  \geq n - C(n)\log^2 n\suchthat\dels_i(G) \le C(n)] = 1- o(1).$ To do
this, we will first bound the number of voters that, with high
probability, end up in cycles. Fix a voter $i$ and sample $i$'s delegation tree. Voter $i$
will only end up in a cycle if $i$ chooses to delegate to someone in
this delegation tree. Since we are conditioning on $\dels_i(G)
\le C(n)$, the maximum size of this tree is $C(n)$. Hence, the total
$\p$ weight that voter $i$ places on someone in the tree is at most
$C(n)$, while the total weight they place on all voters is $n-1.$
Hence, the probability that $i$ delegates to someone in their tree can
be at most $p \cdot C(n)/(n-1)$. Since this is true for each voter
$i$, the expected number of voters in cycles is at most
$np\frac{C(n)}{(n - 1)} \in O(\log n).$ By Markov's inequality, the
probability that more than $\log^2n$ voters are in cycles is at most
$np\frac{C(n)}{(n - 1)\log^2 n} = O(1/\log n) = o(1).$ 

Next, since we have conditioned on $\dels_i(G) \le C(n),$ no single
voter, and in particular no single voter in a cycle, can receive more
than $C(n)$ delegations. So conditioned on the high probability event
that there are at most $\log^2 n$ voters in cycles, there are at most
$C(n)\log^2n$ voters that delegate to those in cycles. This implies that $\totweight(G) \ge n - C(n)\log^2n + \log^2n$ with high probability.

We now show that, conditioned on the graph satisfying these properties,
the instance $(\vec{p}, G)$ satisfies $\sum_{i = 1}^n \weight_i(G)
\cdot p_i \ge n(\mu^* - \g)$ with high probability. Note
that the competencies satisfy that those that don't
delegate are drawn i.i.d.\@ from the distribution of competencies
conditioned on not delegating, which has mean $\mu^*$. Fix an
arbitrary graph $G$ satisfying the properties. Suppose $M$ is the set
of voters that do not delegate. Note that for each $i \in M$,
$\weight_i(G) \le C(n)$, by assumption. Further $\sum_{i \in M}
\weight_i(G) \ge n - C(n) \log^2(n)$. Hence, when we sample the
non-delegator $p_i$s, $\E[\sum_{i \in M} \weight_i(G) \cdot p_i] \ge
(n - C(n)\log^2(n)) \cdot \mu^*$. Moreover, $$\Var[\sum_{i \in M}
  \weight_i(G) \cdot p_i] \le \sum_{i \in M} \weight_i(G)^2 \le C(n)
\cdot n.$$ This follows from the fact that $\Var[p_i]\leq 1$ and that we have fixed the graph $G$ and hence $\weight_i(G)$ for each $i$, so
these terms can all be viewed as constants.  In addition, we know
that, for each voter $i$, 
$\weight_i(G) \le C(n)$, and $\sum_{i = 1}^n \weight_i(G) \le
n$. Hence, we can directly apply Chebyshev's inequality: 
\begin{align*}
    \P_{\D,M_q^C,n}\left[\sum_{i \in M} \weight_i(G)p_i < n(\mu^* - \g)\right] 
    &< \frac{\Var[\sum_{i \in M} \weight_i(G)p_i]}{ (\E[\sum_{i \in M} \weight_i(G)p_i] - n(\mu^* - \g))^2} \\
    &\le \frac{nC(n)}{(\g n - C(n)\log^2(n) \mu^*)^2} \\
    &\in o(1),
\end{align*}
where the final step holds because the numerator is $o(n^2)$ and the denominator is $\Omega(n^2)$. Hence, $\sum_{i \in M} \weight_i(G)p_i \ge n(\mu^* - \g)$ with high probability, as needed.

To summarize, we have proved that,
conditioned on $\dels_i(G) \le C(n)$ for all $i$ and $\totweight(G)
\ge n - C(n)\log^2 n, \sum_{i = 1}^n \weight_i(G) \cdot p_i \ge
n(\mu^* - \g / 3)$ occurs with high probability. Given that,
conditioned on $\dels_i(G) \le C(n), \totweight(G) \ge n - C(n)\log^2
n$ occurs with high probability and that $\dels_i(G) \le C(n)$ occurs
with high probability, we can conclude by the chain rule that the
intersection of these events hold with high probability. Given that
the probability of any of this event is greater than the probability
of the intersection, we can conclude that $\sum_{i = 1}^n \weight_i(G)
\cdot p_i \ge n(\mu^* - \g / 3)$ occurs with probability $1-o(1),$ as
desired. 
\Halmos\endproof

\section{Continuous General Delegation Model}
\label{sec:general}
Finally, we study a model in which voters delegate with fixed probability, and they
do so by picking a voter according to a continuous increasing delegation
function. This is a general model in which delegations can either go to more or less competent neighbors but where more competent voters are more likely to be chosen over less competent ones.  

Formally, let $M^S_{p, \p} =(q^p, \p)$ where $q^p$ is a constant function equal to $p$, that is, $q^p(x) = p$ for all $x \in [0, 1]$, and $\p(x,y)$ is non-zero, continuous, and increasing in $y$. We then have the following.

\begin{theorem}[Continuous General Delegation Model]
All models $M^S_{p, \p}$ with $p \in (0, 1)$ and $\p$ that is non-zero, continuous, and increasing in its second coordinate satisfy probabilistic do no harm and probabilistic positive gain with respect to the class $\mathfrak{D}^C$ of all continuous distributions.
\label{thm:SPM}
\end{theorem}
\proof{Proof}
A majority of the proof is related to Appendix~\ref{app:proof_spm}. In the following, we show the beginning of the proof, which describes the setup for proving \eqref{constr:1}.

Fix $M^S_{p, \p}$ and and $\D \in \mathfrak{D}^C$. Note that since $\p$ is continuous and always positive on the compact set $[0, 1]^2$,
$\p$ is in fact uniformly continuous and there are bounds
$L, U \in \R^+$ such that $\p$ is bounded in the interval $[L,
  U]$. Additionally, we can assume without loss of generality that for all $x \in [0, 1]$, $\E_{\D}[\p(x, \cdot)] = 1$. Indeed, $\E_{\D}
[\p(x, \cdot)]$ is a positive, continuous function of $x$, so
replacing $\p$ by $\p'(x, y) = \p(x, y)/\E_{\D}[\p(x, \cdot)]$
induces the same model and satisfies the desired property. 

\medskip

\emph{The Continuous General Delegation Model satisfies \eqref{constr:1}.}
\smallskip

Our goal is to show there is some $C(n) \in O(\log n)$ such that, with
high probability, no voter receives more than $C(n)$ delegations. To
do this, just as in the proof of \Cref{thm:CBM}, we consider a branching process of the delegations received beginning with some voter
$i$. We will show that under minimal conditions on the sampled competencies (which all occur with high probability), this
branching process will be dominated by a well-known \emph{subcritical
multi-type Poisson branching process}~\citep{bollobas2007phase}, which
has size $O(\log n)$ with high probability. 

For a fixed competence vector $\vec{p}_n$, the branching process for
the number of delegations received by a voter $i$ works as follows. We keep track of three sets of voters: those that are live at time $t$
($L_t$), those dead at time $t$ ($D_t$), and those 
neutral at time $t$ ($N_t$). Unlike in the proof of \Cref{thm:CBM},
where it was 
sufficient to keep track of the number of voters in each category,
here we must keep track of the voter identities as well, as they do not all delegate with the same probability. At time zero, the only live
voter is voter 
$i$ and the rest are neutral, so $L_0 = \set{i}$, $D_0 = \emptyset$, and $N_0 = [n] \setminus \set{i}$. As long as there are still live voters,
we sample the next set of delegating voters $Z_t$ in time $t$ by
choosing some live voter $j \in R_{t - 1}$ and sampling its
children. Once $j$'s children are sampled, $j$ becomes dead, and $j$'s
children become live. All voters that did not delegate and were not delegated to remain neutral. The children are sampled
independently; the probability they are included is the probability they delegate to $j$ conditioned on them not delegating to the dead
voters in $D_{t -1}$. For each voter $k \in N_{t - 1}$, $k$ will be
included with probability 
$$
    p \cdot \frac{\p(p_k, p_j)}{\sum_{k' \in [n] \setminus (D_{t - 1} \cup \set{k})} \p(p_k, p_{k'})}.
$$
This is precisely the probability $k$ delegates to $j$ conditioned on them not delegating to any voter in $D_{t - 1}$. We continue this process until there are no more live voters, at which point the number of delegations is simply the number of dead voters, or equivalently, the total number of time steps. We denote by $\mathfrak{X}^D_{\vec{p}_n, i}$ the size of the branching process parameterized by competencies $\vec{p}_n$ and a voter $i \in [n]$. 

Our goal will be to compare $\mathfrak{X}^D_{\vec{p}_n, i}$ to the
outcome of a well-known multi-type Poisson
branching process. In  this branching process, there are a fixed finite
number $k$ of types of voters.\footnote{In the literature, these are often called \emph{particles}, but to be consistent with our other branching
processes, we call them voters here.} The process itself is
parameterized by a $k \times k$ matrix $M$, where $M_{\tau \tau'}$ is the expected number of children of type $\tau'$ a voter of type $\tau$
will have. The process is additionally parameterized by the type $\tau \in [k]$ of the starting voter. The random variable $Y_t$ keeps
track of the number of live voters of each type; it is 
a vector of length $k$, where the $\tau$th entry is the number of live voters of type $\tau$. Hence,
$Y_0 = e_\tau$, the (basis) vector with a 1 in entry $\tau$ and an
entry 0 for all other types.  
We sample children by taking
an arbitrary live voter of type $\tau'$ (the $\tau'$ component in
$Y_{t - 1}$ must be positive, indicating that there is such a voter), and
sampling its children $Z_t$, which is also a vector of length $k$,
each entry indicating the number of children of that type. The vector
$Z_t$ is sampled such that the $\tau''$ entry is from the
$\pois(M_{\tau'\tau''})$ distribution. That is, children of different
types are sampled independently from a Poisson distribution, with the
given expected value. We have the recursion $Y_t = Y_{t - 1} + Z_t -
e_{\tau'}$. 

Note that this means that there is no ``pool'' of voters to choose from;
in fact, it is possible for this process to grow unboundedly large (see~\citep[Section 11.6]{alon2016probabilistic} for the
classical description of the single-type Poisson branching
process). Nonetheless, this process will still converge often enough
to remain useful. 
We denote by $\mathfrak{X}^P_{M, \tau}$ the random variable that gives the size of this branching process, parameterized by
expected-children matrix $M$ and starting voter type $\tau \in [k]$. 
Such a branching process is considered \emph{sub-critical} if the
largest eigenvalue of $M$ is strictly less than
$1$~\citep{bollobas2007phase}. In such a case, if we begin with voter
of any type $\tau \in [k]$, the probability of the branching process surviving $\ell$ steps decreases exponentially in $\ell$. Hence, there
is some $c$ such that for all $\tau \in [k]$, 
$$\P[\mathfrak{X}^P_{M, \tau} \le c \log (n)] = 1 - o(1/n).$$

To compare these branching processes, we  make a sequence
of adjustments to the original branching process that at each step
creates a dominating branching process slightly closer in flavor to
the multi-type Poisson. In the end, we will be left with a
sub-critical multi-type Poisson process that we can bound. 

Fix some $\e > 0$, which is a parameter in all of our steps. Later, we will choose $\e$ to be sufficiently small (specifically,
such that $p\frac{(1+\e)^3}{1-2\e}<1$) to ensure that the Poisson
branching process is sub-critical. To convert from our delegation branching process to the Poisson branching process, we take a voter's
type to be
their competence (which completely characterizes their delegation
behavior). However, to compare to the Poisson process, there
must be a finite number of types. Hence, we partition the
interval $[0,1]$ into $B$ buckets, each of size $1/B$, such that
voters in the same bucket delegate and are delegated to
``similarly''. We choose $B$ large enough such that all points in
$[0, 1]^2$ within a distance of $\sqrt{2}/B$ of each other differ in
$\p$ by at most $L \cdot \e$. (Recall that the range of $\p$ is in the interval $[L,U]$.)
This is possible since $\p$ is uniformly
continuous. Further, this implies any points $(x, y), (x', y')$ within
a square with side length $1/B$ have the property that $\p(x, y) \le
\p(x', y') + L \cdot \e \le (1 + \e) \cdot \p(x', y')$. Note that $B$
depends only on $\p$ and $\e$, and hence is a constant with respect to
the number of voters $n$. 

We say a voter $i$ is of type $\tau$ if $\frac{\tau-1}{B} <
p_i\leq\frac{\tau}{B}$ for $1\leq\tau\leq B$ (with a non-strict
inequality for $\tau = 1$, so $0$ is of type $1$). Let $S_{\tau} =
(\frac{\tau-1}{B}, \frac{\tau}{B}]$ be the set of competencies of type $\tau$ (except that, in the case that $\tau = 1$, we take $S_{1}$ to be
the closed interval $[0,\frac{1}{B}]$). Let $\pi_{\tau} = \D[S_\tau]$ be the probability
that a voter has type 
$\tau$. Since the types form a partition of $[0, 1]$, we have that
$\sum_{\tau = 1}^B \pi_\tau = 1$. 

For any two types $\tau, \tau',$ we define
\begin{equation*}
    \p'(\tau, \tau') = \sup_{(x,y)\in S_{\tau}\times S_{\tau'}}\p(x,
    y).\footnote{Note that, because $\p$ is increasing in its second
    coordinate, one can actually write $\Tilde{\p}(\tau, \tau') =
    \sup_{x\in S_{\tau}}\p(x, \frac{\tau'}{B})$.}  
\end{equation*}
We abuse notation by extending $\p'$ to operate directly on competencies in $[0, 1]$ by first converting competencies to types and then applying $\p'$.
Then, $\p'$ has the property that for any $p_i, p_j \in [0, 1]$, $$\p(p_i, p_j) \le \p'(p_i, p_j) \le (1 + \e)\p(p_i, p_j).$$
We have that for all $\tau$, if $x \in S_\tau$, then $$\sum_{\tau' =1
}^B \p'(\tau, \tau') \pi_{\tau'} = \E_{\D}[\p'(x, \cdot)] \le (1
+ \e)\cdot \E_{\D}[\p(x, \cdot)] = (1 + \e).$$ Hence, we define 
$$
    \tilde{\p}(\tau, \tau') = \p'(\tau, \tau') \cdot \frac{(1 +
      \e)}{\sum_{\tau'' = 1}^B \p'(\tau, \tau'')\pi_{\tau''}}. 
$$
We again abuse notation to allow $\tilde{\p}$ to operate directly on competencies. We have that $\tilde{\p}(x, y) \ge \p'(x, y) \ge \p(x, y)$ for all competencies $x, y \in [0, 1]$ and further, for all $\tau$, $\sum_{\tau' = 1}^B \tilde{\p}(\tau, \tau') \pi_{\tau'} = 1 + \e$.

The Poisson branching process we will eventually compare to is one with $B$ types parameterized by the expected-children matrix $M$, where 
$$M_{\tau \tau'} = p\frac{(1+\e)^2}{1-2\e}\Tilde{\p}(\tau, \tau').$$
First, we show that $M$ has largest eigenvalue strictly less than
$1$ (for our choice of $\e$), so that the branching process will be
subcritical. Indeed, $M$ has only positive entries, so we need only
exhibit an eigenvector with all nonnegative entries such that the
associated eigenvalue is strictly less than $1$. The Perron-Frobenius theorem tells us this eigenvalue must be maximal. 

The remainder of this proof can be found in Appendix~\ref{app:proof_spm}. At a high level, we give details for proving the Poisson process is subcritical, as well as completing the comparison of between the original delegation process and this one. The comparison makes use of the concentration of the number of voters in each bucket. The proofs of \eqref{constr:2} and \eqref{constr:3} follow a similar structure to Confidence-based, however, they are quite a bit more intricate due to the inter dependencies between competence level and delegation probability.
\Halmos\endproof


\section{Liquid Democracy in Experiments}
In six experiments, we statistically estimate the functions $q$ and $\p$ to assess the real-world implications of our theoretical findings. These experiments rely on a novel design measuring the 
vote  n anon-strategic, non-incentivized liquid democracy stting, while simultaneously estimating voters' competence. Our empirical result are 
consistent
exhibits
a regime in which liquid democracy enhances 
collective intelligence,
leveraging interpersonal knowledge embedded in social networks and identifying diverse sets of experts. Data and code are available at \url{http://tinyurl.com/osf-liqdem}\footnote{Full link: \url{https://osf.io/skxwg/?view_only=3671d431bcfd4a9cb94ded5aa86a0a95}}.

\subsection{Experimental Design}
\subsubsection{Experiments and Material}
We conducted $E=6$ experiments after an initial pre-test\footnote{A description of the setup and results from the pre-study can be found in Appendix~\ref{app:pre} and initial results can be found in~\cite{revelliquid}.} between March 21st and November 27th, 2022.\footnote{Our protocol E-3948 was approved and exempted by the university Committee on the Use of Humans as Experimental Subjects.}  In each experiment $e$, a group of participants\footnote{Note that liquid 
democracy depends on the potential for beneficial delegation. It is therefore necessary to work with participants that have at least a passing familiarity with each other.
Experiments were conducted in places such as classrooms and company workshops, where pre-existing group structures guaranteed such conditions. While significant preparation was needed to ensure correct experimental set-ups for these 
environments, this design did have the benefit of producing high-quality data with few missing entries and minimal drop-out.} performed $|\mathcal{T}_e|$ tasks\footnote{$|\mathcal{T}_e| = 4$, except for experiment $e=6$ in which $|\mathcal{T}_e| = 12$: the final experiment was conducted over a longer period of time, allowing more tasks to be completed.} Each task consisted of $8$ questions on the corresponding 
topic
that were primarily taken and adapted from the work of \citet{simoiu2019studying}.\footnote{To be consistent with the theoretical setup under study, we converted all categorical questions into binary questions. For example, for a question from \citet{simoiu2019studying} of the form ``Where is this famous landmark from?'' with four options (Italy, Tibet, Greece, or Brazil), we selected a possible answer (e.g., Brazil) to reformulate the question as: ``Is this famous landmark from Brazil?'' In more detail, we first randomly selected which questions would be correct 
(to avoid the sense
that most questions are incorrect) and then, for the incorrect ones, drew a wrong option at random. We found multiple inconsistencies in the \citet{simoiu2019studying} data that we corrected, and the prediction questions pertained to events that had passed, so these were replaced with new ones.} 
A total of $N=168$ individuals participated. They hailed from over $30$ countries; $33\%$ were female, $1\%$ were non-binary, $64\%$ were male, and $2\%$ preferred to self-describe. Each experiment $e$ had a number of participants $N_e$ ranging between $14$ to $50$. A description of the settings and group sizes are presented in  Appendix~\ref{app:participation}, and the survey material can be found in Appendix~\ref{app:survey}


\subsubsection{Survey Flow}
Participants began by
providing informed consent and inputting their name. Next, they completed the following steps.

\emph{First experimental stage:} Participants 
were
presented with a task and could either answer a series of questions related to that theme or delegate the task to a peer. For instance, a task read: ``You will be shown images of architectural landmarks from around the world, and asked to select the country where the landmark is located,'' followed by ``Do you want to vote yourself or delegate your vote to a trusted peer?'' If they chose to vote themselves, they were taken to the $8$ questions contained in the task. If they chose to delegate, they were asked to select the name of their delegate and then immediately directed to the next task. Importantly, when deciding whether or not to delegate, participants did not see the questions.

\emph{Second experimental stage:} Participants were then asked to answer ``additional questions.'' These were all the questions corresponding to tasks they had chosen to delegate in the first stage. We collected this data at the end of the experiment so as not to prime the participants on the exercise.\footnote{We validated this approach with a robustness check on the time spent by participants as a function of how often they delegated (see 
Appendix \ref{app:rob}).}

Finally, optional background questions were asked on the last page. Note that the order in which tasks, questions within each task, and the ``True/False'' options appeared were all randomized. The entire flow is summarized in \Cref{Fig:surveyflow} in \Cref{app:survey}.

\subsubsection{Data Collected}\label{sec:data}
Let $[N]$ be the set of $N$ participants and let $[E]$, the set of $E$ experiments. Each experiment $e\in [E]$ has $N_e$ participants so that $N = \sum_{e\in[E]}N_e.$ $[N_e]$ denotes the subset of voters in experiment $e$ and $\mathcal{T}$ is the set of tasks surveyed ($|\mathcal{T}|=15$). For each task $t\in \mathcal{T}$ there is a set $R_t$ of $8$ corresponding questions. We let $R = \bigcup_t R_t$ be the set of all questions. For each participant $i$, $e(i) \in [E]$ is the experiment they participated in; for each question $r$, $t(r) \in \mathcal{T}$ is its corresponding task.

In the experiments, we collect (i) the direct vote to each question $i$ answered $v_{i, r} \in \{0,1\}$ where $1$ means correct and $0$ means incorrect and (ii) the binary signal $\delta_{i, t}$ equal to $1$ if $i$ delegated on task $t$ and $0$ otherwise (note that $\delta_{i, t}$ is constant at the task level), along with which voter they delegated to. From this collected data, we can compute $w_{i, t}$, the weight of voter $i$ on task $t$. This is $i'$s total weight after adding up all transitive delegations; it is set to $0$ when $i$ delegates. \Cref{Fig:graph_main} provides an example of a collected delegation graph.

In rare cases, a delegation could not be included for a couple of possible reasons. First, if a participant delegated to somebody who did not complete the survey. In this case, we would simply ignore the delegation (assuming they directly voted). Second, in an instance of a cycle (e.g., participant $i$ delegated to participant $j$ who delegated to participant $i$). These were also ignored (i.e., assumed that no voter on the cycle delegated). In many real-world implementations, such participants would be notified of the cycle and asked to choose a new delegate or vote directly.

\begin{figure}[htb]
  \centering
  \includegraphics[width=.9\linewidth]{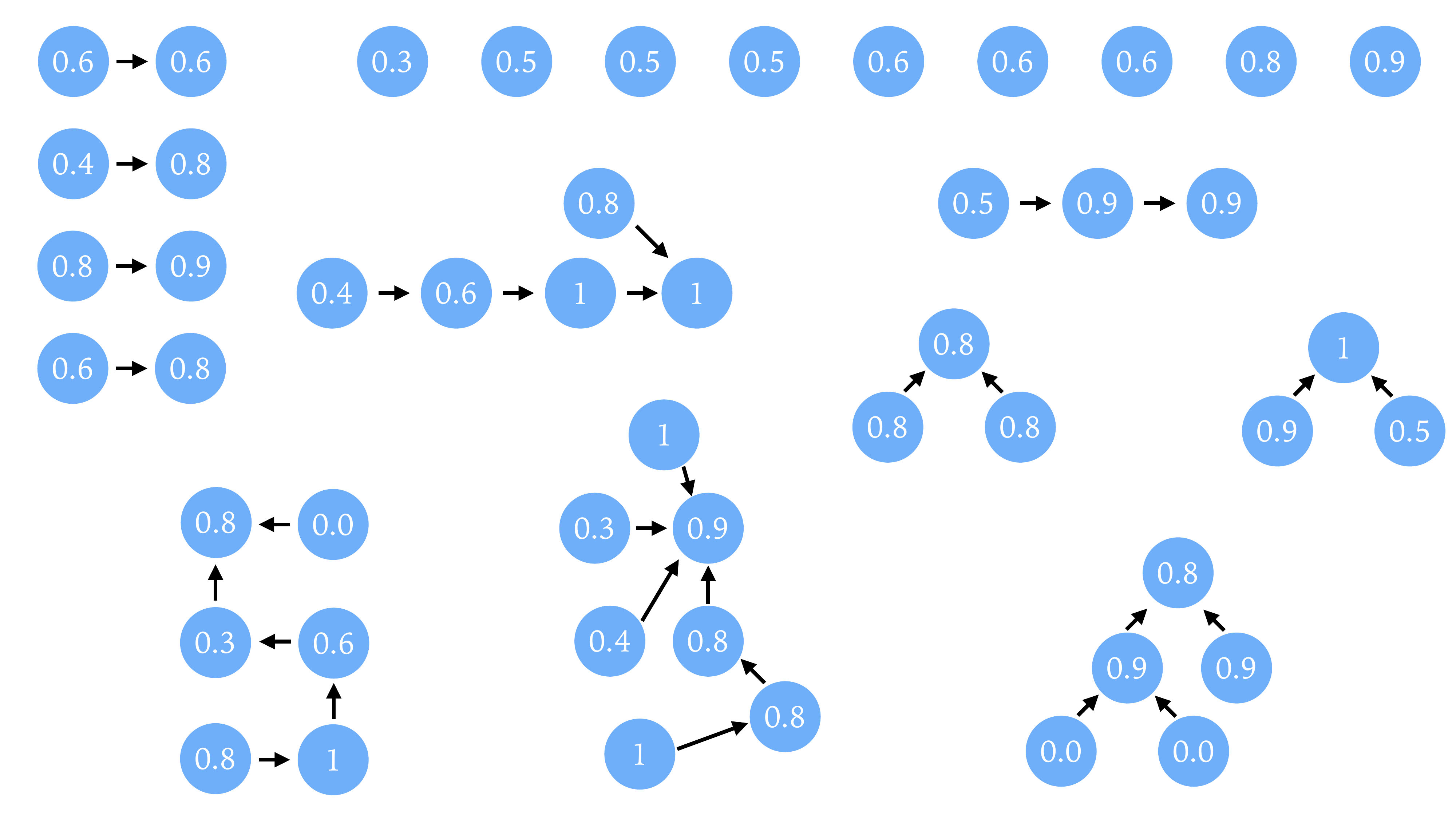}
  \caption{Delegation graphs for task $T_7$ ("You will be given upcoming European men soccer games and asked to predict the games' outcome.") from Experiment $6.$ Each node is a voter and the node's number represents the rounded expertise $\eta_{i, t}$ of a given voter $i$ for task $t,$ computed using Item Response Theory, see \Cref{app:IRT}.}
  \label{Fig:graph_main}
\end{figure}

\subsection{Delegation and Competence Statistics}
Over the $1096$ (participant/task) pairs, we observed a total of $505$ delegations meaning participants delegated  $47\%$ of the time ($std=0.49$) 
The rate varied across experiments from $32\%$ ($std=0.49$) in experiment $2$ to $54\%$ ($std=0.50$) in experiment $5$ and across tasks from  $22\%$ ($std=0.50$) in task $T_8$ to $80\%$ ($std=0.40$) in task $T_{15}.$ Among those who voted directly, $15\%$ received only one delegation besides their own (hence had weight $2$ in the decision), $6\%$ received two delegations, and just about $5\%$ received five or more delegations.
However, in one experiment, 
over half the votes were delegated to a single participant. Additionally, throughout all the experiments, we observed only four delegation cycles, and all were only of size two (where $a$ delegates to $b$, and $b$ delegates back to $a$). These occurred in Experiment $4$ with $N_4=27,$ and in Experiment $6$ with $N_6=50.$ Examples of additional delegation graphs can be found in Appendix~\ref{app:graph-examples}.

\subsubsection{Estimating Competence}\label{comp}
In order to evaluate how delegation behavior relates to competence, we need to estimate participants' competence. We denote by $\eta_{i, t}$ the estimated competence of participant $i$ in task $t.$ Naively, participants' competence per task could be estimated 
by averaging the number of correct
answers given on all $8$ questions of that task,  $\eta^{\text{naive}}_{i, t} = \frac{\sum_{r \in R_t} v_{i, r}}{|R_t|}.$  However, such a computation does not account for the questions' heterogeneity. 
We thus estimate
$\eta_{i, t}$ using the Item Response Theory framework (IRT)~\citep{lalor2023py}, which
provides a widely used parametric model to estimate competence $\eta_{i, t}$ and question difficulty from repeated measurements. We explain the parametric estimation in Appendix~\ref{app:IRT}.\footnote{While $\eta^{\text{naive}}_{i, t}$ takes on one of nine values (multiples of $1/8$), $\eta_{i, t}$ (computed using IRT) is a continuous variable that can take on arbitrary values in $\mathbb{R}.$ We normalize so that $\eta_{i, t} \in [0,1]$, and assume this to be the competence, the probability of being correct. Note that these different methods yield a correlation between $\eta^{\text{naive}}_{i, t}$ and $\eta_{i, t}$ of more than $94\%.$}

\subsubsection{Gender-based statistics} While we might worry that delegation patterns vary across gender due to significant differences in confidence~\cite[e.g.,][]{ellis2016women,sarsons2021confidence}, we 
actually
find no significant differences in these experiments, neither in measured competence in tasks nor in propensity to delegate. ANOVA tests for the propensity to delegate (resp. competence) across gender shows no significant differences with  $p=0.464$ (resp. $p=0.112$). Tukey tests for pairwise mean comparison further validate the absence of significant differences across the different genders (see Appendix~\ref{app:tukey}).

\subsection{Estimating the Probability of Delegating as a Function of Competence}\label{sec:method_q}
We now turn to estimating $q$ and $\p$ as a function of voters' competence. Recall that $q(\eta)$ represents the probability that somebody of competence $\eta$ chooses to delegate. We have observations $\delta_{i,t}$ encoding participant $i'$s delegation choice for task $t$, and an estimate $\eta_{i, t}$ of $i$'s competence on task $t.$ We use these to estimate the relationship between competence $\eta_{i, t}$ and the probability of delegating $q(\eta_{i, t}).$ 

\subsubsection{Methods} To estimate $q$, we fit a logistic model, regressing $\delta_{i,t}$ against $\eta_{i, t}$. The following equation shows the relationship we wish to fit, where $\alpha_0$ is the intercept and $\beta^q$ is the effect size we measure:
\begin{equation}
    \log\left(\frac{\Pr[\delta_{i, t}=1]}{1 - \Pr[\delta_{i, t}=1]}\right) = \alpha_0 + \beta^q \eta_{i, t} + \varepsilon_{i}.
\label{equ:q0}
\end{equation}
To account for potential correlation in the error term within participants' answers, when estimating the parameters in \Cref{equ:q0}, we cluster standard errors at the participant level. We also test for the data normality; results for these test can be found in Appendix~\ref{app:norm}.

We repeat the procedure above
on data sets filtered by task, this time having a distinct $\beta_t$ for each task $t$ to measure the task-specific estimates. Additionally, we run these with fixed-effects for individuals and tasks to more directly measure the impact of competence (rather than just looking at population trends). Additional details and results can be found in Appendix~\ref{app:q}.

\subsubsection{Results} 

We find $\beta^q = -2.24,$ with standard error $s.e.=0.42$, statistics $z=-7.12$ and p-value $p=10^{-7}.$ In turn,  we estimate that
$q(\eta_{i, t}) = \widehat{\Pr}[\delta_{i, t}=1] = \frac{1}{1 + \exp^{-(-1.39-2.24\times\eta_{i, t})}},$
suggesting that the probability of delegating decreases with competence. We can also test for monotonic dependence through a model-free method using a Pearson correlation test and its associated p-value. We find a correlation coefficient of $-0.17$ and $p<5\times10^{-8}.$

\subsection{Estimating Weight Function Used to Delegate}\label{sec:method_phi}
Recall that in the theoretical model, a voter with competence $\eta_1$ delegates to another with competence $\eta_2$ with probability proportional to $\varphi(\eta_1, \eta_2)$.

\subsubsection{Methods}
We first bucket the observed competence levels into $B$ clusters $c_1, \ldots, c_B$. We assume that $\varphi$ is constant across inputs in the same bucket, and fit it based on bucket ``centers,'' $\eta_1, \ldots \eta_B$, which are simply taken to be the mean values of the competences in each bucket, i.e., $\eta_\ell = \frac{\sum_{i, t: \eta_{i, t} \in c_\ell} \eta_{i, t}} {|\{(i, t) \mid \eta_{i, t} \in c_\ell\}|}$. This means we can estimate $\varphi(x, y)$ using the number of delegations from any competence $x'$ to competence $y'$ where $x'$ and $y'$ fall in the same bucket as $x$ and $y$, respectively. Finally, we 
determine the Kendall tau rank
correlation coefficient 
between
$\varphi(x, y)$ and $y$ with its associated p-value to test for monotonic relation between $\varphi$ and its second coordinate.


\paragraph{Bucketing strategies.} We discritize the segment $[0,1]$ into $B$ buckets. We do so using several methods (to ensure the robustness of our approach); we describe here the $k$-means clustering procedure and discuss the rest in Appendix~\ref{app:bucket-strategies}.

To bucket using $k$-means, we optimize for $B$ clusters, $c_1,\ldots,c_B$, that minimize $\sum_{k=1}^B\sum_{\eta_{i,t}\in c_k} \left(\eta_{i,t}-\frac{\sum_{\eta_{i,t}\in c_k}\eta_{i,t}}{|c_k|}\right)^2.$ In words, we compute a partition of the $[0,1]$ segment such that the total squared distance from elements to their cluster centers is minimized. We use the standard $k$-means clustering algorithm to find the clusters~\citep{hartigan1979k}.

\paragraph{Estimation of $\varphi$ for a given delegation graph.}
Next, we wish to fit a function $\varphi.$ For given experiment $e$ and task $t$, we estimate $\varphi_{e, t}(\eta_\ell, \eta_k)$ for each $\ell, k \in [B],$ so for conciseness, we write  $\varphi^{\ell}_{e, t}(\eta_k) := \varphi_{e, t}(\eta_\ell, \eta_k)$.\footnote{Note that because the number of participants in each bucket changes for different experiments/tasks, it is difficult to fit a single function. Instead, we first 
found
the most likely $\varphi$ to have generated each experiment/task and then 
combined
these to find an overall best fit.}

Observe that the experiment can then be viewed as a multinomial trial: from the perspective of an $\ell-$participant, there are $B$ choices to pick from, where the probability of picking a $k-$participant is proportional to all the $\varphi^{\ell}_{e, t}(\eta_k)$ for $k\in[B]$ and the number of participants of each competence level. We observe instances of these choices, the $z^\ell_k$ that are the observed number of times that someone of type $\ell$ delegated to someone of type $k.$ It then suffices to find the maximum likelihood estimators for $\varphi^{\ell}_{e, t}(\eta_k)$ as a function of $z^\ell_k$, $n_k$ and $n_{\ell}.$ We provide more technical details in Appendix~\ref{app:mle-phi}.

\paragraph{Testing for monotonic dependence of $\varphi$ in its second coordinate.}\label{phi:test}
Finally, we test for potential monotonic dependence of $\varphi_{e, t}^{\ell}(\eta_k)$ as a function of $\eta_k$ 
visualizing the $\varphi_{e, t}^{\ell}(\eta_k)$ in \Cref{Fig:phiQ} and computing 
the Kendall tau rank
correlation coefficient between $\varphi_{e, t}^{\ell}(\eta_k)$ and $\eta_k$ for a fixed $\ell.$ and its associated p-value. The Kendall tau rank correlation coefficient evaluates the similarity between two vectors of rank -- its form and significance is detailed in  Appendix~\citeauthor{app:phi-kendall-tau}.

\subsubsection{Results} We show here the results for using k-means bucketing with $B=4$.\footnote{This was the optimal number found using the \emph{Kneedle algorithm}~\citep{satopaa2011finding}} Additional results for other strategies and other numbers of buckets can be found in Appendix~\ref{app:buckets-results}. Descriptions of these four buckets can be found in \Cref{tab:bucket-stats}.

\begin{table}[htb]
    \centering
    \caption{Bucket Descriptions}
    \label{tab:bucket-stats}
    \begin{tabular}{cccc}
    \toprule
        Bucket & Interval & Mean competence & Proportion of participants \\
        \midrule
        $c_1$ &$[0.00, 0.514]$ &$0.43$&$16\%$ \\
        $c_2$ &$[0.515, 0.674]$&$0.60$& $ 32\%$\\
        $c_3$ & $[0.677, 0.814]$&$0.75$& $35\%$\\
        $c_4$ & $[0.818, 1.00]$&$0.88$& $17\%$\\
        \bottomrule
    \end{tabular}
    
\end{table}

In 
\Cref{Fig:phiQ}, the blue crosses in each column show the $\varphi^{\ell}_{e, t}(\eta_k)$ for all $(e, t)$ (for a particular $\eta_\ell$ in the four plots on the left and for all combined in the right plot) as a function $\eta_k$. The pink points represent the average across all experiments and tasks for a given $\eta_k,$ and the regression line corresponds to an ordinary least square regression on the mean values. We show this both pooled only for those in the same bucket, as well as all grouped together.



\begin{figure}[htb]
\centering

\begin{subfigure}[b]{.55\textwidth}
\includegraphics[width=\linewidth]{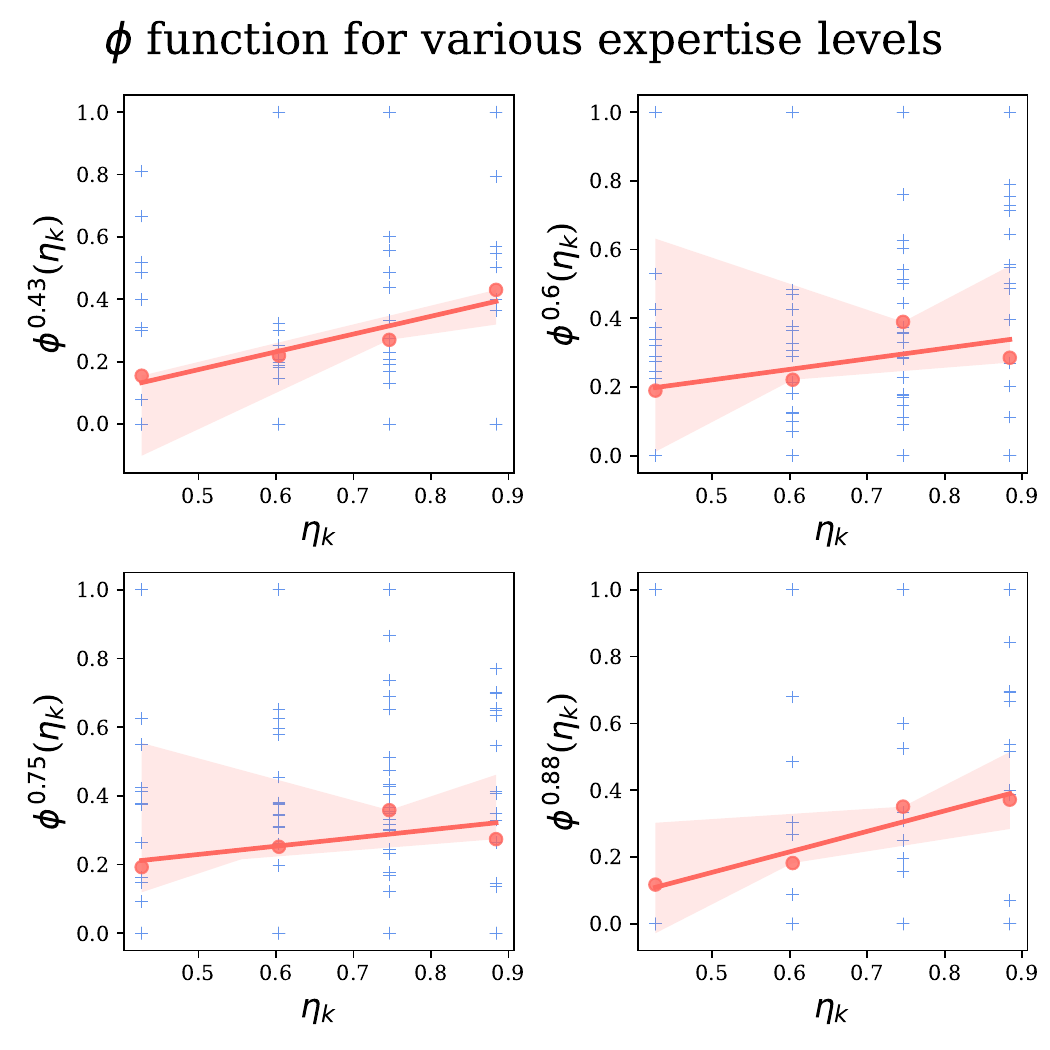}
\caption{Estimated values separated by input bucket.}
\end{subfigure}
\begin{subfigure}[b]{.44\textwidth}
    \includegraphics[width=\linewidth]{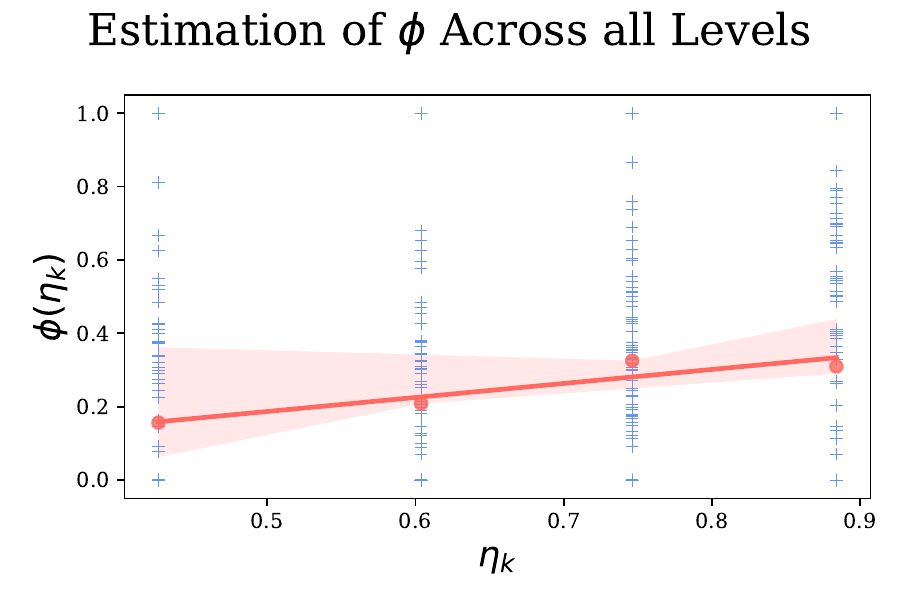}
    \vspace{3.3em}
    \caption{\small Estimated values pooled together.}
    \end{subfigure}
  \caption{Pooled estimates of $\varphi^{\ell}_{e, t}$, both for each bucket individually, and grouped together. The blue crosses show the values computed for $\varphi_{e,t}^{\ell}(\eta_k)$. The pink dots show the average across all values for that $\eta_k$, and the pink lines correspond to a linear regression over the mean values. We observe increasing trends across the board, with slope (coefficient of determination) being $0.53 (0.90), 0.28 (0.46), 0.29 (0.47)$ and $0.60 (0.92)$, respectively, for individual buckets, and $0.38 (0.85)$ for the pooled test. The shaded area represents the $95\%$ confidence interval.\label{Fig:phiQ}}
\end{figure}

To test the significance of the trends observed in \Cref{Fig:phiQ}, we test whether the Kendall 
tau
rank correlation coefficient between $\varphi_{e, t}^{\ell}(\eta_k)$ and $\eta_k$ signals significant associations, both at the overall level and when fixing $\ell,$ or $\eta_{\ell},$ the first coordinate in $\varphi_{e, t}(\eta_{\ell}, \eta_k).$ \Cref{phi:test_table} shows the resulting correlation coefficients and significance tests. Both the trends observed in \Cref{Fig:phiQ} and in \Cref{{phi:test_table}} confirm that there is 
a
statistically significant increase in $\varphi_{e, t}(\eta_{\ell}, \eta_k)$ as a function $\eta_{\ell}$ across all 
expertise levels,confirming that voters behave according to 
the general continuous delegation model. We further note that these significant trends are valid at the granularity of three of the four buckets (the third bucket $c_3$ 
exhibits
non-statistically significant positive 
Kendall
tau rank correlation.)
We also run the same tests partitioned into tasks. The results can be found in Appendix~\ref{app:phi-task-specific}. We last check that these results are not sensitive to the bucketing strategy in~\ref{app:buckets-results}.

\begin{table}[htb]
\caption{Summary of correlation effects}\label{phi:test_table}

  \centering
{  \footnotesize\begin{adjustbox}{angle=0}\begin{tabular}{@{\extracolsep{5pt}}lccccc} 
\\[-1.8ex]\hline 
\hline \\[-1.8ex] 
& \multicolumn{1}{c@{}}{Overall} & 
\multicolumn{3}{c@{}}{For fixed $\ell$}\\ 
\cmidrule(lr){2-2}
\cmidrule(lr){3-6} \\[-2ex] 
& & $c_1$ & \mc{$c_2$} & \mc{$c_3$} & \mc{$c_4$}\\ 
\midrule
Correlation & $0.17^{****}$ & 
 $0.29^{**}$ & $0.12^{*}$ & $0.11$ & $0.28^{***}$\\ 
  P-value & $2\times 10^{-5}$ & 
  $2\times10^{-2}$& $9\times10^{-2}$ & $1\times10^{-1}$ & $3\times10^{-3}$\\ \hline \\ [-1.8ex] \textit{Note:}  & \multicolumn{5}{r}{$^{*}$p$<$0.1; $^{**}$p$<$0.05; $^{***}$p$<$0.01; $^{****}$p$<$0.0001} 
\end{tabular}\end{adjustbox}}
\end{table}
\subsection{Experimental Conclusions}
We found that
voters' likelihood to delegate decreases with their competence (as suggested by the confidence-based model). In addition, voters are more likely to delegate to someone of increasing competence (as suggested by the general continuous model). Note that, in fact, the general continuous model can be generalized to allow monotonically decreasing $q$ as well (indeed, it suffices to consider the expectation of $q(p_i)$ 
taken over the distribution
of competence as 
the
constant probability of voting). 
Our empirical 
results are hence
consistent
with a general continuous delegation model. Unsurprisingly, the 
upward-delegation
model we described based on \cite{caragiannis2019contribution,kahng} that 
leads to a catastrophic concentration of power is not consistent with experimental data: 
voters do not delegate only to higher-competence agents, and we do not observe constant 
delegation competence.

\subsection{Additional Results}
In the Appendix~\ref{app:additional}, we run additional tests and check other properties of the collected data. These include comparing the frequency of correctness between liquid and direct democracy (Appendix~\ref{app:frequency}), analyzing the increase in competence (Appendix~\ref{app:increase}), and analyzing the concentration of power,  using both the maximum weight (Appendix~\ref{app:maxweight}) and the power of small coalitions (Appendix~\ref{app:coalition}). Note that 
the latter results are particularly relevvant to the
concentration of power in liquid democracy, a topic 
that was extensively discussed in previous work. Recall that, while previous work exhibited scenarios in which concentration of power occurs, our work is concerned with measuring whether such 
an extreme concentration of power is likely. In all but 
one of the $32$ instances, we found no evidence of concentration of power. Anecdotally, we further observe that a number roughly the square root of the number of voters controls half of the votes.

\section{Discussion}
\label{sec:disc}


Our paper relies on a set of assumptions and modeling choices that are worth discussing.

First, a prominent feature of our model is that 
there is no underlying social network, that is, there is no restriction on whom a voter may delegate to. As we explained in Section~\ref{sec:intro}, we believe this is realistic in a variety of scenarios. But we can, in fact, extend our results to a model where a directed social network is first sampled, and then a $(q, \p)$-model is followed. The social network must be sampled such that the neighbors of each voter are chosen uniformly at random, although the number of such neighbors could follow any small-tailed distribution.
Intuitively, delegation proportional to weighting the neighbors of $i$ (rather than the entire population) is equivalent to a possibly different weighting over the entire population. \footnote{This extension does not carry over to undirected networks, since if voters have a small number of neighbors, we would expect many 2-cycles to form after delegation, which, under the worst-case cycle approach, may not be canceled out by the overall increase in competence.}An open research direction is to consider graph topologies not covered by these dynamics.

Second, building on \citet{kahng}, we assume that there exists a true best alternative. Needless to say, this assumption is necessary if we wish to ``defend'' liquid democracy against their conclusions. It is also an extremely well-studied assumption that dates back to the 
18th century~\citep{Young88}. The existence of a ground truth is easily justified in the contexts of prediction markets or corporate governance, where alternative policies can be measured in terms of concrete metrics like ``estimated revenue in five years,'' and these metrics can be communicated to voters. That said, some decisions inherently rely on other subjective criteria that we do not capture.


Third, again like previous papers~\citep{kahng,caragiannis2019contribution,becker2021can}, we assume that voters vote independently. Admittedly, this is not a realistic assumption; relaxing it, as it was relaxed for the classic Condorcet Jury Theorem~\citep{haggstrom2006law, nitzan2017collective}, is a natural direction for future work.


Fourth, our models do not take strategic behaviors into account. In the same vein, our experiments 
do not involve
explicit incentives and, while we do not have reasons to believe that significant strategic voting occurs, studying it 
these issues is beyond the scope of this work. It would be of interest to extend our results to a more game-theoretic setting, and relate them to work focusing on
game-theoretic issues in 
liquid democracy~\citep{bloembergen2019rational,zhang2021power,dhillon2023information}. Along those lines, experiments with monetary incentives would be interesting.

Fifth, our framework for stochastic delegations open interesting directions for more research. For instance, one could 
characterize all the delegation models satisfying positive gain and do no harm. One may also consider more general 
local-delegation models
that would depend on the competences of all voters.

More generally, our work aims to provide a better understanding of a prominent shortcoming of liquid democracy: concentration of power. But there are others. For example, any voter can see the complete delegation graph under current liquid democracy systems\emdash a feature that helps voters make informed delegation decisions (because one's vote can be transitively delegated). 
This may lead to voter coercion, however, and the tradeoff between transparency and security is poorly understood.

Finally, to summarize, we have introduced a general framework to investigate stochastic dynamics in liquid democracy and 
proved
new conditions for the convergence of weighted majorities; we then 
identified
regimes in which liquid democracy leads to correct outcomes with high probability. 
In that sense, our work is to liquid democracy what the Condorcet Jury Theorem is to direct democracy. 
There are many reasons to be excited about the potential of liquid democracy~\citep{blum2016liquid}. We believe that our results provide another such reason and hope that our techniques will be useful in continuing to build the theoretical and empirical understanding of this compelling paradigm.




\bibliographystyle{informs2014}
\bibliography{abb,sample-bibliography}

\begin{thebibliography}{49}
\providecommand{\natexlab}[1]{#1}
\providecommand{\url}[1]{\texttt{#1}}
\providecommand{\urlprefix}{URL }

\bibitem[{Alon \protect\BIBand{} Spencer(2016)}]{alon2016probabilistic}
Alon N, Spencer JH (2016) \emph{The Probablistic Method} (John Wiley \& Sons).

\bibitem[{Atanasov et~al.(2017)Atanasov, Rescober, Stone, Swift,
  Servan-Schreiber, Tetlock, Ungar, \protect\BIBand{}
  Mellers}]{atanasov2017distilling}
Atanasov P, Rescober P, Stone E, Swift SA, Servan-Schreiber E, Tetlock P, Ungar
  L, Mellers B (2017) Distilling the wisdom of crowds: Prediction markets vs.
  prediction polls. \emph{Management science} 63(3):691--706.

\bibitem[{Barab{\'a}si \protect\BIBand{} Albert(1999)}]{barabasi1999emergence}
Barab{\'a}si AL, Albert R (1999) Emergence of scaling in random networks.
  \emph{science} 286(5439):509--512.

\bibitem[{Becker et~al.(2021)Becker, D’angelo, Delfaraz, \protect\BIBand{}
  Gilbert}]{becker2021can}
Becker R, D’angelo G, Delfaraz E, Gilbert H (2021) Unveiling the truth in
  liquid democracy with misinformed voters. \emph{Proceedings of the 7th
  International Conference on Algorithmic Decision Theory (ADT)}, 132--146.

\bibitem[{Benhaim et~al.(2023)Benhaim, Falk, \protect\BIBand{}
  Tsoukalas}]{benhaim2023scaling}
Benhaim A, Falk BH, Tsoukalas G (2023) Scaling blockchains: Can committee-based
  consensus help? \emph{Management Science} 69(11):6525--6539.

\bibitem[{Bloembergen et~al.(2019)Bloembergen, Grossi, \protect\BIBand{}
  Lackner}]{bloembergen2019rational}
Bloembergen D, Grossi D, Lackner M (2019) On rational delegations in liquid
  democracy. \emph{Proceedings of the 33rd AAAI Conference on Artificial
  Intelligence (AAAI)}, 1796--1803.

\bibitem[{Blum \protect\BIBand{} Zuber(2016)}]{blum2016liquid}
Blum C, Zuber CI (2016) Liquid democracy: Potentials, problems, and
  perspectives. \emph{Journal of Political Philosophy} 24(2):162--182.

\bibitem[{Bollob{\'a}s et~al.(2007)Bollob{\'a}s, Janson, \protect\BIBand{}
  Riordan}]{bollobas2007phase}
Bollob{\'a}s B, Janson S, Riordan O (2007) The phase transition in
  inhomogeneous random graphs. \emph{Random Structures \& Algorithms}
  31(1):3--122.

\bibitem[{Brill \protect\BIBand{} Talmon(2018)}]{brill2018pairwise}
Brill M, Talmon N (2018) Pairwise liquid democracy. \emph{Proceedings of the
  27th International Joint Conference on Artificial Intelligence (IJCAI)},
  137--143.

\bibitem[{Campbell et~al.(2022)Campbell, Casella, de~Lara, Mooers,
  \protect\BIBand{} Ravindran}]{campbell2022liquid}
Campbell J, Casella A, de~Lara L, Mooers VR, Ravindran D (2022) Liquid
  democracy. two experiments on delegation in voting. Technical report,
  National Bureau of Economic Research.

\bibitem[{Caragiannis \protect\BIBand{}
  Micha(2019)}]{caragiannis2019contribution}
Caragiannis I, Micha E (2019) A contribution to the critique of liquid
  democracy. \emph{Proceedings of the 28th International Joint Conference on
  Artificial Intelligence (IJCAI)}, 116--122.

\bibitem[{Chen et~al.(2008)Chen, Ingersoll~Jr, \protect\BIBand{}
  Kaplan}]{chen2008modeling}
Chen MK, Ingersoll~Jr JE, Kaplan EH (2008) Modeling a presidential prediction
  market. \emph{Management Science} 54(8):1381--1394.

\bibitem[{Christoff \protect\BIBand{} Grossi(2017)}]{christoff}
Christoff Z, Grossi D (2017) Binary voting with delegable proxy: An analysis of
  liquid democracy. \emph{Proceedings of the 16th Conference on Theoretical
  Aspects of Rationality and Knowledge (TARK)}, 134--150.

\bibitem[{Chung et~al.(2003)Chung, Handjani, \protect\BIBand{}
  Jungreis}]{chung2003generalizations}
Chung F, Handjani S, Jungreis D (2003) Generalizations of polya's urn problem.
  \emph{Annals of combinatorics} 7(2):141--153.

\bibitem[{Collevecchio et~al.(2013)Collevecchio, Cotar, \protect\BIBand{}
  LiCalzi}]{collevecchio2013preferential}
Collevecchio A, Cotar C, LiCalzi M (2013) On a preferential attachment and
  generalized p{\'o}lya’s urn model. \emph{The Annals of Applied Probability}
  23(3):1219--1253.

\bibitem[{Dhillon et~al.(2023)Dhillon, Kotsialou, Ravindran, \protect\BIBand{}
  Xefteris}]{dhillon2023information}
Dhillon A, Kotsialou G, Ravindran D, Xefteris D (2023) Information aggregation
  with delegation of votes. \emph{arXiv preprint arXiv:2306.03960} .

\bibitem[{Drinea et~al.(2001)Drinea, Enachescu, \protect\BIBand{}
  Mitzenmacher}]{drinea2001variations}
Drinea E, Enachescu M, Mitzenmacher MD (2001) Variations on random graph models
  for the web. Technical report, Harvard Computer Science Group.

\bibitem[{Durrett(2007)}]{durrett2007random}
Durrett R (2007) \emph{Random graph dynamics} (Cambridge University Press).

\bibitem[{Eggenberger \protect\BIBand{}
  P{\'o}lya(1923)}]{eggenberger1923statistik}
Eggenberger F, P{\'o}lya G (1923) {\"U}ber die statistik verketteter
  vorg{\"a}nge. \emph{Zeitschrift f{\"u}r Angewandte Mathematik und Mechanik}
  3(4):279--289.

\bibitem[{Ellis et~al.(2016)Ellis, Fosdick, \protect\BIBand{}
  Rasmussen}]{ellis2016women}
Ellis J, Fosdick BK, Rasmussen C (2016) Women 1.5 times more likely to leave
  stem pipeline after calculus compared to men: Lack of mathematical confidence
  a potential culprit. \emph{PloS one} 11(7):e0157447.

\bibitem[{Fortuin et~al.(1971)Fortuin, Kasteleyn, \protect\BIBand{}
  Ginibre}]{fortuin1971correlation}
Fortuin CM, Kasteleyn PW, Ginibre J (1971) Correlation inequalities on some
  partially ordered sets. \emph{Communications in Mathematical Physics}
  22(2):89--103.

\bibitem[{Gautschi(1959)}]{gautschi1959some}
Gautschi W (1959) Some elementary inequalities relating to the gamma and
  incomplete gamma function. \emph{Journal of Mathematics and Physics}
  38(1):77--81.

\bibitem[{G{\"o}lz et~al.(2018)G{\"o}lz, Kahng, Mackenzie, \protect\BIBand{}
  Procaccia}]{golz}
G{\"o}lz P, Kahng A, Mackenzie S, Procaccia AD (2018) The fluid mechanics of
  liquid democracy. \emph{Proceedings of the 14th Conference on Web and
  Internet Economics (WINE)}, 188--202.

\bibitem[{Green-Armytage(2015)}]{green}
Green-Armytage J (2015) Direct voting and proxy voting. \emph{Constitutional
  Political Economy} 26(2):190--220.

\bibitem[{H{\"a}ggstr{\"o}m et~al.(2006)H{\"a}ggstr{\"o}m, Kalai,
  \protect\BIBand{} Mossel}]{haggstrom2006law}
H{\"a}ggstr{\"o}m O, Kalai G, Mossel E (2006) A law of large numbers for
  weighted majority. \emph{Advances in Applied Mathematics} 37(1):112--123.

\bibitem[{Hardt \protect\BIBand{} Lopes(2015)}]{hardt2015google}
Hardt S, Lopes LCR (2015) Google votes: A liquid democracy experiment on a
  corporate social network. \emph{Technical Disclosure Commons} .

\bibitem[{Hartigan et~al.(1979)Hartigan, Wong et~al.}]{hartigan1979k}
Hartigan JA, Wong MA, et~al. (1979) A k-means clustering algorithm.
  \emph{Applied statistics} 28(1):100--108.

\bibitem[{Hoeffding(1963)}]{H63}
Hoeffding W (1963) Probability inequalities for sums of bounded random
  variables. \emph{Journal of the American Statistical Association}
  58(301):13--30.

\bibitem[{Huang(2023)}]{huang2023thy}
Huang J (2023) Thy neighbor’s vote: Peer effects in proxy voting.
  \emph{Management Science} 69(7):4169--4189.

\bibitem[{Janson(2020)}]{janson2020rate}
Janson S (2020) Rate of convergence for traditional p{\'o}lya urns.
  \emph{Journal of Applied Probability} 57(4):1029--1044.

\bibitem[{Johnson \protect\BIBand{} Kotz(1978)}]{munford1978urn}
Johnson NL, Kotz S (1978) \emph{Urn Models and Their Application: An Approach
  to Modern Discrete Probability Theory} (Wiley).

\bibitem[{Kahng et~al.(2021)Kahng, Mackenzie, \protect\BIBand{}
  Procaccia}]{kahng}
Kahng A, Mackenzie S, Procaccia AD (2021) Liquid democracy: An algorithmic
  perspective. \emph{Journal of Artificial Intelligence Research}
  70:1223--1252.

\bibitem[{Kling et~al.(2015)Kling, Kunegis, Hartmann, Strohmaier,
  \protect\BIBand{} Staab}]{kling2015voting}
Kling CC, Kunegis J, Hartmann H, Strohmaier M, Staab S (2015) Voting behaviour
  and power in online democracy: A study of {LiquidFeedback} in germany's
  pirate party. \emph{Proceedings of the 9th International AAAI Conference on
  Web and Social Media (ICWSM)}.

\bibitem[{Lalor \protect\BIBand{} Rodriguez(2023)}]{lalor2023py}
Lalor JP, Rodriguez P (2023) py-irt: A scalable item response theory library
  for python. \emph{INFORMS Journal on Computing} 35(1):5--13.

\bibitem[{Li et~al.(2023)Li, Xu, \protect\BIBand{} Duan}]{li2023liquid}
Li C, Xu R, Duan L (2023) Liquid democracy in dpos blockchains.
  \emph{Proceedings of the 5th ACM International Symposium on Blockchain and
  Secure Critical Infrastructure}, 25--33.

\bibitem[{List \protect\BIBand{} Goodin(2001)}]{LG01}
List C, Goodin RE (2001) Epistemic democracy: Generalizing the {Condorcet Jury
  Theorem}. \emph{Journal of Political Philosophy} 9(3):277--306.

\bibitem[{Mahmoud(2009)}]{mahmoud2009Polya}
Mahmoud H (2009) \emph{P{\'o}lya Urn Models} (CRC Press).

\bibitem[{Markov(1917)}]{markoff1917quelques}
Markov AA (1917) Sur quelques formules limites du calcul des probabilit{\'e}s.
  \emph{Bulletin de l'Acad\'emie des Sciences} 11(3):177--186.

\bibitem[{Natesan et~al.(2016)Natesan, Nandakumar, Minka, \protect\BIBand{}
  Rubright}]{natesan2016bayesian}
Natesan P, Nandakumar R, Minka T, Rubright JD (2016) Bayesian prior choice in
  irt estimation using mcmc and variational bayes. \emph{Frontiers in
  psychology} 7:1422.

\bibitem[{Nitzan \protect\BIBand{} Paroush(2017)}]{nitzan2017collective}
Nitzan S, Paroush J (2017) Collective decision making and jury theorems.
  \emph{The Oxford Handbook of Law and Economics} 1.

\bibitem[{Pivato(2012)}]{Piv12}
Pivato M (2012) A statistical approach to epistemic democracy. \emph{Episteme}
  9(2):115--137.

\bibitem[{P{\'o}lya(1930)}]{Polya1930quelques}
P{\'o}lya G (1930) Sur quelques points de la th{\'e}orie des probabilit{\'e}s.
  \emph{Annales de l'institut Henri Poincar{\'e}} 1(2):117--161.

\bibitem[{Revel et~al.(2022)Revel, Halpern, Berinsky, \protect\BIBand{}
  Jadbabaie}]{revelliquid}
Revel M, Halpern D, Berinsky A, Jadbabaie A (2022) Liquid democracy in
  practice: An empirical analysis of its epistemic performance.
  \emph{Proceedings of the 2nd ACM Conference on Equity and Access in
  Algorithms, Mechanisms, and Optimization (EAAMO)}.

\bibitem[{Sarsons \protect\BIBand{} Xu(2021)}]{sarsons2021confidence}
Sarsons H, Xu G (2021) Confidence men? evidence on confidence and gender among
  top economists. \emph{AEA Papers and Proceedings}, volume 111, 65--68
  (American Economic Association 2014 Broadway, Suite 305, Nashville, TN
  37203).

\bibitem[{Satopaa et~al.(2011)Satopaa, Albrecht, Irwin, \protect\BIBand{}
  Raghavan}]{satopaa2011finding}
Satopaa V, Albrecht J, Irwin D, Raghavan B (2011) Finding a" kneedle" in a
  haystack: Detecting knee points in system behavior. \emph{2011 31st
  international conference on distributed computing systems workshops},
  166--171 (IEEE).

\bibitem[{Simoiu et~al.(2019)Simoiu, Sumanth, Mysore, \protect\BIBand{}
  Goel}]{simoiu2019studying}
Simoiu C, Sumanth C, Mysore A, Goel S (2019) Studying the “wisdom of
  crowds” at scale. \emph{Proceedings of the AAAI Conference on Human
  Computation and Crowdsourcing}, volume~7, 171--179.

\bibitem[{Simon(1955)}]{simon1955class}
Simon HA (1955) On a class of skew distribution functions. \emph{Biometrika}
  42(3/4):425--440.

\bibitem[{Young(1988)}]{Young88}
Young HP (1988) Condorcet's theory of voting. \emph{The American Political
  Science Review} 82(4):1231--1244.

\bibitem[{Zhang \protect\BIBand{} Grossi(2021)}]{zhang2021power}
Zhang Y, Grossi D (2021) Power in liquid democracy. \emph{Proceedings of the
  35th AAAI Conference on Artificial Intelligence (AAAI)}, 5822--5830.

\end{thebibliography}

\newpage
\begin{center}
    \LARGE \textbf{Tracking Truth with Liquid Democracy - Appendix} \\[1.5em]
    \large Adam J. Berinsky, Department of Political Science, MIT \\[1em]
    \large Daniel Halpern,  School of Engineering and Applied Sciences, Harvard University\\[1em]
    \large Joseph Y. Halpern, Department of Computer Science, Cornell University\\[1em]
    \large Ali Jadbabaie, Institute for Data, Systems, and Society, MIT\\[1em]
    \large Elchanan Mossel, Department of Mathematics, MIT\\[1em]
    \large Ariel Procaccia, School of Engineering and Applied Sciences, Harvard University\\[1em]
    \large Manon Revel, Institute for Data, Systems, and Society, MIT (Corresponding author: mrevel@mit.edu)\\[1em]
\end{center}

\begin{APPENDIX}{}


\section{Hoeffding's Inequality}\label{app:proofHO}
Throughout many of the proofs, we will make use of the following well-known concentration inequality~\citep{H63}:
\begin{lemma}[Hoeffding's Inequality]
\label{lem:hoeffding}

Let $Z_1, \cdots , Z_n$ be independent, bounded
random variables with $Z_i \in [a, b]$ for all i, where $-\infty < a \leq b < \infty$. Then
\begin{equation*}
    \P\left[\frac{1}{n}\sum_{i=1}^n{Z_i - \E[Z_i]}\geq t\right] \leq \exp\left(-\frac{2nt^2}{(b-a)^2}\right)
\end{equation*}

and 

\begin{equation*}
    \P\left[\frac{1}{n}\sum_{i=1}^n{Z_i - \E[Z_i]}\leq -t\right] \leq \exp\left(-\frac{2nt^2}{(b-a)^2}\right)
\end{equation*}
for all $t\geq0.$
\end{lemma}

\section{Missing Proofs}
\subsection{Proof of Lemma \ref{lem:core}}\label{app:core-proof}

We establish the two properties separately.

\medskip
\emph{Probabilistic do-no-harm:}
We first show that a model $M$ that satisfies
conditions~\eqref{constr:1} and \eqref{constr:2}
satisfies probabilistic do no harm. Fix an arbitrary competence
distribution $\D \in \mathfrak{D}$ and let $\a$ and $C$ be such that
\eqref{constr:1} and \eqref{constr:2} are satisfied. Without loss of
generality, 
suppose that $C(n) \le n$ for all $n$, as replacing any larger values of $C(n)$ with $n$ will not affect \eqref{constr:1} (since
$\maxweight(G_n) \le n$ for all graphs $G_n$ on $n$ vertices). Fix
$\e, \d > 0$. We must identify some $n_0$ such that for all $ n \ge
n_0$, $\P_{\D, M,n}[\gain(\vec{p}_n, G_n) \ge  - \e] > 1 - \d$. 

We will begin by showing there exists $n_1 \in \N$ such that for all instances $(\vec{p}_n, G_n)$ on $n \ge n_1$ voters, if both
\begin{align}
    &\maxweight(G_n) \le C(n) \mbox{ and } \label{eq:cn}\\
    &\sum_{i = 1}^n \weight_i(G_n) \cdot p_i - \sum_{i = 1}^n p_i \ge 2 \a n, \label{eq:exp}
\end{align}
then
\begin{equation}
    \label{eq:gain}
    \gain(\vec{p}_n, G_n) \ge -\e.
\end{equation}

Since \eqref{eq:cn} and \eqref{eq:exp} each hold with probability $1 - o(1)$ by 
\eqref{constr:1} and \eqref{constr:2}, for sufficiently large $n$, say
$n \ge n_2$, 
they will each occur with probability at least $1 - \d / 2$. Hence, by
a union bound, for all $n \ge n_2$, they both occur with probability at least
$1 - \d$. By taking $n_0 = \max(n_1,n_2)$, this implies that
probabilistic do no harm is satisfied.

We now prove that, for sufficiently large $n$, \eqref{eq:cn} and \eqref{eq:exp}
imply \eqref{eq:gain}.
First, we will show that
\begin{equation}
    \label{eq:gain_suff}
    \gain(\vec{p}_n, G_n) \ge - \P_{\vec{p}_n}[X^D_n > X^F_{G_n}].
\end{equation}
Indeed, we have that
\begin{align*}
   \P_{\vec{p}_n}[X^D_n > n/2 ] &= \P_{\vec{p}_n}[X^D_n > n/2, X^F_{G_n} > n/2 ] + \P_{\vec{p}_n}[X^D_n > n/2, X^F_{G_n} \leq n/2 ] \\
   &\leq  \P_{\vec{p}_n}[X^F_{G_n} > n/2 ] + \P_{\vec{p}_n}[X^D_n > X^F_{G_n}]
\end{align*}
where the first 
transition holds by the law of total probability, 
and the second because the corresponding events are contained in each other. That is, $$\set{X^D_n > n/2, X^F_{G_n} > n/2} \subseteq \set{X^F_{G_n} > n/2}$$ and $$\set{X^D_n > n/2, X^F_{G_n} \leq n/2} \subseteq \set{X^D_n > X^F_{G_n}}.$$
Re-arranging the terms above yields \eqref{eq:gain_suff}.

Hence, for our purpose, it suffices to show that \eqref{eq:cn} and \eqref{eq:exp} imply $\P_{\vec{p}_n}\left[X^D_n > X^F_{G_n}\right] \le \e.$ Intuitively, we will use \eqref{eq:exp} to show the expected value of $X^D_n$ is well below the expected value of $X^F_{G_n}.$ Then we will show both $X^D_n$ and $X^F_{G_n}$ concentrate well around their means, where for the latter we will need \eqref{eq:cn}. Together, these observations imply that $X^F_{G_n} > X^D_n$ with high probability.

Fix an instance $(\vec{p}_n, G_n)$ on $n$ voters satisfying \eqref{eq:cn} and \eqref{eq:exp}. We will show that
for sufficiently large $n$,                                 
    \begin{equation}
    \label{eq:direct}
        \P_{\vec{p}_n}\left[X^D_n < \sum_{i = 1}^n p_i + \a  n \right] > 1 - \e / 2
    \end{equation}
    and
    \begin{equation}
        \label{eq:liquid}
        \P_{\vec{p}_n}\left[X^F_{G_n} > \sum_{i = 1}^n \weight_i(G_n) \cdot p_i - \a n \right] > 1 -  \e / 2.
    \end{equation}
Note that since \eqref{eq:exp} holds for this instance, $\sum_{i = 1}^n p_i + \a  n \le \sum_{i = 1}^n \weight_i(G_n) \cdot p_i - \a n$. Therefore, when both events 
whose probability is considered in \eqref{eq:direct} and \eqref{eq:liquid} hold, $X^D_n \le X^F_n$. Hence, 
\begin{align*}
    \P_{\vec{p}_n}[X^D_n \le X^F_{G_n}]
    &\ge \P_{\vec{p}_n}\left[X^D_n < \sum_{i = 1}^n p_i + \a  n, X^F_{G_n} > \sum_{i = 1}^n \weight_i(G_n) \cdot p_i - \a n \right] > 1 - \e
\end{align*}
where the last inequality holds by a union bound.
This implies that $\P_{\vec{p}_n}[X^D_n \le X^F_{G_n}] < \e$, as needed.



It remains to be shown that \eqref{eq:direct} and \eqref{eq:liquid} hold for
        sufficiently large $n$. For \eqref{eq:direct}, this follows
        directly from Hoeffding's inequality
        (\Cref{lem:hoeffding} in Appendix~\ref{app:proofHO}). To prove \eqref{eq:liquid}, first note
        that, as shown in \citet{kahng}, 
	\begin{align*}
		\Var_{\vec{p}_n}\left[X^F_{G_n}\right]
		&= \sum_{i = 1}^n \weight_i(G_n)^2  \cdot p_i (1 - p_i)\\
		&\le \frac{1}{4} \cdot \sum_{i = 1}^n \weight_i(G_n)^2\\
		&\le \frac{1}{4} \cdot \sum_{i = 1}^{\lceil n/C(n) \rceil} C(n)^2\\
		& < nC(n) \in o(n^2),
	\end{align*}
	where the first inequality holds because $p(1 - p)$ is upper
        bounded by $1/4$, the second because $\sum_{i = 1}^n
        \weight_i(G_n) \le n$ with each $\weight_i(G_n) \le C(n)$ so
        the value is maximized by setting as many terms to $C(n)$ as
                possible, and the final inequality holds because $C(n) \le n$. 
	
	Hence, by Chebyshev's inequality,
	$$
	    \P_{\vec{p}_n}\left[X^F_{G_n} \le \E_{\vec{p}_n}\left[X^F_{G_n}\right] - \a n \right]
	    \le \frac{\Var_{\vec{p}_n}[X^F_{G_n}]}{( \a n)^2}.
	$$
	This bound is $o(1)$ because the numerator is $o(n^2)$ and the
        denominators is $\Omega(n^2)$. This implies that for sufficiently large
        $n$, it will be strictly less than $\e/2$, so       \eqref{eq:liquid} holds. 
	
	\medskip
	\emph{Probabilistic positive gain:}
	Fix a distribution $\D
        \in \mathfrak{D}$ and an $\alpha \in (0,1)$ such that
        \eqref{constr:3} holds.  We want to show that $M$ 
        satisfies probabilistic positive gain. 
        Since $\D \in \mathfrak{D}$, it
        also satisfies \eqref{constr:1} for some $C$. We show below that
there exists an $n_3$ such that all
 instances $(\vec{p}_n, G_n)$ with $n \ge n_3$ voters satisfying
        \eqref{eq:cn} for which $\sum_{i = 1}^n  p_i + \a n \le n/2 \le
        \sum_{i = 1}^n \weight_i(G_n) \cdot p_i -  \a  n$, 
        we have that $\gain(\vec{p}_n,
        G_n) \ge 1- \e$. As with the DNH part of the proof, since the
        events of \eqref{constr:1} and \eqref{constr:3} each hold with probability
        $1 - o(1)$, for sufficiently large $n$, say $n \ge n_4$, they
        each occur with probability at least $1 - \d / 2$. Hence,
        by a union bound, for all $n \ge n_4$, they both occur with
        probability $1 - \d$. For $n_0 = \max(n_3,n_4)$,
        probabilistic positive gain is satisfied. 
	
	It remains to show that that if 
         \eqref{constr:1} and \eqref{constr:3}  hold for a specific
        instance $(\vec{p}_n, G_n)$, then $\gain(\vec{p}_n, G_n) \ge
        1- \e$ for sufficiently large $n$. Since $\D \in
        \mathfrak{D}$,  \eqref{eq:direct} and \eqref{eq:liquid} are both
        satisfied for sufficiently large $n$. When $$\sum_{i = 1}^n
        p_i + \a n \le n/2 \le \sum_{i = 1}^n p_i - \weight_i(G_n)
        \cdot \a  n$$ is satisfied as well, we get that
        $\P_{\vec{p}_n}\left[X^D_n > 
          n/2 \right] < \e / 2$ and $\P_{\vec{p}_n}\left[X^L_{G_n} > n/2
          \right] > 1 - \e / 2$, so $\gain( \vec{p}_n, G_n) > 1 - \e$
        is immediate. \qed

\subsection{Missing Details from the Proof of Lemma~\ref{lem:polya}}\label{app:lemma_proof}
The first missing details was the proof of \Cref{eq:UDM-exp},
\[
    \E[W_T^{(k)}] = \frac{\Gamma(T+p)\Gamma(k)}{\Gamma(p+k)\Gamma(T)}
\]
for all $k \le T$, where $\Gamma$ represents the Gamma function. The second was showing
\[
    \sum_{k = 1}^{T^\g} \P[W^{(k)}_T > n^\d] = o(1).
\]

Recall that $W_k^{(k)} = 1$ and we have the following recurrence for all $t > k$:  
$$
W_t^{(k)} = 
\begin{cases}
    W_{t-1}^{(k)} + 1 & \text{with probability } \frac{p \cdot W_{t-1}^{(k)}}{t-1} \\
        W_{t-1}^{(k)} & \text{with probability } 1- \frac{p \cdot W_{t-1}^{(k)}}{t-1}.
\end{cases}
$$
 By the tower property of expectation, for all $t \ge k + 1$,
\begin{align*}
    \E[W_t^{(k)}]
    &= \E[\E[W_t^{(k)} \suchthat W_{t-1}^{(k)}]]\\
    &=  \E[W_{t-1}^{(k)}(1 + \frac{p}{t-1})]\\
        &= \E[W_{t-1}^{(k)}](1 + \frac{p}{t-1}).
\end{align*}
Thus, by a straightforward induction argument and the fact that $\E[W_k^{(k)}] = 1$,
\[
    \E[W_T^{(k)}] = \E[W_k^{(k)}] \prod_{i=k}^{T-1} (1+\frac{p}{i}) = \prod_{i=k}^{T-1} (1+\frac{p}{i}).
\]
Expanding this, we have
\begin{align*}
    \prod_{i=k}^{T-1} (1+\frac{p}{i})
    &= \prod_{i=k}^{T-1} \frac{i+p}{i}\\
    &= \frac{1}{\prod_{i = k}^{T - 1} i} \cdot \prod_{i=k}^{T-1} (i + p)\\
    &=  \frac{(k-1)!}{(T-1)!} \cdot \frac{\prod_{i=0}^{T-1} (i+p)}{\prod_{i=0}^{k-1} (i+p)}\\
    & =  \frac{(k-1)!}{(T-1)!}\frac{\frac{\Gamma(p+T)}{\Gamma(p)}}{\frac{\Gamma(k + p)}{\Gamma(p)}}\\
    &= \frac{\Gamma(T+p)\Gamma(k)}{\Gamma(p+k)\Gamma(T)},
\end{align*}
where the fourth equality holds because
$\Gamma(x+1)=x\Gamma(x)$ for all $x \in \R$, and the last uses the fact
that $\Gamma(n)=(n-1)!$ for all $n\in\N$. 
This proves \eqref{eq:UDM-exp}.

For the second missing detail, can now use Markov's inequality to show that for all $k$,
\begin{equation*}
  \Pr\left[W_T^{(k)} > T^\d \right]
  \le \frac{\E[W_T^{(k)}]}{T^\d}
  = \frac{1}{T^{\d}} \cdot \frac{\Gamma(T+p)}{\Gamma(T)} \cdot  \frac{\Gamma(k)}{\Gamma(k + p)}.
\end{equation*}
Hence,
$$\sum_{k = 1}^{T^\g} \Pr[W_T^{(k)} > T^\d]
\le \sum_{k=1}^{T^\g} \frac{1}{T^{\d}} \cdot \frac{\Gamma(T+p)}{\Gamma(T)} \cdot  \frac{\Gamma(k)}{\Gamma(k + p)}\\
= \frac{1}{T^{\d}} \cdot \frac{\Gamma(T+p)}{\Gamma(T)} \cdot  \sum_{k=1}^{T^\g} \frac{\Gamma(k)}{\Gamma(k + p)}.$$
What remains to be shown is that 
$$\frac{1}{T^{\d}} \cdot \frac{\Gamma(T+p)}{\Gamma(T)} \cdot  \sum_{k=1}^{T^\g} \frac{\Gamma(k)}{\Gamma(k + p)} = o(1).$$
To do this, we will use Gautschi's inequality~\citep{gautschi1959some} which states that for all $x > 0$, since $p \in (0, 1)$,
\[
(x+p-1)^p \leq \frac{\Gamma(p+x)}{\Gamma(x)} \leq(x+p)^p
\]
We then have that
\begin{align*}
    \frac{1}{T^{\d}} \cdot \frac{\Gamma(T+p)}{\Gamma(T)} \cdot  \sum_{k=1}^{T^\g} \frac{\Gamma(k)}{\Gamma(k + p)}
    &\le \frac{(T + p)^p}{T^{\d}} \cdot    \sum_{k=1}^{T^\g} \frac{1}{(k + p - 1)^p}\\
    &= \frac{(T + p)^p}{T^{\d}} \cdot \left(\frac{1}{p^p} + \frac{1}{( 1 + p)^p} + \sum_{k=3}^{T^\g} \frac{1}{(k + p - 1)^p}\right)\\
    &\le \frac{(T + p)^p}{T^{\d}} \cdot \left(\frac{1}{p^p} + \frac{1}{( 1 + p)^p} + \sum_{k=3}^{T^\g} \frac{1}{(k - 1)^p}\right)\\
    &= \frac{(T + p)^p}{T^{\d}} \cdot \left(\frac{1}{p^p} + \frac{1}{( 1 + p)^p} + \sum_{k=2}^{T^\g - 1} \frac{1}{k^p}\right)\\
    &\le \frac{(T + p)^p}{T^{\d}} \cdot \left(\frac{1}{p^p} + \frac{1}{( 1 + p)^p} + \sum_{k=2}^{T^\g} \frac{1}{k^p}\right)\\
    &\le \frac{(T + p)^p}{T^{\d}} \cdot \left(\frac{1}{p^p} + \frac{1}{( 1 + p)^p} + \int_{1}^{T^\g} \frac{1}{x^p} \, \mathop{dx}\right)\\
    &= \frac{(T + p)^p}{T^{\d}} \cdot \left(\frac{1}{p^p} + \frac{1}{( 1 + p)^p} + \frac{x^{1 - p}}{1 - p} \Big|_{1}^{T^\g}\right)\\
    &= \frac{(T + p)^p}{T^{\d}} \cdot \left(\frac{T^{\g (1 - p)}}{1 - p} +\frac{1}{p^p} + \frac{1}{( 1 + p)^p} - \frac{1}{1 - p}\right).
\end{align*}
Notice that asymptotically, this upper bound is $O(T^{-\d +
  p + \g \cdot (1 - 
  p)})$. By our choice of $\d$, $\d > p + \g \cdot (1 - p)$, so this implies that it is is $o(1)$, as desired.
  
\subsection{Missing Details from the Proof of Theorem~\ref{thm:UDM}}\label{app:proof_udm}

The missing details were showing that the Upward Delegation Mechansim satisfied \eqref{constr:2} and \eqref{constr:3}.

\medskip

\emph{Upward Delegation satisfies \eqref{constr:2}}
\smallskip

We will show there exists  $\a \in (0,1)$ such that $\sum_{i = 1}^n
\weight_i(G_n) \cdot p_i - \sum_{i = 1}^n p_i \ge 2 \a n$ with high probability, so \eqref{constr:2} is satisfied. Note that in the present scheme, cycles are impossible, so do need to
worry about ignored voters. 

Since $\D$ is a continuous distribution, there exists $a < b$ such that $\pi_a := \D[\set{p: p < a}] > 0$ and $\pi_b := \D[\set{p: p > b}] > 0$.
Let $N_{a,n}(\vec{p}_n)$ be the number of voters in $\vec{p}_n$ with
competence $p_i < a$ and 
$N_{b, n}(\vec{p}_n)$ be the number of voters with competence $p_i > b$. When we
sample competencies, since each is chosen independently, $N_{a, n} \sim
\text{Bin}(n, \pi_a)$ and $N_{b, n} \sim \text{Bin}(n, \pi_b)$. By Hoeffding's inequality
(\Cref{lem:hoeffding}) and the union bound, with probability $1 - o(1)$,
there will be at least $\pi_a /2 \cdot n$ voters with competence $p_i
< a$ and $\pi_b / 2 \cdot n$ voters with competence $p_i > b$. Indeed, 
\begin{equation}\label{eq:step1}
\begin{aligned}
    \D^n[N_{a,n} > \frac{n\pi_a}{2}, N_{b,n} > \frac{n\pi_b}{2}] 
    &= 1 - \D^n[\{N_{a,n} \leq \frac{n\pi_a}{2}\} \cup \{N_{b,n} \leq \frac{n\pi_b}{2}\}]\\
    &\geq 1 - (\D^n[N_{a,n} \leq \frac{n\pi_a}{2}] + \D^n[N_{b,n} \leq
      \frac{n\pi_b}{2}])\\ 
    &\geq 1 - \exp(-\frac{n\pi_a^2}{2}) - \exp(-\frac{n\pi_b^2}{2}),
\end{aligned}
\end{equation}
where the first line comes from De Morgan's law, the second from the
union bound, and the last from Heoffding's inequality (\Cref{lem:hoeffding}).

Conditioned on this occurring, each voter with competence $p_i < a$ has probability at least $p \pi_b / 2$ of delegating to a voter
with competence at least $b$. As they each decide to do this independently, the number $N_{ab,n}$ of $n$ voters deciding to do this
stochastically dominates a random variable following the
$\text{Bin}(\pi_a / 2 \cdot n, p \cdot \pi_b / 2)$ distribution. We can again apply Hoeffding's inequality to conclude that with probability
$1 - o(1)$, at least $\pi_a \cdot \pi_b \cdot p / 8 \cdot n$ voters do
so.  
Indeed, 
\begin{equation}
\label{eq:step2}
    \begin{aligned}
      \D[N_{ab,n} > \frac{np\pi_a \pi_b}{8}\mid N_{a,n} > \frac{n\pi_a}{2},
        N_{b,n} > \frac{n\pi_b}{2}] &\geq  
    \D[\text{Bin}(\frac{n\pi_a}{2}, \frac{p\pi_b}{2}) > \frac{np\pi_a \pi_b}{8}]\\
        &\geq 1 - \exp(-\frac{np^2\pi_a\pi_b^2}{4}),
\end{aligned}
\end{equation}
where the first inequality holds because $N_{ab,n}$ stochastically dominates the corresponding binomial random variable
and the second  holds
by Hoeffding's inequality. Finally, using \eqref{eq:step1} and
\eqref{eq:step2}, we have 
\begin{equation*}
\label{eq:conclu}
    \begin{aligned}
    \D[N_{ab,n} > \frac{np\pi_a \pi_b}{8}]
    &\geq \D[N_{ab,n} > \frac{np\pi_a \pi_b}{8} \suchthat N_{a,n} > \frac{n\pi_a}{2}, N_{b,n} > \frac{n\pi_b}{2}]\\
    &\qquad \cdot \D[N_{a,n} > \frac{n\pi_a}{2}, N_{b,n} > \frac{n\pi_b}{2}]\\
        &\geq 1 - o(1).
\end{aligned}
\end{equation*}

Under these upward delegation models, delegations can only increase the total competence of all voters. Hence,
\begin{align*}
    \sum_{i = 1}^n \dels_i(G_n) \cdot p_i - \sum_{i = 1}^n p_i \geq (b-a)N_{ab.n}.
\end{align*}
Each of these $\pi_a \cdot \pi_b \cdot p / 8 \cdot n$ voters results
in a competence increase of at least $b - a$. Hence, under these high
probability events, the total competence increase is at least $(b - a)
\cdot \pi_a \cdot \pi_b \cdot p / 8 \cdot n$. Indeed, since
$\D[N_{ab,n} > \frac{np\pi_a \pi_b}{8}]=1-o(1)$, this implies
$\D[\sum_{i = 1}^n \dels_i(G_n) \cdot p_i - \sum_{i = 1}^n p_i  >
  \frac{np\pi_a \pi_b}{8}] = 1 - o(1)$. By choosing $\a =
\frac{p\pi_a\pi_b}{8}(b - a)$, we see that there is an $\a \cdot n$
increase in competence with high probability, as needed. 

We have proved that for any continuous distribution $\D$, and for well-behaved realizations of $\vec{p}_n$ which occur with high probability, the graph generated from the random delegation process yields an increase in the expected sum of the votes of at least $\a \cdot n$. We can then conclude that $M^U_p$ satisfies \Cref{constr:2} with respect to the class of continuous distributions.
\medskip

\emph{Upward Delegation satisfies \eqref{constr:3}}

\smallskip
We now show that there exists a distribution $\D$ such that $\sum_{i =
  1}^n  p_i + \a n \le n/2 \le \sum_{i = 1}^n \weight_i(G_n) \cdot p_i
-  \a  n$ with probability $1 - o(1)$ for some $\a > 0$. This 
implies that the model satisfies probabilistic positive gain by
\Cref{lem:core}, and will conclude the proof.  

We take $\D$ to be $\D_\eta$,
the uniform distribution $\U[0, 1-2\eta]$ for some small $0<\eta<p/512.$  Let $\a = \eta/2$. 
Clearly, $\mu_{\D_\eta}$, the mean of $D_\eta$, is $1/2-\eta.$ Since each $p_i \overset{i.i.d.}{\sim}\D_\eta$,
the $p_i$s are bounded independent random variables with mean
$1/2-\eta$, so Hoeffding's inequality directly implies that $\sum_{i = 1}^n
p_i \leq n/2 - n\eta/2 = n/2 - n\a$ with high probability. 

Now consider $\cE_F$, the event consisting of instances
$(\vec{p}_n,G_n)$ such 
that $\sum_{i =
  1}^n  \weight_i(G_n) \cdot p_i \geq n/2 + n\a.$ We denote by $\cE_D$
the event that $\sum_{i = 1}^n  p_i \geq n/2 - 3n\eta/2.$ The same reasoning as before implies that $\Pr_{\D_\eta,M^U_p,n}(\cE_D) = 1-o(1)$. 

Let $a=1/4-\eta/2$ and $b=1/2-\eta$, so we have that $\pi_a := \D_\eta[p_i < a] = 1/4$ and $\pi_b :=
\D_\eta[p_i > b] = 1/2.$ 
We proved in the preceding derivation that
$\sum_{i = 1}^n \weight_i(G_n) \cdot p_i - \sum_{i = 1}^n
p_i > \frac{np\pi_1\pi_b}{8} = \frac{np}{64}(1-\eta)$ with high probability. Hence, if both this and $\cE_D$ occur, which is the case with
high probability, by the union bound, it follows that $\sum_{i
  = 1}^n \weight_i(G_n) \cdot p_i > n/2 + n(\frac{p}{128}(1-\eta) -
3\eta/2)$ with high probability.  

Since $\eta<\frac{p}{512} < 1/2$, we have that $$\frac{p}{128}(1-\eta) - 3\eta/2> \frac{p}{256} - 3\eta/2 > 2\eta - 3\eta/2 = \eta/2 = \a,$$
and we can conclude that $\cE_F$ occurs with high probability. Hence, $M^U_p$ satisfies \Cref{constr:3}.

\subsection{Missing Details from the Proof of Theorem~\ref{thm:CBM}}\label{app:proof_cdm}
We show that the Confidence-based Model satisfies \eqref{constr:3}.

\medskip

\emph{Confidence-Based Delegation satisfies \eqref{constr:3}}
\smallskip

We finally show there exists a distribution $\D$ such that $\sum_{i =
  1}^n  p_i + \a n \le n/2 \le \sum_{i = 1}^n \weight_i(G_n) \cdot p_i
-  \a  n$ with probability $1 - o(1).$ This implies that the
model $M_q^C$ satisfies probabilistic positive gain by \Cref{lem:core}.

Using the notation of the analogous proof in \Cref{sec:up}, let
$\D_\eta = \U[0, 1-2\eta]$ for $\eta \in [0, 1/2)$. Note that as a 
function of $\eta$, $\frac{\E_{\D_\eta}[q^+]}{\E_{\D_\eta}[\bar{q}]}$,
  the expected competence conditioned on not delegating, is continuous. Moreover, if $\eta = 0$,  then
$$\frac{\E_{\D_0}[q^+]}{\E_{\D_0}[\bar{q}]} > \mu_{\D_0} = 1/2.$$
Hence, for sufficiently small $\eta > 0$, 
$$\frac{\E_{\D_\eta}[q^+]}{\E_{\D_\eta}[\bar{q}]} > 1/2 > \mu_{\D_\eta}.$$

We choose $\D_\eta$ to be our distribution for this choice of
$\eta$. As in the previous 
section, let $\mu^*_{\D_\eta} =
\frac{\E_{\D_\eta}[q^+]}{\E[\bar{q}]}$. Note that $\mu_{\D_\eta} = 
1/2 - \eta$. Let $\g = \min(\frac{1/2 - \mu_{\D_\eta}}{2}, \frac{\mu^*_{\D_\eta} -
  1/2}{2})$ and $\a = \g$. By the earlier argument for
\eqref{constr:2}, we have that with high probability 
$$
    \sum_{i = 1}^n p_i \le n(\mu + \g) \le n/2 - \a n 
$$
and
$$\sum_{i = 1}^n \weight_i(G)p_i \ge n(\mu^* - \g) \ge n/2 + \a n.$$

By the union bound, we have that both occur simultaneously with high
probability, so \eqref{constr:3} holds.

\subsection{Missing Details from the Proof of Theorem~\ref{thm:SPM}}\label{app:proof_spm}
We first show the remainder of the General Continuous Model satisfying \eqref{constr:1}.

The eigenvector we consider is $\vec{\pi} = (\pi_1, \ldots,
\pi_B)$ (which has nonnegative entries, as each $\pi_\tau$ is a
probability). We show it has eigenvalue
$p\frac{(1+\e)^3}{1-2\e}$, strictly less than $1$ due to our choice of
$\e$. We show it has eigenvalue
$p\frac{(1+\e)^3}{1-2\e}$, strictly less than $1$ due to our choice of
$\e$. Indeed, we have that 
$$
(M\vec{\pi})_\tau = \sum_{\tau'=1}^B \pi_{\tau} \tilde{\p}(\tau, \tau')\pi_{\tau'}=\pi_{\tau}p\frac{(1+\e)^3}{1-2\e}
$$
by the definition of  $\Tilde{\p}.$ Hence, $\vec{\pi}$ is our desired eigenvector.

Since $\mathfrak{X}^P_{M, \tau}$ is sub-critical for all $\tau$, we
have that there is some $c$ such that for all $\tau \in [B]$,
$\P[\mathfrak{X}^P_{M, \tau} \le c \log(n)] = 1 - o(1/n)$. We take $C(n) = c \log(n)$. 

Now we consider our branching process, $\mathfrak{X}^D_{\vec{p},
  i}$. To make the comparison, we will need some minimal
concentration properties. We first show that the sampled
competencies $\vec{p}$ satisfy these properties with high probability, and then show that, conditioned on these properties, the branching process
$\mathfrak{X}^D_{\vec{p}, i}$ is easily comparable to a Poisson
process. The properties are the following:
\begin{enumerate}
    \item For each voter $i \in [n]$, $\sum_{j \ne i} \p(p_i, p_j) \ge (1 - \e) \cdot n$.
    \item For each type $\tau \in [B]$, the number of voters of type $\tau$, $|\set{i \suchthat p_i \in S_\tau}| \le (1 + \e) \pi_\tau n$.
\end{enumerate}

For the first property, fix the competence $p_i$ of
a single voter $i$. Then when sampling the $p_j$s, $\sum_{j \ne i}
\p(p_i, p_j)$ is the sum $n - 1$ independent variables, all in the interval $[L, U]$, with mean $1$. Hence, by Hoeffding's inequality, for all
competencies $c$,
$\D^n[\sum_{j \ne i} \p(p_i, p_j) \ge (1 - \e)n \mid p_i = c] = 1 -
o(1/n)$, where the $o(1/n)$ term is independent of $c$. By the law of total
probability, this implies that even when $p_i$ is sampled as well, the
$1 - o(1/n)$ bound continues to hold. By a union bound over all $n$
voters, this holds for everybody with probability $1 - o(1)$. 

For the second property, note that the number of voters of type $\tau$
follows a $\bin(n, \pi_\tau)$ distribution. A simple application of 
Hoeffding's inequality implies that for this $\tau$,  $|\set{i \suchthat
  p_i \in S_\tau}| \le (1 + \e) \pi_\tau n$ (note that this holds even
in the extreme cases where $\pi_\tau = 0$ or $\pi_\tau = 1$). As the
number $B$ of types is fixed and independent of $n$, a union bound
over all $B$ types implies this holds for all $\tau$ with probability
$1 - o(1)$. 

Now fix some voter competencies $\vec{p}$ such that both properties hold.
We will first upper bound the probability a voter of type $\tau$ delegates to a voter of type $\tau'$. Hence, we can compare our branching process to one with these larger probabilities, and this will only dominate our original process.

To that end, since $|D_{t - 1}| = t - 1 \le t$ (recall that $D_{t-1}$
consists of the dead voters at time $t-1$), using the first
property, we have that for all $i \in [n]$, $$\sum_{j \in [n]
  \setminus(D_{t - 1} \cup \set{i})} \p(p_i, p_j) \ge (1 - \e) n - U
\cdot t.$$ Hence, as long as $t \le \e n/U$, $\sum_{j \in [n]
  \setminus(D_{t - 1} \cup \set{i})} \p(p_i, p_j) \ge (1 - 2\e)n$. 

Including the fact that $\p(p_i, p_j) \le \tilde{\p}(p_i, p_j)$ for all $p_i$ and $p_j$, we have that for all time steps $t\leq \e n/U,$
$$
    p \cdot \frac{\p(p_i, p_j)}{\sum_{k' \in [n] \setminus (D_{t - 1} \cup \set{k})} \p(p_i, p_{k'})} \leq \frac{p}{n} \cdot \frac{\Tilde{\p}(p_i, p_j)}{(1-2\e)}.
$$
Note that for sufficiently large $n$, $C(n) \le \e n/U$, so from now
on we restrict ourselves to such $n$. 

Further, note that by the second property, there will never be more than
$(1 + \e)\pi_{\tau} n$ neutral voters of type $\tau$. Hence, if we
take a voter of type $\tau'$ at time step $t \le C(n)$, the number of
children it will have of type $\tau$ will be stochastically dominated
by a $\bin((1 + \e)\pi_{\tau} n, \frac{p}{n} \cdot
\frac{\Tilde{\p}(p_i, p_j)}{(1-2\e)})$, and this is independent for
each $\tau$. As $n$ grows large, this distribution approaches a
$\pois(p\frac{(1+\e)}{1-2\e}\Tilde{\p}(\tau, \tau'))$. In particular,
this means that for sufficiently large $n$, it will be stochastically
dominated by a $\pois(p\frac{(1+\e)^2}{1-2\e}\Tilde{\p}(\tau, \tau'))$
distribution (note the extra $(1+\e)$ factor). Hence, if voter $i$ is
of type $\tau$, up to time $t \le C(n)$, $\mathfrak{X}^D_{\vec{p}, i}$
is dominated by $\mathfrak{X}^P_{M, \tau}$, so 
$$\P_{\D,M^S_{p,\p}.n}[\mathfrak{X}^D_{\vec{p}, i} \ge C(n)] \ge
\P_{\D,M^S_{p,\p}.n}[\mathfrak{X}^P_{M, \tau} \ge C(n)] = 1 -
o(1/n).$$ 
A union bound over all $n$ voters tells us this is true for all voters
simultaneously with probability $1 - o(1)$, as needed. 

\medskip

\emph{The Continuous General Delegation Model satisfies \eqref{constr:2}.} 
\smallskip

To show \eqref{constr:2} holds, we first show the following.

Let $\mu_\D$ be the mean of the competence distribution $\D$.
For a fixed $x$, let $\p^+_x(y)$ be the function $\p(x,y)\cdot y$.
We show that there
is some $c > 0$ such that for all $x \in [0, 1]$, 
\begin{equation}
    \label{eq:exp-phi}
        \E_{\D}[\p^+_x] \ge \mu_\D + c.
\end{equation}

Indeed, if we view $\E_{\D}[\p^+_x]$ as a
function of $x$ for $x \in [0, 1]$, first note that it is a
continuous function on a compact set, and hence it attains its
minimum. Further, for all $x \in [0, 1]$, since $\p(x, y)$ and $y$ are
both increasing functions of $y$, by the FKG inequality
\cite{fortuin1971correlation}, 

$$\E_{\D}[\p^+] > \E_{\D}[\p(x, \cdot)] \cdot \mu_\D = \mu_D,$$
since, by assumption, $\E_{\D}[\p(x, \cdot)] = 1$.
Hence, this attained minimum must be strictly larger than $\mu$,
implying \eqref{eq:exp-phi}. 

Since $\p(x, y)$ is normalized so that $\E_{y \sim \D}[\p(x, y)] =
1$, $\E_{y \sim \D}[\p(x, y) \cdot y]$ is the expected competence of
the voter to whom someone of competence $x$ delegates to (prior to
other competencies being drawn). Hence, \eqref{eq:exp-phi} tells us that
``on average'', all voters (regardless of competence) tend to delegate
to those with competence strictly above the mean. Ideally, we would
choose $\a \approx c/2$ and hope that some concentration result tells
us that the weighted 
competencies after
delegation will be strictly
above $\mu + c/2$ (the mean of all competencies will be close to $\mu$
by standard concentration results). However, proving this
concentration result is surprisingly subtle, as there are many dependencies
between different voter delegations. Indeed, if one voter with high
competence and many delegations chooses to delegate ``downwards''
(that is, to someone with very low competence), this can cancel out
all of the ``expected'' progress we had made thus far. Hence, the rest
of this proof involves proving concentration \emph{does} in fact
hold. We prove this by breaking up the process of sampling
instances into much more manageable pieces, where, in each, as long as
nothing goes ``too'' wrong, concentration will hold. 

In particular, we will prove that for all $\g > 0$, with high probability,
\begin{equation}
    \label{eq:gamma-gain}
    \sum_{i = 1}^n \weight_i(G) \cdot p_i - \sum_{i = 1}^n p_i \ge (c(1 - p) - \g) n.
\end{equation}
Fix such a $\g$.
As in the previous part, fix $\e > 0$ which will paramaterize our steps. We will later choose $\e$ sufficiently small to get our desired result (precisely $\e$ such that $6\e + \e^2 < \g$). By choosing $\g < c(1 - p)$, this value is positive, so we can choose $\a = \frac{c(1-p) - \g}{2}$ which proves \Cref{constr:2}

To that end, we define a sequence of six sampling steps that together
are equivalent to the standard
sampling process with respect to $\D$ and $M^S_{p, \p}$. In each step,
we will show that with high probability, nothing ``goes wrong'', and
conditioned on nothing going wrong in all these steps, we will get the $\alpha$ improvement that  we desire. The six steps are as follows: 
\begin{enumerate}
    \item Sample a set $M \subseteq [n]$ of voters that choose not to delegate. Each voter is included independently with probability $p$.\label{step:M}
    \item Sample competencies $p_i$ for $i \in [n] \setminus M$. Each $p_i$ is sampled i.i.d.\@ from $\D$.\label{step:del-pis}
    \item Sample competencies $p_j$ for $j \in M$. Each $p_j$ is sampled i.i.d.\@ from $\D$.\label{step:not-del-pis}
    \item Sample a set $R \subseteq [n] \setminus M$ of delegators
      that delegate to those in $M$. Each voter $i \in [n] \setminus
      M$ is included independently with probability $\frac{\sum_{j \in
          M} \p(p_i, p_j)}{\sum_{j \in [n] \setminus \set{i} \p(p_i,
          p_j)}}$, that is, the total $\p$ weight they put on voters
      in $M$ divided by the total $\p$ weight they put on all
      voters.\label{step:L} 
    \item Sample delegations of voters in $[n] \setminus (M \cup R)$. At this point, we are conditioning on such voters delegating, and when they do delegate, they do so to voters in $[n] \setminus M$. Hence, for each $i \in [n] \setminus (M \cup R)$, they delegate to $j \in [n] \setminus (M \cup \set{i})$ with probability $\frac{\p(p_i, p_j)}{\sum_{j' \in [n] \setminus (M \cup \set{i})} \p(p_i, p_{j'})}$.\label{step:not-L-dels}
    \item Sample delegations of voters in $R$. At this point, we are conditioning on such voters delegating to those in $M$. Hence, for each $i \in R$, they choose to delegate to $j \in M$ with probability $\frac{\p(p_i, p_j)}{\sum_{j' \in M} \p(p_i, p_{j'})}$.\label{step:L-dels}
\end{enumerate}
We now analyze each step, describing what could ``go wrong''.
Let $\cE_1, \ldots, \cE_6$ be the events that nothing goes
wrong in each of the corresponding steps. We define these events formally
below. Our goal is
to show that $\P_{\D,M^S_{p,\p},n}[\cE_1 \cap \cdots \cap \cE_6] = 1 - o(1)$.

$\bullet$ Let $\cE_1$ be the event that $(p - \e) \cdot n \le |M| \le
(p + \e) \cdot n.$ Note that $M$ is the sum of $n$ independent Bernouilli random variables with success probability $p.$ It
follows directly from a union bound over both variants of Hoeffding's
inequality that $$\P_{\D,M^S_{p,\p},n}[(p - \e) \cdot n \le |M| \le (p + \e)
  \cdot n] = 1 - o(1).$$ 

$\bullet$  Let $\cE_2$ be the event that $\sum_{i \in [n] \setminus M}
p_i \le n(\mu + \e) (1 - p + \e).$ Note that $\sum_{i \in [n]
  \setminus M} p_i$ is the sum of $n-|M|$ i.i.d.\@ random variables
with mean $\mu.$ Conditioning on event $\cE_1,$ $|M|$ is lower bounded
by $n(p - \e)$, implying that $n - |M| \le n(1 - p + \e)$ as well. It
follows from \Cref{lem:hoeffding} that $$\P_{\D,M^S_{p,\p},n}\left[\sum_{i \in [n]
    \setminus M} p_i \le n(\mu + \e) (1 - p + \e) \ \middle|
  \  \cE_1\right] = 1 - o(1)$$ which, combined with $\P_{\D,M^S_{p,\p},n}[\cE_1] = 1 -
o(1)$, proves that $\cE_1 \cap \cE_2$ occurs with probability $1 -
o(1)$. 

$\bullet$ 
Let $\cE_3$ be the event consisting of all instances $(\vec{p},G)$
such that
$$\frac{\sum_{j \in M} \p(p_i, p_j) \cdot p_j}{\sum_{j \in M} \p(p_i, p_j)} \ge \frac{(1 - \e)}{(1 + \e)} (\mu + c)$$
for all $i \in [n] \setminus M$.

We show $\cE_3$ occurs with high probability conditional on $\cE_1$
and $\cE_1$ (conditioning on $\cE_2$ is unnecessary. but
makes the final statement easier). Fix a set of voters $M$ and $p_i$
for $i \in [n] \setminus M$ satisfying $\cE_1$ and $\cE_2$. For each $i \in [n]
\setminus M$, we will show that with probability $1 - o(1/n)$, when we
sample the $p_j$s for $j \in M$, they satisfy 
\begin{equation}
    \sum_{j \in M} \p(p_i, p_j) \le |M|(1 + \e)
    \label{eq:exp_rand_den}
\end{equation}
and
\begin{equation}
    \sum_{j \in M} \p(p_i, p_j) \cdot p_j \ge |M|(1 - \e)(\mu + c).
    \label{eq:exp_rand_den2}
\end{equation}

\eqref{eq:exp_rand_den} follows from the fact that $\sum_{j \in M}
\p(p_i, p_j)$ is the sum of $|M|$ bounded independent random variables

with mean $\E_{y\sim\D}[\p(p_i, y)] = 1.$ By Hoeffding's inequlaity,
since $|M|$ is linear in $n$, $\sum_{j \in M} \p(p_i, p_j)$ is at most
$|M|(1 + \e)$ with probability $1 - o(1/n)$. 

\eqref{eq:exp_rand_den2} follows from the fact that $\sum_{j
\in M} \p(p_i, p_j) \cdot p_j$ is also the sum of $|M|$ bounded independent random variables with mean $\E_{\D}[\p^+_{p_i}]$.
Again, since we have conditioned on $\cE_1$, $|M|$ is lower
bounded by $(p-\e)n$, which by Hoeffding's inequality implies that
$\sum_{j \in M} \p(p_i, p_j)p_j$ is at least
$|M|(1-\e)\E_{\D}[\p_{p_i}]$ with probability $1-o(1/n)$. 

Finally, we can conclude via a union bound that $\frac{\sum_{j \in M}
  \p(p_i, p_j) \cdot p_j}{\sum_{j \in M} \p(p_i, p_j)} \ge
\frac{(1-\e)}{(1+\e)}(\mu + c)$ with probability $1-o(1/n)$ for any $i
\in [n] \setminus M$. Hence, by another union bound over the at most
$n$ voters $i \in [n] \setminus M$, $\frac{\sum_{j \in M} \p(p_i, p_j)
  \cdot p_j}{\sum_{j \in M} \p(p_i, p_j)} \ge
\frac{(1-\e)}{(1+\e)}(\mu + c)$ for all $i \in [n] \setminus M$ with
high probability. 

By the law of total probability, $\cE_3$ conditioned on $\cE_1$ and $\cE_2$ occurs with probability $1-o(1),$ which proves that $\cE_1 \cap \cE_2 \cap \cE_3$ occurs with probability $1 - o(1)$ by the chain rule.

$\bullet$  Let $\cE_4$ be the entire sample space. Nothing  can ``go wrong'' during this sampling step. So trivially, $\cE_1 \cap \cE_2 \cap \cE_3 \cap \cE_4$ occurs with probability $1 - o(1)$. 

$\bullet$ Let $\cE_5$ be the event that $\dels_i(G) \le C(n)$ for all
$i \in [n] \setminus M$ and $\totweight(G) \ge n - C(n)^2\log(n)$ in
the subgraph $G$ sampled (i.e., with delegations only from voters not in $R$ or $M$). We will show $\cE_5$ occurs with high probability
even when we sample a full delegation graph (that is, samples
delegations for all voters), which implies it continues to hold even
when we sample only some delegations (recall that at this step we have
only sampled delegations from voters in $[n] \setminus (M \cup R)$). 

The proof of this is very similar to the one in \Cref{thm:CBM}, with one extra step to allow for different $\p$ weights.

It was proved in the previous part of this proof that, for all voters
$i$, we have that  $\dels_i(G)
\le C(n)$ with probability $1-o(1)$ (not conditioned on anything) when
we sample entire delegation graphs,
so we can safely condition on this fact.
We now prove that $\P_{\D,M^S_{p,\p},n}[\totweight(G_n) \ge n -
  O(\log^3 n) \suchthat \dels_i(G) \le C(n)] = 1 - o(1)$.

We 
begin by bounding the number of voters that end up in cycles. Fix some
voter $i$, and let us begin by sampling their delegation tree.

Since we are conditioning on the tree having size at most  $C(n)$,
the most 
weight that voter $i$ can place on all of the voters in $i$'s delegation
tree is $U \cdot C(n)$. The minimum weight that $i$ can place on all voters is $L(n - 1)$. Hence, the probability that $i$ delegates to someone
in $i$'s tree conditional the delegation tree having size at most
$C(n)$ is at most $p \cdot \frac{U \cdot
  C(n)}{L \cdot (n - 1)}$.  
Since $i$ was arbitrary, this implies that the expected number of
voters in cycles can be at most $n \cdot p \cdot \frac{U \cdot C(n)}{L
  \cdot (n - 1)} \in O(\log n)$.

Applying Markov's inequality just as in the analogous proof in the
previous section, the
probability that more than $\log^2 n$ voters are in cycles is at most
$np\frac{UC(n)}{L(n - 1)\log^2 n} = O(1/\log n) = o(1)$. Further, the
total number of people that could delegate to  voters in cycles is at most $C(n)$  times the number of voters in cycles.
Hence, with probability $1 - o(1)$, there are at
most $C(n) \cdot \log^2 n$ voters delegating to those in cycles. This
implies the desired bound. Hence, we have proved that 
$
\P_{\D,M^S_{p,\p},n}[\cE_5] = 1 - o(1)
$. Since we have already shown that $\P_{\D,M^S_{p,\p},n}[\cE_1 \cap
  \cE_2 \cap \cE_3 \cap \cE_4] = 1 - o(1)$, a union bound implies that
$\cE_1 \cap \cE_2 \cap \cE_3 \cap \cE_4 \cap \cE_5$ occurs with
probability $1 - o(1)$ as well. 

$\bullet$ We now consider the sixth step. To define $\cE_6$, we
need some new notation. Fix competencies $\vec{p}$ and a partial
delegation graph $G$ such that $(\vec{p},G)$ is in the first five events. 
We define $Q_i$ for $i \in R$ to be the random variable representing
the competence of the voter to whom $i$ delegates. Since we know $i$
delegates to a voter in $M$, note that  
$$
Q_i(G) = p_j \text{ with probability }  \frac{\p(p_i, p_j)}{\sum_{j' \in
M} \p(p_i, p_{j'})} \text{ for all } j\in M. 
$$

Let $\cE_6$ be the event consisting of all instance $(\vec{p},G)$ such that
that $\sum_{i \in R} \dels_i(G) \cdot Q_i(G) \ge
\frac{(1 - \e)^2}{1 + \e} (\mu + c) (1 - p - 2\e) \cdot n$. 
We show that $\P_{\D,M^S_{p,\p},n}[\cE_6 \suchthat \cE_1 \cap
  \cdots \cap \cE_5] = 
1 - o(1)$. This, combined with the the fact $\P_{\D,M^S_{p,\p},n}[\cE_1 \cap
  \cdots \cap \cE_5] = 1 - o(1)$ (shown earlier), implies that
$\P_{\D,M^S_{p,\p},n}[\cE_1 \cap \cdots 
  \cap \cE_6] = 1 - o(1)$. 
It follows from the definition of $Q_i$ that
$$\E[Q_i] = \sum_{j \in M} \frac{\p(p_i, p_j)}{\sum_{j' \in M} \p(p_i, p_{j'})} \cdot p_j =  \frac{\sum_{j \in M} \p(p_i, p_j) \cdot p_j}{\sum_{j \in M} \p(p_i, p_j)}. $$

By conditioning on $\cE_3$, we have that $\E[Q_i] \ge \frac{(1 -
  \e)}{(1 + \e)} (\mu + c)$ for each $i \in R$. 
Hence, $\E[\sum_{i \in R} \dels_i(G) \cdot Q_i] \ge (n - |M| - C(n)^2\log(n)) \cdot \frac{1 + \e}{1 - \e} \cdot (\mu + c)$, since we
are conditioning on
$\cE_3$ and $\cE_5$. Further, for sufficiently large $n$,
$C(n)^2\log(n) \le \e n$; since we are conditioning on $\cE_1$, $|M|
\le (p + \e)n$, so we 
have that for sufficiently large $n$, 
$$\E[\sum_{i \in R} \dels_i(G) \cdot Q_i] \ge (1 - p - 2\e) \cdot
\frac{1 + \e}{1 - \e} \cdot (\mu + c) \cdot n \in \Omega(n).$$ 
Next, consider $\Var[\sum_{i \in R} \dels_i(G) \cdot
  Q_i]$. Since each $Q_i$ takes on values in $[0, 1]$, $\Var[Q_i] \le
1$. Further, each summand is independent, as each $Q_i$ is independent and we have fixed $G$, so we can can view $\dels_i(G)$ as a
constant. Hence, $\Var[\sum_{i \in R} \dels_i(G) \cdot Q_i] \le \sum_{i \in R} \dels_i(G)^2 \in o(n^2)$ since, for all $i$,  $\dels_i(G) \le
C(n) \in O(\log n)$ and $\sum_i \dels_i(G) \le n$. Hence, 
\begin{align*}
    &\phantom{\le} \P_{\D,M^S_{p,\p},n}[\sum_{i \in R} \dels_i(G) \cdot Q_i < \frac{(1 - \e)^2}{1 + \e} (\mu + c) (1 - p - 2\e) \cdot n]\\
    &\le  \P_{\D,M^S_{p,\p},n}[\sum_{i \in R} \dels_i(G) \cdot Q_i < (1 - \e)\E[\sum_{i \in R} \dels_i(G) \cdot Q_i]]\\
    &\le \frac{\Var[\sum_{i \in R} \dels_i(G) \cdot Q_i]}{\e^2 \cdot \E[\sum_{i \in R} \dels_i(G) \cdot Q_i]^2} \in o(1)
\end{align*}
where the second inequality is due to Chebyshev's inequality, which is $o(1)$ because the numerator is $o(n^2)$ and the denominator is $\Omega(n^2)$. This implies the desired result.

Finally, we show that for all instance $(\vec{p},G) \in \cE_1 \cap
\cdots \cap \cE_6$, 
\eqref{constr:2} holds, and hence so does \eqref{eq:gamma-gain}. We
have that $\sum_{i = 1}^n w_i(G) \cdot p_i = \sum_{i \in R} \dels_i(G) \cdot Q_i(G) +
  \sum_{j \in M} p_j$, because in $G$ each voter $i
\in R$ delegates all of their $\dels_i(G)$ votes to the voter in
$M$ with competence $Q_i(G)$. Hence, $\sum_{i = 1}^n w_i(G) \cdot p_i -
\sum_{i = 1}^n p_i = \sum_{i \in R} \dels_i(G) \cdot Q_i(G) - \sum_{i \in [n] \setminus (M \cup R)} p_i$. Since $(\vec{p},G) \in \cE_2$, we
have that $\sum_{i 
  \in [n] \setminus M} p_i \le n(\mu + \e) (1 - p + \e)$. Since
$(\vec{p},G) \in \cE_6$,
we have that $\sum_{i \in R} \dels_i(G) \cdot Q_i(G) \ge \frac{(1 -
  \e)^2}{1 + \e} (\mu + c) (1 - p - 2\e) \cdot n$. Hence, this
difference is at least 
\begin{align*}
    ((\mu + c) (1 - p - 2\e) - (\mu + \e) (1 - p + \e)) n
    &\ge (c( 1-p) -3\e \mu - 2\e c - (1 - p) \e - \e^2) n\\
    &\ge (c(1 - p) - 6 \e - \e^2)n
\end{align*}
where the second inequality holds because, $c, (1 - p), \mu \le 1$. By
choosing $\e$ such that $6\e + \e^2 \le \g$ ($\e = \min(\g / 7, 1)$
will do), \eqref{eq:gamma-gain} follows. 
\medskip

\emph{The Continuous General Delegation Model Satisfies \eqref{constr:3}}
\smallskip

We now show that there exists a distribution $\D$ and $\a > 0$
such that $\sum_{i = 1}^n  p_i + \a n \le n/2 \le \sum_{i = 1}^n
\weight_i(G_n) \cdot p_i -  \a  n$ with probability $1 - o(1).$ This
implies that the model $M^S_{p,\p},n$ satisfies probabilistic positive gain
by \Cref{lem:core}.

As in earlier arguments, let $\D_\eta = \U[0, 1-2\eta]$ for $\eta \in
[0, 1/2)$. Note that $$f(\eta) = \inf_{x \in [0, 1]} \set{\E_{ D_{\eta}}[\p^+_x]} \cdot (1 - p) - 3\eta / 2$$ is a continuous
  function of $\eta$. 

  Moreover,
$f(0) > 0$. Hence, for sufficiently small $\eta > 0$, $f(\eta) > 0$.

 Consider $\D_\eta$ for some $\eta > 0$ such that $f(\eta) > 0$.
 Let $\a = \min(\eta / 2, f(\eta)/2)$.
  Since $\mu_{\D_\eta} = 1/2 - \eta$, 
  by Hoeffding's inequality, $\sum_{i = 1}^n p_i \le (1/2 - \eta / 2) n
  \le n/2 - \a n$ with high probability. 

Next, note that we can choose $c = \inf_{x \in [0,
    1]}\set{\E_{D_{\eta}}[\p^+_x]}$ in order to satisfy 
  \eqref{eq:exp-phi}. Hence, by choosing $\g = f(\eta)  / 2$,
it follows from \eqref{eq:gamma-gain} that  
$$
\sum_{i = 1}^n \weight_i(G) \cdot p_i - \sum_{i = 1}^n p_i \ge (c(1 - p) - f(\eta) / 2)n = (3\eta / 2 + f(\eta) / 2)n \ge (3\eta / 2 + \a)n
$$
with high probability.
Further, by Hoeffding's inequality,
$\sum_{i = 1}^n p_i \ge (1/2 - 3\eta/2)n$ with high probability, so by
the union bound applied to these inequalities,
$$\sum_{i =1}^n \weight_i(G) \cdot p_i \ge n/2 + \a n$$
with high probability, as needed.
\newpage


\section{Groups Characteristics}\label{app:participation}



\begin{table}[ht]
  \caption{Qualitative group descriptions and sizes from regular experiment}
  \centering
  \begin{tabular}{clc}
    \toprule
    Group ID & Group Description & Group Size \\
    \midrule
    1 & Company Employees Present at a Workshop & 14 \\
    2 & Undergraduate Students Present in Class & 22 \\
    3 & Research Department Meeting & 19 \\
    4 & Company Employees Present at a Workshop & 27  \\
    5 & Participants at an Academic Conference & 36  \\
    6 & Participants at an Academic Conference & 50 \\
    \bottomrule
  \end{tabular}
  \label{tab:experiments_app}
\end{table}

\newpage

\section{Survey Materials and Flow}\label{app:survey}

This section describes the participant experience of the survey. 
\Cref{Fig:surveyflow} shows an example of the survey flow. The green boxes represent the pre and post-survey steps (providing informed consent, name, and optional background questions). In the first stage, participants performed tasks, deciding to either delegate (providing a name) or vote (answering the eight question). The upper red block exemplifies a task prompt (in which the options ``delegate'' and ``vote'' also appear in a random order). In the second stage, participants answer additional questions (those they delegated) and optional background questions. \Cref{tab:tasks} shows all the prompts given to participants for each of the tasks.  Finally, \Cref{Fig:screenshot} shows screenshots of the survey from the participant's perspective.

\begin{figure}[htb]
  \begin{center}
      \includegraphics[width=\linewidth]{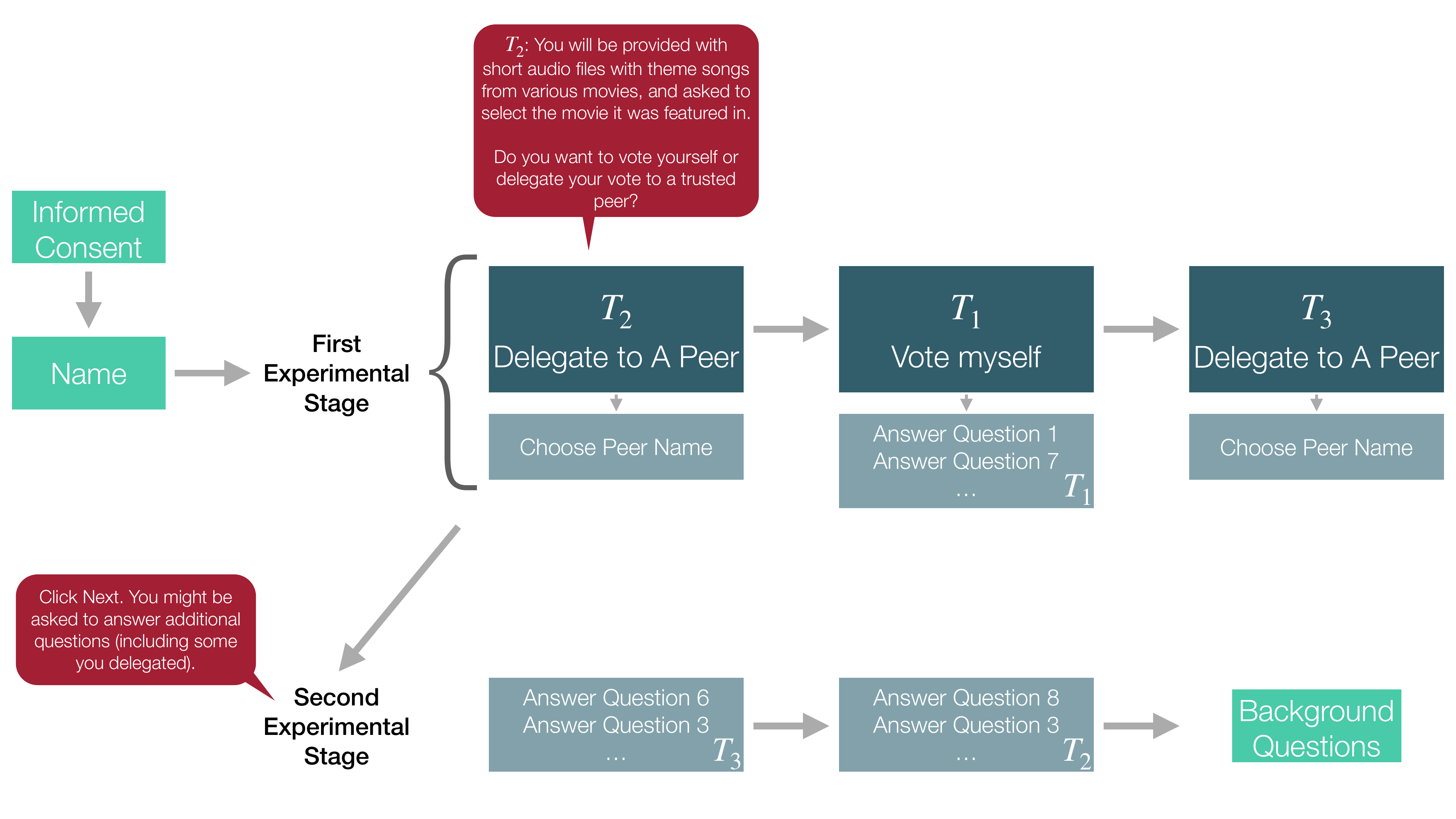}
  \end{center}
  
  \caption{Example of a survey flow with three tasks. }
  \label{Fig:surveyflow}
\end{figure}

\begin{table}[hbt]
  \caption{Prompts Presented at the beginning of each task}
  \small

  \centering
\begin{tabular}{p{0.5cm}p{12.5cm}p{2.6cm}}
\toprule
\textbf{ID} &\textbf{Task Prompts} & \textbf{Corresponding Experiment(s)} \\
\midrule
$T_1$& You will be shown images of architectural landmarks from around the world, and asked to select the country where the landmark is located. & $1, 2, 3, 4, 5, 6$ \\
$T_2$& You will be provided with short audio files with theme songs from various movies, and asked to select the movie it was featured in. & $1, 2, 3, 4, 5, 6$ \\
$T_3$& You will be given English idioms, and asked to identify their meaning. An idiom is a group of words that have a meaning not deducible from those of the individual words (e.g., rain cats and dogs, see the light). & $1, 2, 3, 4, 5, 6$ \\
\midrule
$T_4$& You will be given upcoming sport events (soccer and tennis games), and asked to predict the games' outcome? & $1$ \\
$T_5$& You will be given the names of tennis players, and asked to predict which round they will make it to in the Tennis French Open (Roland Garros), taking place in May-June 2022? & $2, 3, 4$ \\
$T_6$& You will be given the names of tennis players (women and men), and asked to predict which round they will make it to in the ongoing Wimbledon Tennis Tournament (The Championships, Wimbledon), taking place between June 27 and July 10, 2022. & $5$ \\
$T_7$& You will be given upcoming European men soccer games and asked to predict the games' outcome. & $6$ \\
\midrule
$T_8$& You will be shown images of flags from around the world, and asked to identify their country of origin. & $6$ \\

$T_9$& You will be shown 20 images of famous buildings from around the world, and asked to estimate the year in which the building was completed. & $6$ \\
$T_{10}$&You will be shown images of constellations and asked to identify them. & $6$ \\
$T_{11}$& You will be given headlines, and asked to identify the magazine that published the article, between The Economist and WIRED. & $6$ \\
$T_{12}$& You will be given words, and asked to identify the correct synonym corresponding to each word. & $6$ \\
$T_{13}$& You will be asked to listen to audio clips of classical compositions, and asked to identify the composer. & $6$ \\
$T_{14}$& You will be given the names of American states, and be asked to predict whether the majority of Congress members elected in that state will be Republican or Democrat. & $6$ \\
$T_{15}$& You will be given upcoming NBA games and asked to predict the games' outcome. & $6$ \\
\bottomrule
\end{tabular}
\label{tab:tasks}
\end{table}

\begin{figure}[htb]
\centering
\begin{subfigure}[b]{\textwidth}
\includegraphics[width=\linewidth]{Figures/Exp1.pdf}
\caption{Screenshot of a task prompt along with a choice to vote or delegate.}
\end{subfigure}

\vspace{1cm}

\begin{subfigure}[b]{\textwidth}
\includegraphics[width=\linewidth]{Figures/Exp2.pdf}
\caption{Screenshot of a specific question.}
\end{subfigure}

  \caption{Screenshots from the survey of both task prompts and specific questions.}
  \label{Fig:screenshot}
\end{figure}

The entire survey
can be found in the \href{https://osf.io/jnfcs}{OSF repository}.




\clearpage
\section{Item Response Rate Theory and Technical Details in Support of \Cref{comp}}\label{app:IRT}
In order to evaluate how delegation behavior relates to competence, we need to estimate participants' competence. We denote by $\eta_{i, t}$ the estimated competence of participant $i$ in task $t.$ A naive way to compute participants' competence per task would be to average their number of correct answers given on all $8$ questions of that task,  $\eta^{\text{naive}}_{i, t} = \frac{\sum_{r \in R_t} v_{i, r}}{|R_t|}.$ However, such a computation does not account for the varying difficulty of the questions.

Instead, we estimate $\eta_{i, t}$ using the Item Response Theory framework (IRT)~\citep{lalor2023py}, which provides a widely used parametric model to estimate competence and question difficulty from repeated measurements. To do this, we fit a one-parameter logistic model to estimate person $i$'s latent ability to 
do
task $t,$ $\eta_{i, t},$ as well as question $r$'s latent difficulty, $\theta_{r},$ where the probability that person $i$ is correct 
on
question $r$ depends on the person's ability at task $t$ and the question's difficulty.\footnote{We also fit a three-parameter logistic model estimating $ \Pr[v_{i, r}=1|\eta_{i, t},\theta_{r},c_{r},a_{r}] = c_{r} + \frac{1 - c_{r}}{1+\exp^{-a_{r}(\eta_{i, t} - \theta_{r})}},$ with $c_{r},$ the effect of guessing on question $r$ and $a_{r},$ the degree to which question $r'$ differentiate between participants. The resulting competence levels are highly correlated and we stick with the one-parameter model as a result.} Specifically,
we assume a generative process of
\begin{equation*}
    \Pr[v_{i, r}=1|\eta_{i, t},\theta_{r}] = \frac{1}{1+\exp^{-(\eta_{i, t} - \theta_{r})}}
\end{equation*}
and fit $\eta_{i, t}$ and $\theta_{r}$ to be consistent with the observed answers $v_{i, t}$. We fit these parameters in the canonical way using the Python package py-irt. See \citet{natesan2016bayesian} for more details.

While $\eta^{\text{naive}}_{i, t}$ takes on one of nine values (multiples of $1/8$), $\eta_{i, t}$ is a continuous variable that can take on arbitrary values in $\mathbb{R}.$ We normalize so that $\eta_{i, t} \in [0,1]$, and assume this to be the competence, the probability of being correct. Even using these different methods, the correlation between $\eta^{\text{naive}}_{i, t}$ and $\eta_{i, t}$ is above $94\%.$

\clearpage
\section{Delegation Graph Examples}\label{app:graph-examples}

\Cref{Fig:graphs_works,Fig:graphs_concentration,Fig:graphs_13g,Fig:graphs_13b} 
show representative examples of delegation graphs across different experiments and tasks. Nodes are labeled by the percentage of correct answers for that task, i.e., $\eta^{\text{naive}}_{i,t}.$ 
\Cref{Fig:graphs_works,Fig:graphs_13g} display examples of relatively little concentration of power with delegations reaching participants with relatively high competence, \Cref{Fig:graphs_concentration} displays some of the worst concentration of power observed across the $32$ instances and \Cref{Fig:graphs_13b} illustrates an example of relatively little concentration of power with delegations reaching participants with relatively low competence.

\begin{figure}[htb]
  \centering
  \includegraphics[width=.8\linewidth]{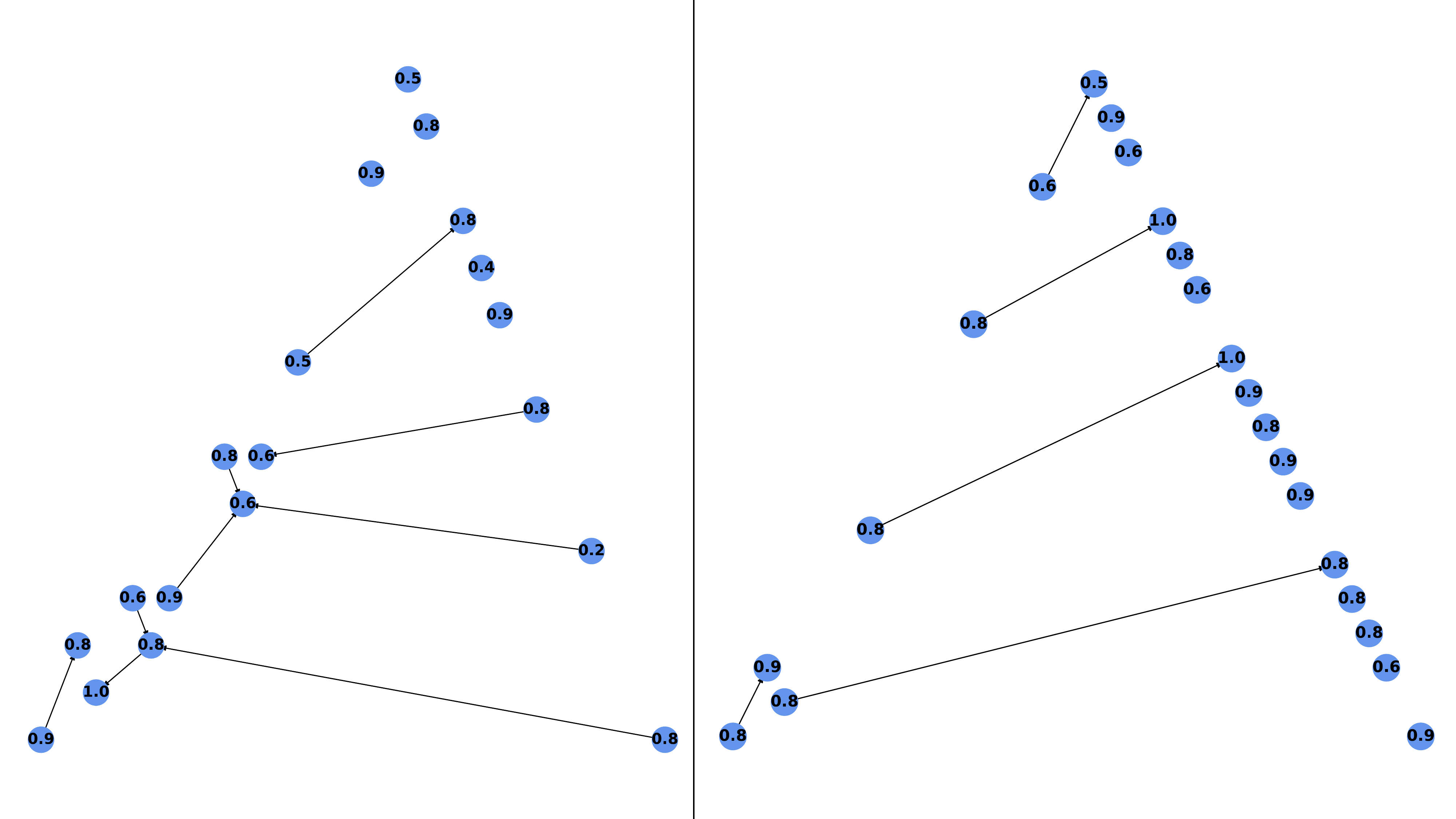}
  \caption{Delegation graphs for task $T_3$ from Experiment $2$ (left) and task $T_2$ from Experiment $3$ (right).}
  \medskip
  \small 
  \label{Fig:graphs_works}
\end{figure}

\begin{figure}[htb]
  \centering
  \includegraphics[width=0.4\linewidth]{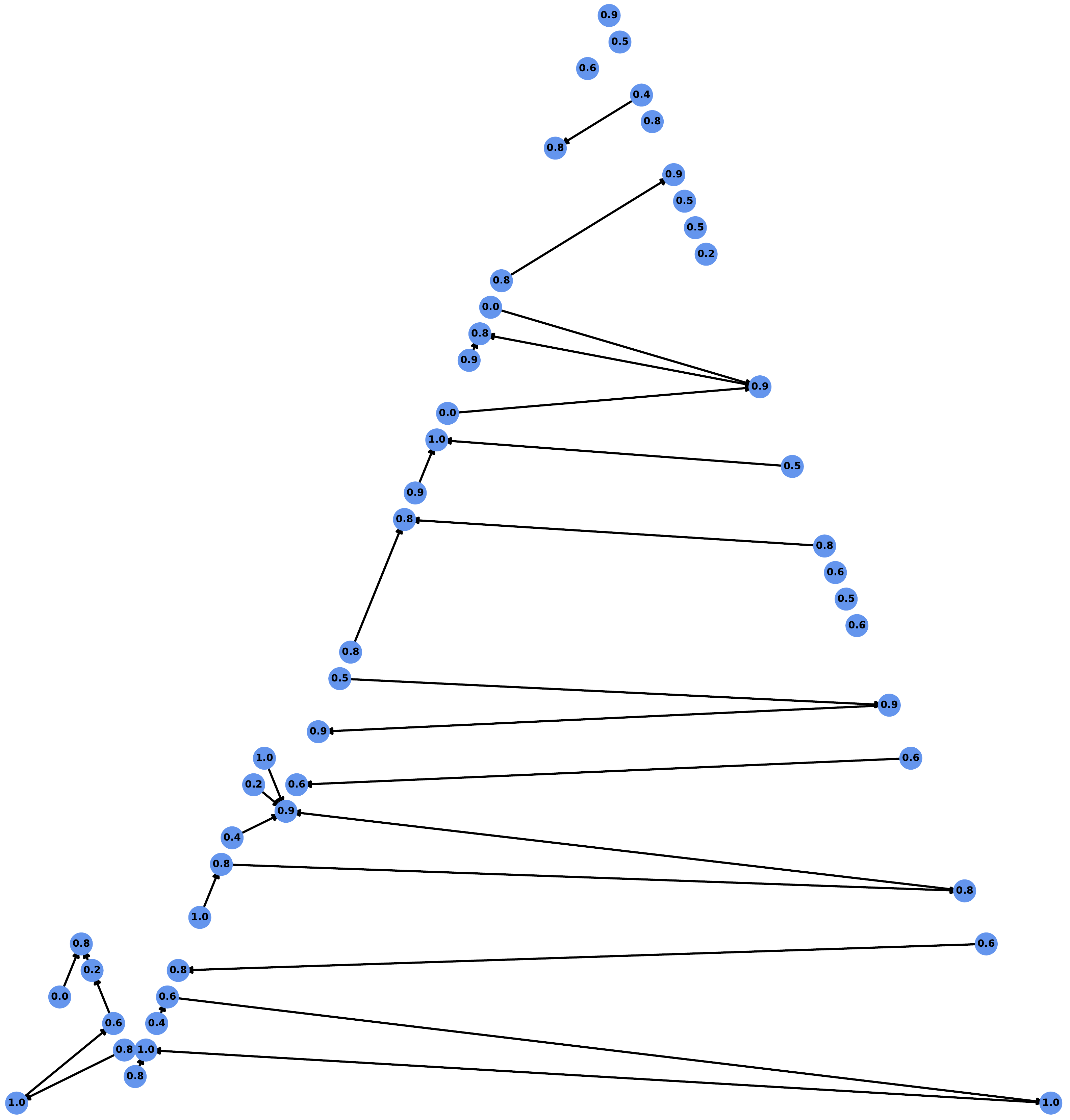}
  \caption{Delegation graph from experiment $6$ and task $T_7.$}
  \label{Fig:graphs_13g} 
\end{figure}

\begin{figure}[htb]
  \centering
  \includegraphics[width=.8\linewidth]{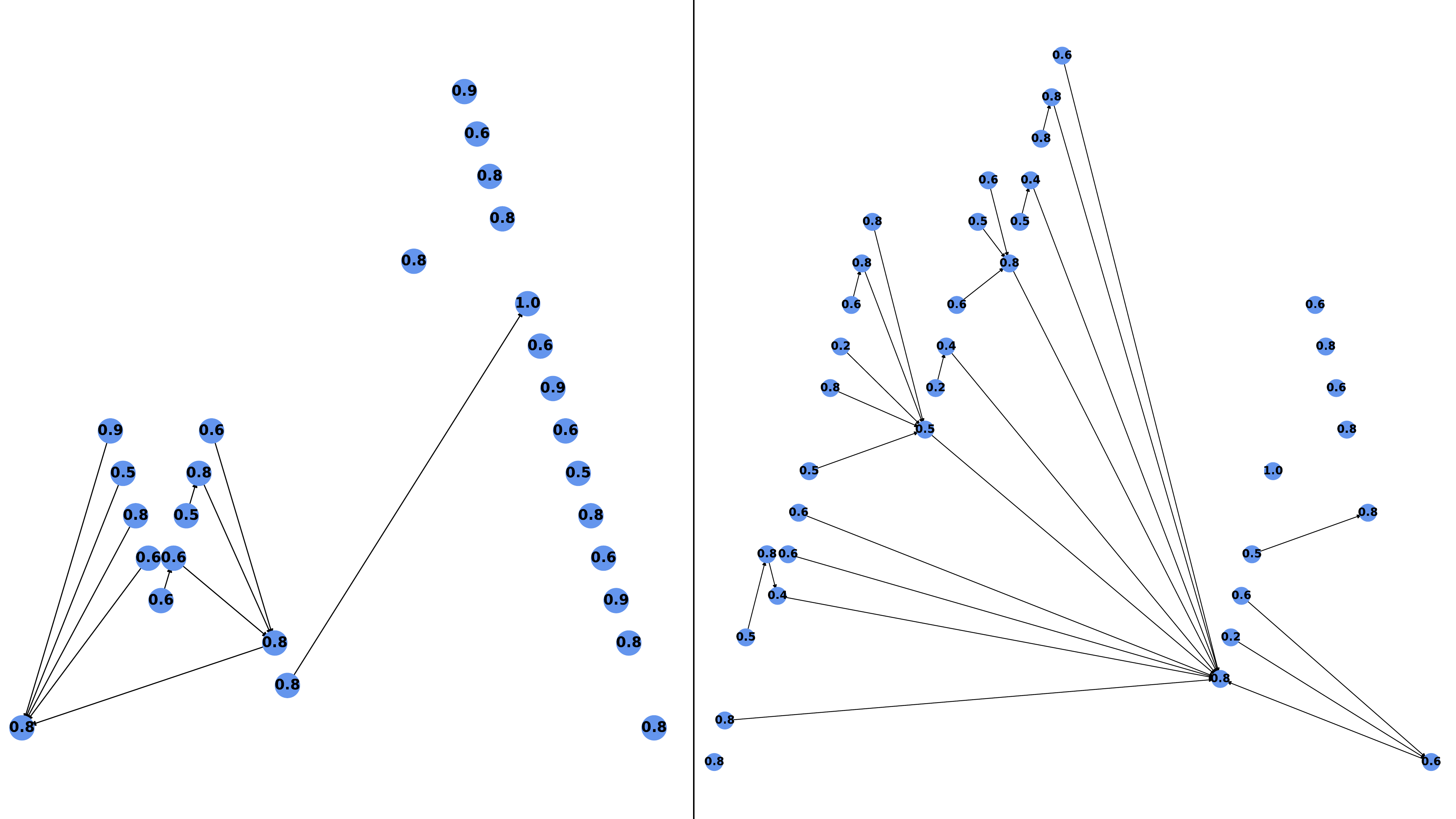}
  \caption{ Delegation graphs for task $T_1$ from Experiment $4$ (left) and task $T_5$ from Experiment $5$ (right). }
  \label{Fig:graphs_concentration}
\end{figure}




\begin{figure}[htb]
  \centering
  \includegraphics[width=0.4\linewidth]{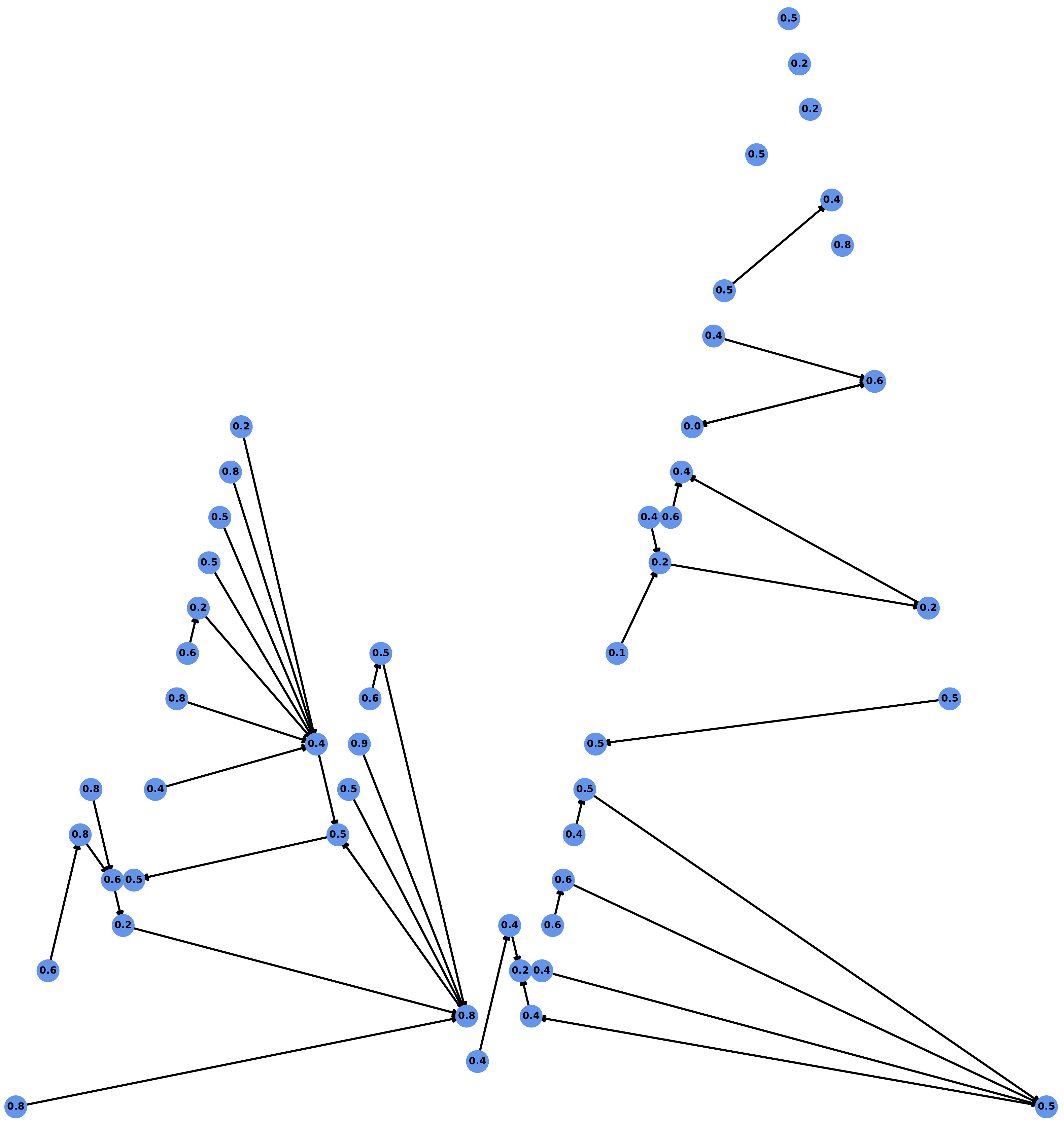}
  \caption{Delegation graph from experiment $6$ and task $T_{15}.$}
  \label{Fig:graphs_13b}
\end{figure}

\clearpage
\section{Pairwise Tukey Tests}\label{app:tukey}
\begin{figure}[h]
  \begin{center}
    \includegraphics[width=0.6\linewidth]{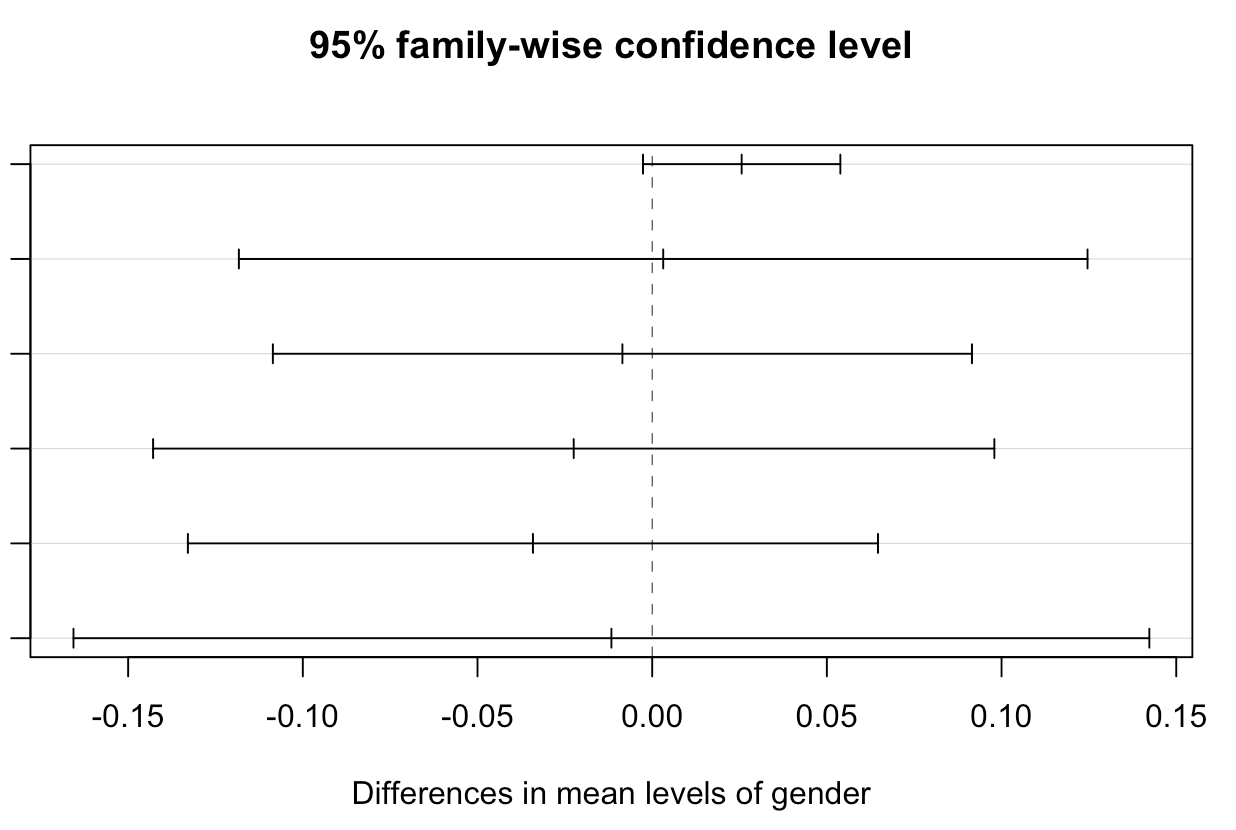}
  \includegraphics[width=0.6\linewidth]{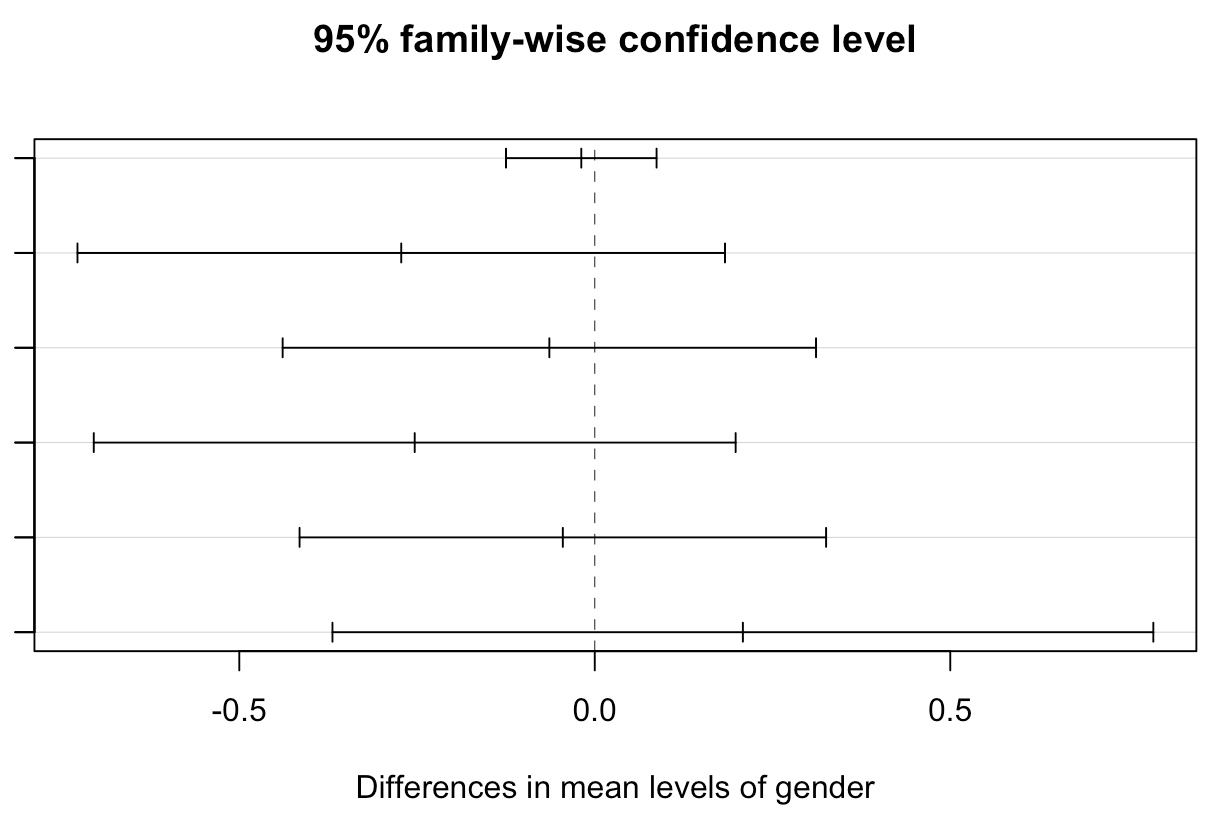}
  \end{center}
  \caption{Pairwise Tukey test across different gender regarding competence (top) and propensity to delegate (bottom). Pairwise test shows that competence and propensity to delegate is indistinguishable across gender.}
\end{figure}

\newpage
\section{Normality Assumptions for Regressions}\label{app:norm}
We check that the variables used in the different regressions are indeed normal, per the tests' assumptions.

\begin{figure}[H]
\begin{center}
      \includegraphics[width=0.4\linewidth]{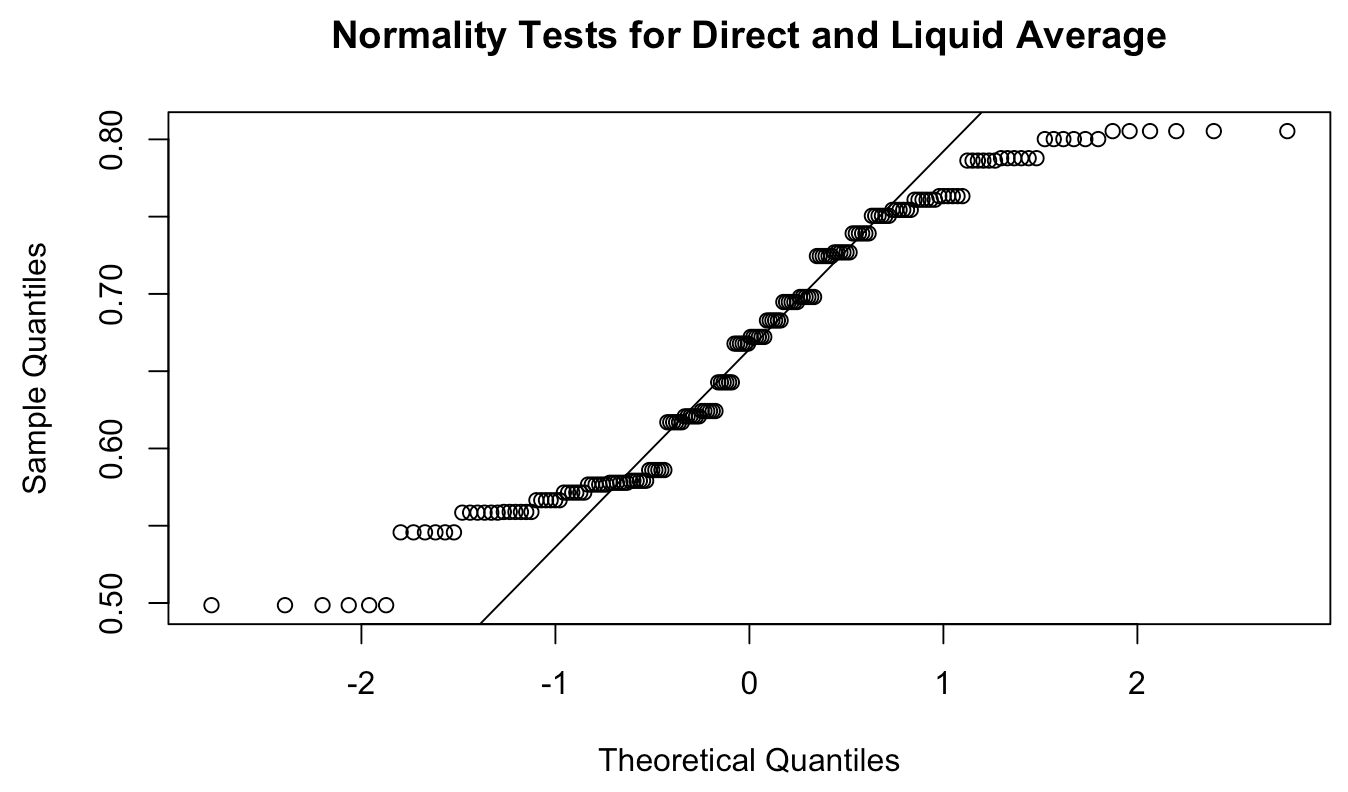}
  \includegraphics[width=0.4\linewidth]{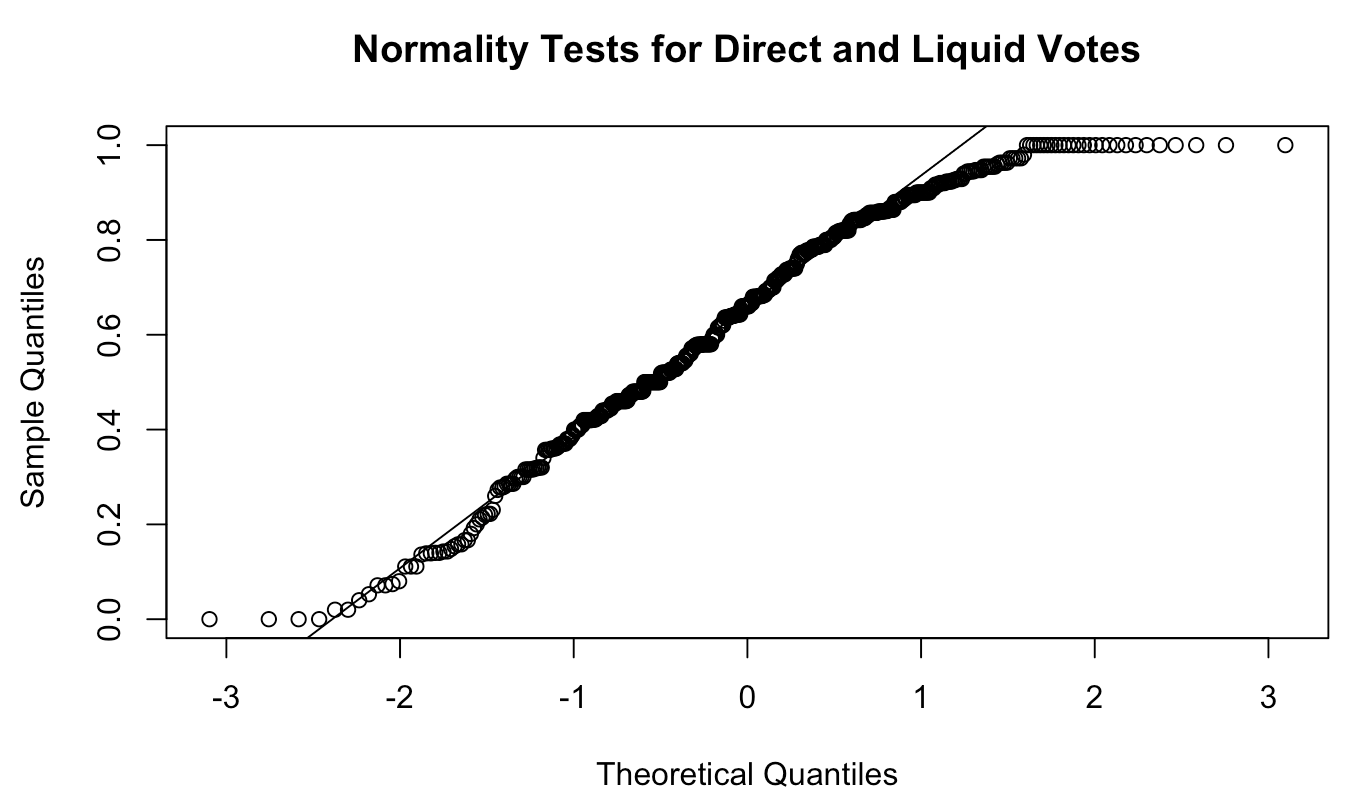}
  \includegraphics[width=0.4\linewidth]{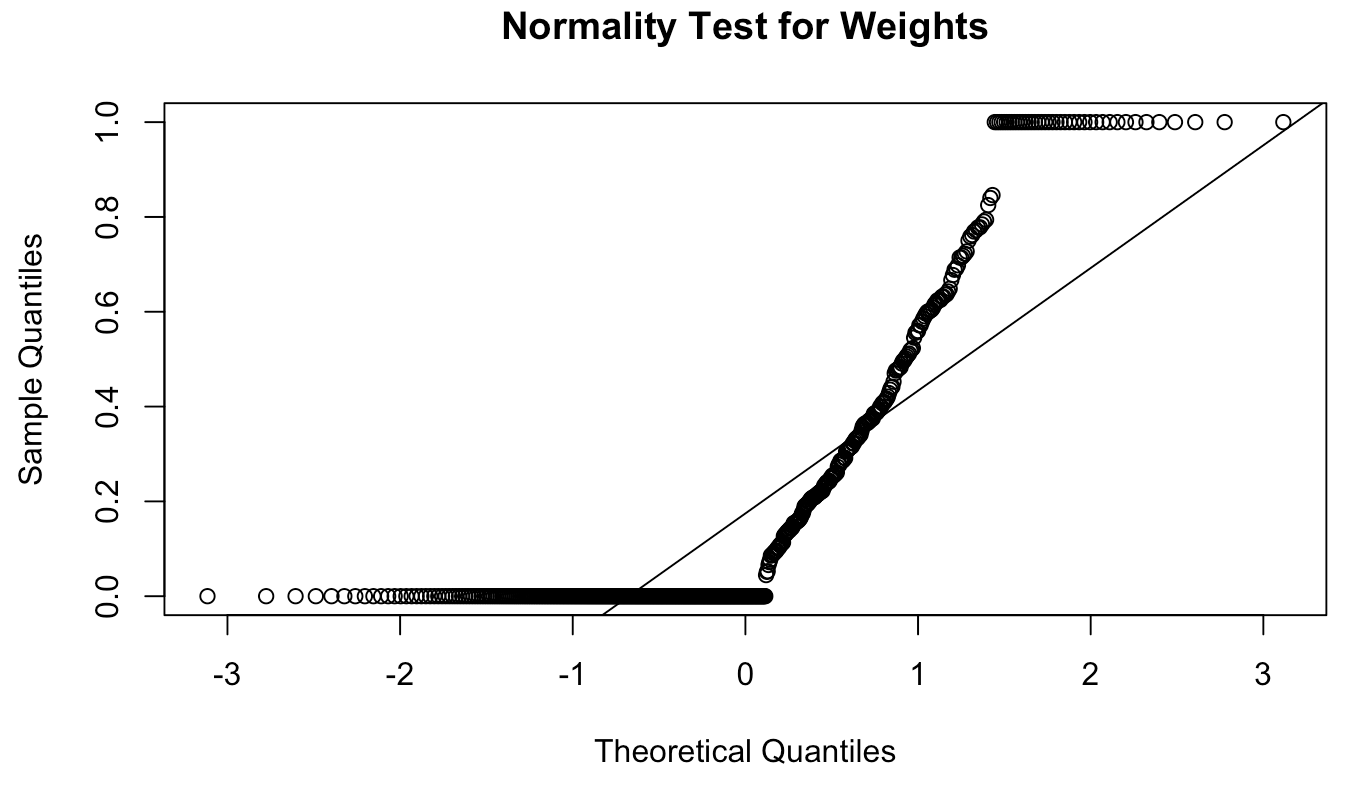}
  \includegraphics[width=0.4\linewidth]{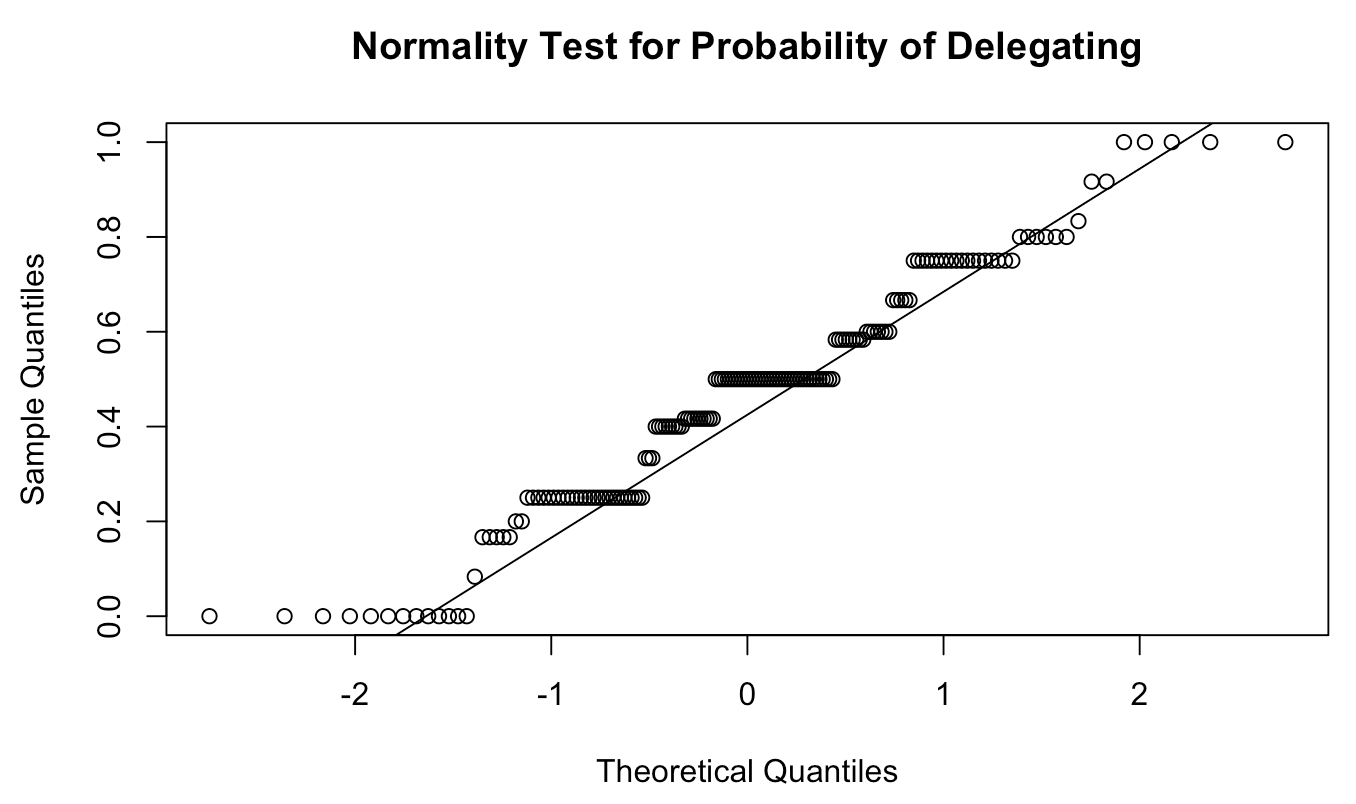}
  \includegraphics[width=0.4\linewidth]{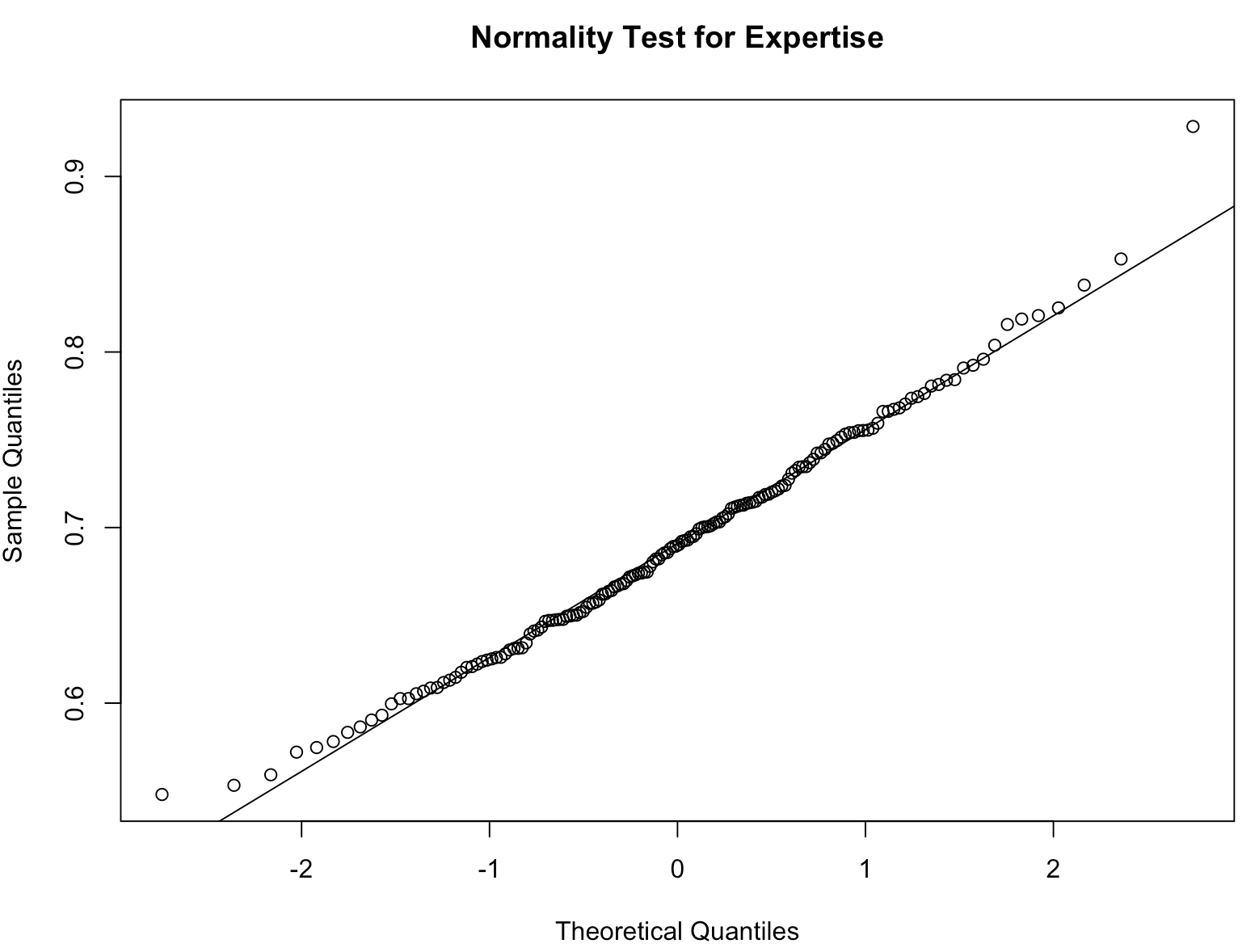}
\end{center}

  \caption{Normality Tests}
  \medskip
  \small From top to bottom and left to right, normality tests for the average competence in direct and liquid democracies, the average vote outcome in direct and liquid democracies, the estimated values of $\varphi$, the probability of delegating and the participants' competence.
\end{figure}

\clearpage
\section{Additional tests for \Cref{sec:method_q}}\label{app:q}

We run a few additional tests for understanding the relationship between the propensity to delegate and competence.
First, we repeat the same analysis we did in \Cref{sec:method_q}, but partitioning by task first. Results can be found in \Cref{table:resultsQ,tab:task-specific-6}. Note that tasks 8--15 only appeared in experiment 6, so fewer data points were available, resulting in much larger standard errors.

\begin{table}[htb]
\caption{Estimates of $\beta_q$ for tasks 1--7.}\label{table:resultsQ}

  \centering
{  \footnotesize\begin{adjustbox}{angle=0}\begin{tabular}{@{\extracolsep{5pt}}lcccccccc} 
\\[-1.8ex]\hline 
\hline \\[-1.8ex] 
& \multicolumn{1}{c@{}}{Overall Model} & 
\multicolumn{7}{c@{}}{Tasks Models}\\ 
\cmidrule(lr){2-2}
\cmidrule(lr){3-9}
& \multicolumn{1}{c@{}}{\textit{}}
& \multicolumn{7}{c@{}}{(Tasks)}
\\[1ex] 
& \mc{} & \mc{($T_1$)} & \mc{($T_2$)} & \mc{($T_3$)} & \mc{($T_4$)} & \mc{($T_5$)} & \mc{($T_6$)} & \mc{($T_7$)}\\ 
\midrule
Effect Size $\beta^q$  &  $-2.23^{****}$ & 
 $-2.05$ & $-1.96$ & $-4.06^{***}$ & $2.33$ & $0.21$ & $-8.04^{*}$ & $-1.50$ \\ 
  & (0.42) & 
  (1.58) & (1.50) & (1.25) & (4.05) & (1.57)& (4.88)& (0.23) \\ [1ex]
\midrule
  \begin{tabular}{@{}c@{}} Clustered S.E.\end{tabular} & \multicolumn{1}{c}{\textrm{i}}  & \multicolumn{1}{c}{\textrm{NA}} & \multicolumn{1}{c}{\textrm{NA}} &
  \multicolumn{1}{c}{\textrm{NA}} & \multicolumn{1}{c}{\textrm{NA}} &
  \multicolumn{1}{c}{\textrm{NA}}
   & \multicolumn{1}{c}{\textrm{NA}} &
  \multicolumn{1}{c}{\textrm{NA}}\\ 
\hline \\ [-1.8ex] 
\textit{Note:}  & \multicolumn{8}{r}{$^{*}$p$<$0.1; $^{**}$p$<$0.05; $^{***}$p$<$0.01} \\
\end{tabular}
\end{adjustbox}}

\end{table}

\begin{table}[htb]
\footnotesize
\caption{Estimates of $\beta_q$ for tasks 8--15. \label{tab:task-specific-6}}

  \centering
{\begin{tabular}{@{\extracolsep{5pt}}lcccccccc} 
\\[-1.8ex]\hline 
\hline \\[-1.8ex] 
& \multicolumn{6}{c@{}}{Tasks Models}\\ 
\cmidrule(lr){2-9}
& \multicolumn{1}{c@{}}{\textit{}}
& \multicolumn{5}{c@{}}{(Tasks)}
\\[1ex] 
& \mc{($T_8$)} & \mc{($T_9$)} & \mc{($T_{10}$)} & \mc{($T_{11}$)} & \mc{($T_{12}$)} & \mc{($T_{13}$)} & \mc{($T_{14}$)} & \mc{($T_{15}$)}\\ 
\midrule
Effect Size $\beta^q$  & $-5.13$ & $1.57$& $-3.87$& $-2.43$& $-2.50$ & $-3.62^{*}$& $3.70$& $-1.58$\\ 
  & (5.09) & (2.15) & (2.66) &  (2.38)&  (2.01)&  (1.88)&  (2.89)&  (1.57)\\ [1ex]

\midrule
  \begin{tabular}{@{}c@{}} Clustered S.E.\end{tabular} & \multicolumn{1}{c}{\textrm{NA}}  & \multicolumn{1}{c}{\textrm{NA}} & \multicolumn{1}{c}{\textrm{NA}} &
  \multicolumn{1}{c}{\textrm{NA}} & \multicolumn{1}{c}{\textrm{NA}} &
  \multicolumn{1}{c}{\textrm{NA}}
   & \multicolumn{1}{c}{\textrm{NA}} &
  \multicolumn{1}{c}{\textrm{NA}}\\ 
\hline \\ [-1.8ex] 
\textit{Note:}  & \multicolumn{8}{r}{$^{*}$p$<$0.1; $^{**}$p$<$0.05; $^{***}$p$<$0.01} \\
\end{tabular}}
\end{table}

Next, we consider adding fixed effects to the model.
This helps us understand the direct impact of competence, even when keeping fixed an individual (as their competence varies across tasks) and the task (which could intrinsically lead to different propensities of delegation). More formally, we fit the following generalized linear model
\begin{equation}
    \log\left(\frac{\Pr[\delta_{i, t}=1]}{1 - \Pr[\delta_{i, t}=1]}\right) = \alpha_0 + \alpha_i + \alpha_t + \beta^q \eta_{i, t} + \varepsilon_{i}
\label{equ:q1}
\end{equation}
where $\alpha_0$ is the population average across participants and tasks, $\alpha_i$ is the fixed effect for person $i$ and $\alpha_t$ is the fixed effect for task $t,$ and $\beta^q$ is the coefficient for competence. Results can be found in \Cref{tab:fixed-effects}.

\begin{table}[htb]
\caption{Estimates for $\beta^q$ accounting for fixed-effects.\label{tab:fixed-effects}}

  \centering
{  \footnotesize\begin{adjustbox}{angle=0}\begin{tabular}{@{\extracolsep{5pt}}lcccc} 
\\[-1.8ex]\hline 
\hline \\[-1.8ex] 
& \multicolumn{1}{c@{}}{Overall Model} & 
\multicolumn{3}{c@{}}{With Fixed Effects}\\ 
\cmidrule(lr){2-2}
\cmidrule(lr){3-5}
& \multicolumn{1}{c@{}}{\Cref{equ:q0}}
& \multicolumn{3}{c@{}}{\Cref{equ:q1}}
\\[1ex] 
& \mc{} & \mc{} & \mc{} & \mc{}\\ 
\midrule
Effect Size $\beta^q$  & $-2.23^{****}$ & 
 $-2.90^{****}$ & $-1.67^{****}$ & $-1.86^{****}$ \\ 
  & (0.42) & 
  (0.60) & (0.35) & (0.45)\\ [1ex]
\midrule
  \begin{tabular}{@{}c@{}}Fixed Effects \end{tabular} & \multicolumn{1}{c}{NA}  & \multicolumn{1}{c}{\textrm{i}} & \multicolumn{1}{c}{\textrm{t}} &
  \multicolumn{1}{c}{\textrm{t, i}}\\
\midrule
  \begin{tabular}{@{}c@{}} Clustered S.E.\end{tabular} & \multicolumn{1}{c}{\textrm{i}}  & \multicolumn{1}{c}{\textrm{i}} & \multicolumn{1}{c}{\textrm{i}} &
  \multicolumn{1}{c}{\textrm{i}}\\ 
\hline \\ [-1.8ex] 
\textit{Note:}  & \multicolumn{4}{r}{$^{****}$p$<$0.0001} \\
\end{tabular}\end{adjustbox}}

\end{table}

The results show that participant-specific characteristics (such as confidence) are at play in delegation decisions. We also see that within each task, the probability of delegating decreases more slowly with competence than in the model without 
fixed effect.
This implies that task-specific characteristics (such as difficulty) 
are
also at play in delegation decisions.





\clearpage
\section{Additional tests and details for \Cref{sec:method_phi}}\label{app:phi}

\subsection{Alternative Bucketing Strategies}\label{app:bucket-strategies}
We experiment with varying bucketing methods, to ensure that the results are robust. The strategies (other than $k$-means) are the following. An example illustrating the differences can be found in \Cref{Fig:buckets}.

\paragraph{Equal cut:}
With equal cut, we divide the $[0,1]$ line in $B$ buckets $c_\ell = [(\ell - 1)/B, \ell/B]$ for $\ell \in \{1, \ldots, B\}$ of equal size.
We vary the number of buckets $B$ from $3$ to $10$ to ensure robustness of the approach. When buckets are empty, we compute the weights on the existing types and re-normalize in the final stage. 

\paragraph{Quantile cut:} We cut the $[0,1]$ line in $B$ quantiles $c_{\ell}$ for $\ell \in [b]$ so that the number of competence values in each bucket is the same and we take the mean competence in the designated bucket to be the representative competence, $\eta_{\ell}.$ That is, $\eta_{\ell} = \frac{\sum_{i,t} \eta_{i, t}\mathbb{I}[\eta_{i, t} \in c_{\ell}]}{\sum_{i,t} \mathbb{I}[\eta_{i, t} \in c_{\ell}]}.$ In short, each competence $\eta_{i, t}$ is assigned to a bucket $\eta_{\ell}$ such that $\ell \leq \eta_{i, t} <l+1/b.$ We vary the number of buckets $B$ to ensure robustness of the approach (see Appendix~\ref{buckets-results}).

\paragraph{Gaussian Mixture Model:}
We assume that the competence level
$\eta_{i,t}$ 
is
drawn from $B$ Gaussian distributions that we intend to reconstruct. To do so, we maximize the log-likelihood $\log\Pr(\vec{\eta}) = \sum_{i,t}\sum_{k}^b \log\left(\pi(k)\mathcal{N}(\eta_{i,t}|\eta_k,\sigma^2)\right),$ where $\vec{\eta}$ is the vector of $\eta_{i,t},$ $\pi(k)$ is the probability of being in Gaussian $k,$ $\eta_k,\sigma^2$ are the mean and variance of the $k-$th Gaussian and $\mathcal{N}(x|\mu, s^2)$ denotes the probability density function of a Gaussian with mean $\mu$ and standard deviation $s$ evaluated at $x.$ This optimization cannot be solved in closed-form, and we use the Expectation-Maximization (EM) algorithm to estimate the Gaussian means $\eta_k,$ as well as the marginal probabilities $\Pr(k|\eta_{i,k}).$\footnote{See details here: \url{https://people.csail.mit.edu/rameshvs/content/gmm-em.pdf}.} Last, we find the number of Gaussian $B$ that maximizes the likelihood's cross-validation estimate. In turn, we obtain an assignment of competence $\eta_{i,t}$ to $B$ Gaussian and $B$ Gaussian's mean, that we denote $\eta_{\ell}$ for the $\ell-$th Gaussian. Each competence $\eta_{i, t}$ is assigned to a Gaussian $\eta_{\ell}$ based on $\Pr(k|\eta_{i,k}).$ 

\begin{figure}[htb]
  \centering
  \includegraphics[width=\linewidth]{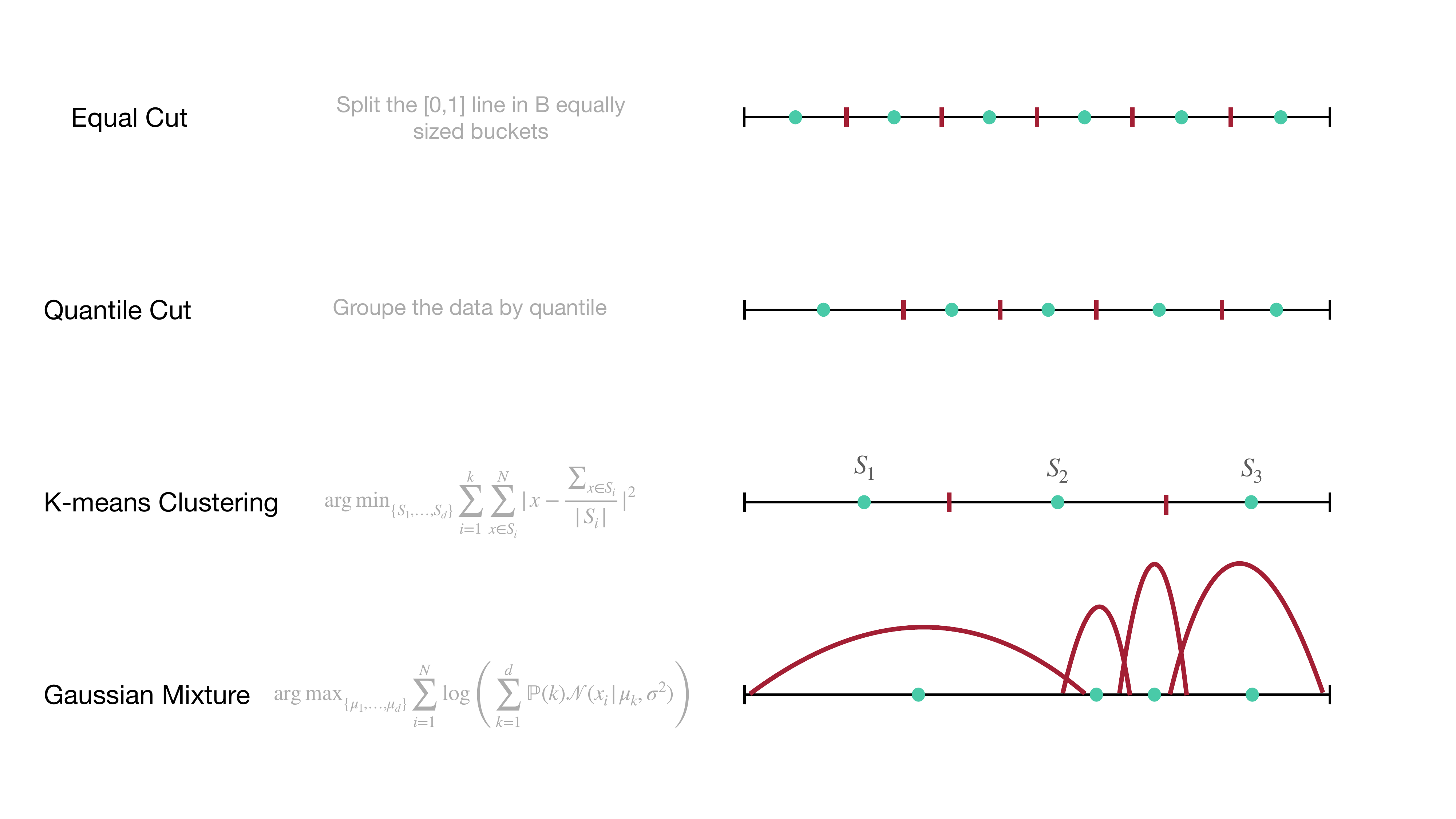}
  \caption{Alternative Bucketing Strategies Illustrated}
  \label{Fig:buckets}
\end{figure}

\subsection{Maximum Likelihood Estimation for $\varphi^\ell_{e, t}$}\label{app:mle-phi}
Fix a delegation graph, corresponding to experiment $e$ and task $t$. We are interested in reconstructing the most likely estimate of $\varphi_{e,t}^\ell(\eta_k)$ 
(we will drop $e$ and $t$ from the notation when they are clear from context). Let $z^\ell_k$ be the observed number of times that someone of type $\ell$ delegated to someone of type $k.$ Let $n_\ell$ be the number of people of type $\ell$, and let $\Tilde{n}_\ell$ be the number of people of type $\ell$ that delegated. 
Note that $\varphi$ is always invariant to scaling, so for consistency across experiments, we will always normalize 
SO
that $\sum_k \varphi^{\ell}_{e, t}(\eta_k) = 1$.

\begin{proposition}
    The maximum likelihood estimators for $\varphi_{e,t}^\ell(\eta_1), \ldots, \varphi_{e,t}^\ell(\eta_B)$ are\footnote{Or $0$ in if the denominator is $0$.}
    \[\varphi^\ell_{e, t}(\eta_k) = 
    \begin{cases}
        \frac{z^\ell_k}{n_k} & \text{ if } k \ne \ell\\
        \frac{z^\ell_k}{n_k - 1} & \text{ if } k = \ell
    \end{cases}
\]
up to arbitrary normalization constants. 
\end{proposition}
\proof{Proof}
    Fix $e$, $t$, and $\ell$. Let $\tilde{n}_k = n_k$ for $k \ne \ell$ and $\tilde{n}_\ell = n_\ell - 1$. Note that when a voter in bucket $\ell$ is considering whether or not to delegate, they see $\tilde{n}_k$ other voters of each type $k$. 
    Hence, for a parameter setting $\vec{\varphi} = (\varphi^{\ell}(\eta_1), \ldots, \varphi^{\ell}(\eta_k))$, the probability they delegate to a specific voter of type $k$ is $\frac{\varphi^\ell(\eta_k)}{\sum_{k'} \tilde{n}_{k'}\varphi^\ell(\eta_{k'})}$.
    In the data observed, a voter in bucket $\ell$ delegated to another in bucket $k$ a total of $z^\ell_k$ times. Hence, the likelihood of $\vec{\varphi}$ can be written as
    \[
        L(\vec{\varphi}) = \prod_k \left(\frac{\varphi^\ell(\eta_k)}{\sum_{k'} \tilde{n}_{k'}\varphi^\ell(\eta_{k'})}\right)^{z^{\ell}_k}
    \]
    
    Notice that this probability does not change by changing all the $\varphi$ values by the same constant factor, so when finding the most likely, it is without loss of generality to optimize over $\vec{\varphi}$ such that $\sum_{k'} \tilde{n}_{k'}\varphi^\ell(\eta_{k'}) = 1$. In such cases, the likelihood simplifies to
    \[
        L(\vec{\varphi}) = \prod_k \left(\varphi^\ell(\eta_k)\right)^{z^{\ell}_k}.
    \]
    Additionally, notice that optimizing over a constant times the likelihood results in the same optimization problem, so the MLE $\vec{\varphi}$ must also maximize, multiplying by $\prod_k (\tilde{n}_k)^{z^\ell_k}$,
    \[
        L'(\vec{\varphi}) = \prod_k \left(\tilde{n}_k \varphi^\ell(\eta_k)\right)^{z^{\ell}_k}.
    \]
    Finally, note that choosing $\vec{\varphi}$ such that $\sum_{k'} \tilde{n}_{k'}\varphi^\ell(\eta_{k'}) = 1$ so as to maximize $L'$ is equivalent to finding the MLE of a multinomial distribution with $B$ outcomes and probability $\tilde{n}_k\varphi^\ell(\eta_k)$ of ending up in outcome $k$. This is well known to be maximized with $\tilde{n}_k\varphi^\ell(\eta_k) \propto z_k$ 
    Therefore, up to arbitrary normalization factors, the MLE $\vec{\varphi}$ satisfy $\varphi^\ell(\eta_k) = \frac{z_k}{\tilde{n}_k}$.
\Halmos\endproof

\subsection{Kendall tau rank correlation coefficient }\label{app:phi-kendall-tau}
Let $R(\varphi_{e, t}^{\ell}(\eta_k))$ be the rank of $\varphi_{e, t}^{\ell}(\eta_k)$ amongst all the estimated weights and  $R(\eta_k)$ be the rank of $\eta_k$ 
among all competenceS. 
We then compute $$\tau = \frac{\sum_{k\in[N]}\mathbb{I}[R(\varphi_{e, t}^{\ell}(\eta_k)) = R(\eta_k)]-\sum_{k\in[N]}\mathbb{I}[R(\varphi_{e, t}^{\ell}(\eta_k)) \neq R(\eta_k)]}{N - N_{\varphi} - N_{\eta}}.$$ where $N_{\varphi}, N_{\eta}$ represent the number of ties in $R(\varphi_{e, t}^{\ell}(\eta_k))$ and $R(\eta_k)$ respectively. Next, we compare the z-score for 
the Kendall tau rank,
$z = \frac{3\tau\sqrt{N(N-1)^2(N-2)}}{\sqrt{8(N(N-1)+5)}}$ the values of a Gaussian distribution.\footnote{(Or, more precisely, we let the function stats.kendalltau in python's library scipy handle all of this for us.)} We further 
check the Kendall 
tau
rank correlation coefficient between $\varphi_{e, t}^{\ell}(\eta_k)$ and $\eta_k$ for a fixed $\ell.$

\subsection{Task-specific Effects}\label{app:phi-task-specific}

In \Cref{table:resultsphi_task}, we report the 
the Kendall tau correlation coefficients for each task to understand whether the observed behavior is constant across tasks. Note that for the first three tasks, there are $96$ data points per regression. For those, we observe a clear trend that $\varphi$ is increasing in its second coordinate across all three tasks. The following ones, however, rely on only $16$ data points.

\begin{table}[thb]
  \caption{Correlation and p-values Across Tasks}
  \small 

   \centering
   \vspace{0.5em}
  \begin{tabular}{ccc}
    \hline
    Task & Correlation$^{}$ & p-value \\
    \hline
    $T_1$ & 0.3647$^{***}$ & 0.0011 \\
    $T_2$ & 0.2621$^{**}$ & 0.0079 \\
    $T_3$ & 0.3830$^{****}$ & $<$ 0.0001 \\
    $T_4$ & -0.1540 & 0.6233 \\
    $T_5$ & -0.2055 & 0.1004 \\
    $T_6$ & -0.1372 & 0.5666 \\
    $T_7$ & 0.4768$^{*}$ & 0.0201 \\
    $T_8$ & -0.1925 & 0.5392 \\
    $T_9$ & -0.5231$^{**}$ & 0.0127 \\
    $T_{10}$ & -0.2490 & 0.2300 \\
    $T_{11}$ & 0.5427$^{**}$ & 0.0108 \\
    $T_{12}$ & 0.0321 & 0.9115 \\
    $T_{13}$ & 0.5644$^{**}$ & 0.0059 \\
    $T_{14}$ & 0.5768$^{**}$ & 0.0051 \\
    $T_{15}$ & -0.2689 & 0.1948 \\
    \hline
    
  \end{tabular}\label{table:resultsphi_task}
  
  \vspace{.5em}
  $^{*}$p$<$0.1; $^{**}$p$<$0.05; $^{***}$p$<$0.01; $^{****}$p$<$0.0001
\end{table}

\subsection{Alternative Bucketing Results}\label{app:buckets-results}

We run the same experiments as in \Cref{sec:method_phi} with $B \in \{3, 5, 7, 10\}$. \Cref{Fig:b3,Fig:b5,Fig:b7,Fig:b10} and \Cref{tab:b3,tab:b5,tab:b7,tab:b10} show the equivalent of \Cref{Fig:phiQ} and \Cref{phi:test_table} for these various values of $B$. Similarly, \Cref{Fig:b7_q} and \Cref{tab:b7_q} show these results for $B=7$ using quantile cut bucketing, but all different variations lead to qualitatively similar results.

\begin{figure}[htb]
  \centering
  \includegraphics[width=0.7\linewidth]{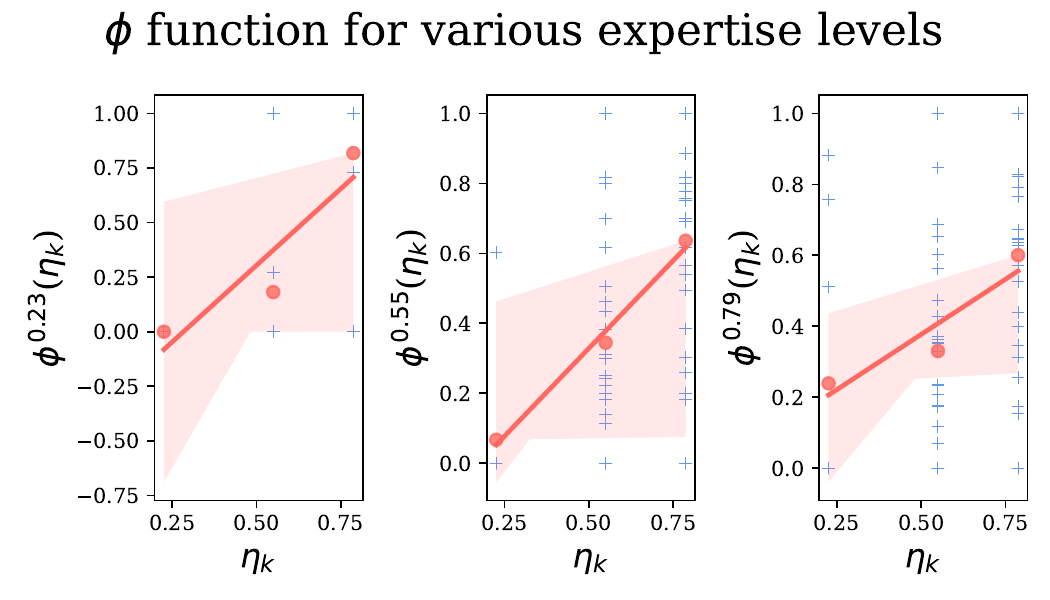}
  \includegraphics[width=0.5\linewidth]{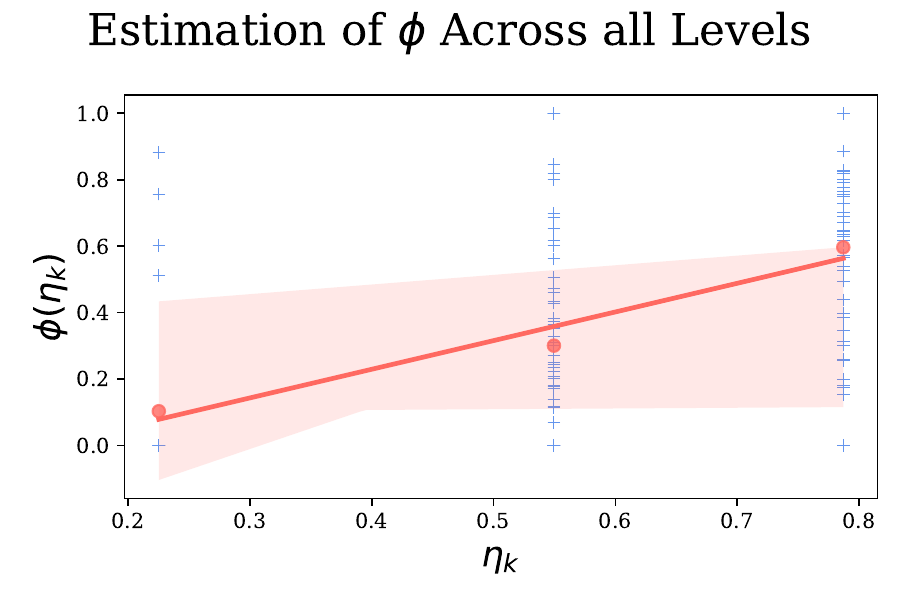}
  \caption{Estimation of $\varphi$ with $B=3$}
  
  \label{Fig:b3}
\end{figure}

\begin{table}[htb]
\caption{Estimation of $\varphi$ with $B=3$}\label{tab:b3}
  \centering
{  \footnotesize\begin{adjustbox}{angle=0}\begin{tabular}{@{\extracolsep{5pt}}lcccc} 
\\[-1.8ex]\hline 
\hline \\[-1.8ex] 
& \multicolumn{1}{c@{}}{Overall} & 
\multicolumn{3}{c@{}}{For fixed $\ell$}\\ 
\cmidrule(lr){2-2}
\cmidrule(lr){3-5} \\[-2ex] 
& & $c_1$ & \mc{$c_2$} & \mc{$c_3$}\\ 
\midrule
Correlation & $0.44^{****}$ & 
 $0.60^{*88*}$ & $0.43^{*888}$ & $0.34^{****}$\\ 
  P-value & $2\times 10^{-11}$ & 
  $2\times10^{-3}$& $2\times10^{-5}$ & $7\times10^{-4}$\\ \hline \\ [-1.8ex] \textit{Note:}  & \multicolumn{4}{r}{$^{*}$p$<$0.1; $^{**}$p$<$0.05; $^{***}$p$<$0.01; $^{****}$p$<$0.0001} 
\end{tabular}\end{adjustbox}}
\end{table}



\begin{figure}[htb]
  \centering
  \includegraphics[width=0.6\linewidth]{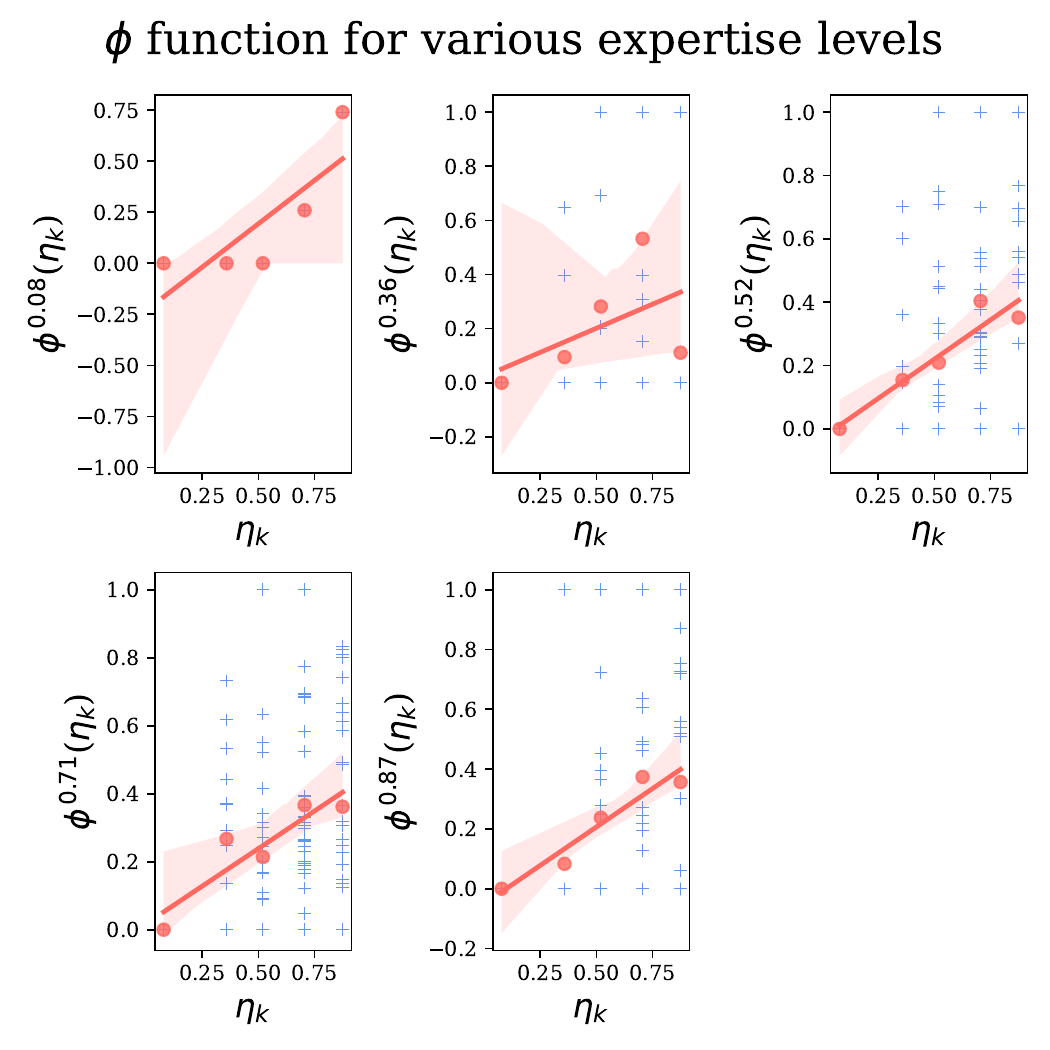}
  \includegraphics[width=0.4\linewidth]{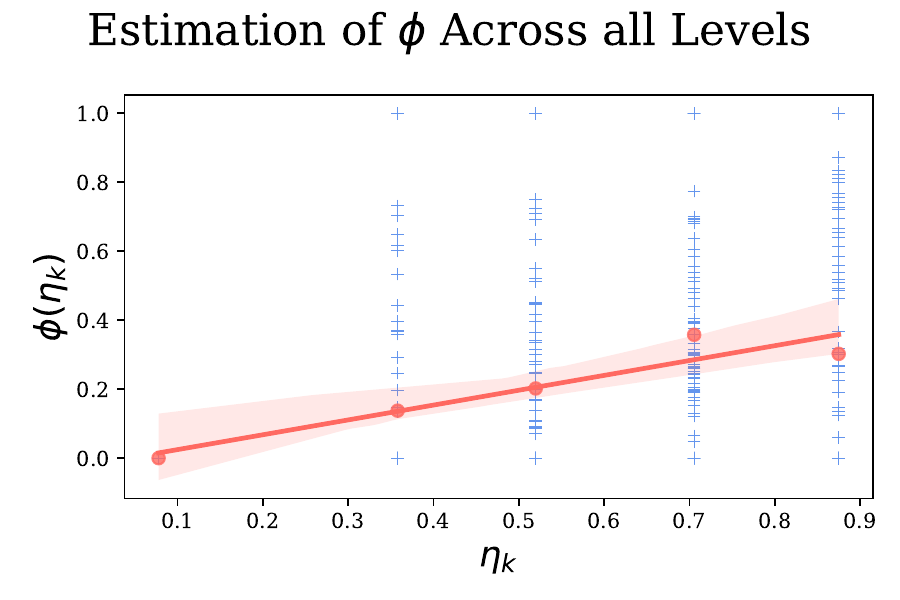}
  \caption{Estimation of $\varphi$ with $B=5$}
  
  \label{Fig:b5}
\end{figure}

\begin{table}[htb]
\caption{Estimation of $\varphi$ with $B=5$}\label{tab:b5}

  \centering
{  \footnotesize\begin{adjustbox}{angle=0}\begin{tabular}{@{\extracolsep{5pt}}lcccccc} 
\\[-1.8ex]\hline 
\hline \\[-1.8ex] 
& \multicolumn{1}{c@{}}{Overall} & 
\multicolumn{3}{c@{}}{For fixed $\ell$}\\ 
\cmidrule(lr){2-2}
\cmidrule(lr){3-7} \\[-2ex] 
& & $c_1$ & \mc{$c_2$} & \mc{$c_3$} & \mc{$c_4$}& \mc{$c_5$}\\ 
\midrule
Correlation & $0.13^{****}$ & 
 $0.14^{**}$ & $0.06$ & $0.15^{**}$ & $0.10$ & $0.21^{***}$\\ 
  P-value & $6\times 10^{-5}$ & 
  $6\times10^{-2}$& $3\times10^{-1}$ & $4\times10^{-2}$ & $10^{-1}$ & $5\times10^{-3}$\\ \hline \\ [-1.8ex] \textit{Note:}  & \multicolumn{6}{r}{$^{*}$p$<$0.1; $^{**}$p$<$0.05; $^{***}$p$<$0.01; $^{****}$p$<$0.0001} 
\end{tabular}\end{adjustbox}}
\end{table}

\begin{figure}[htb]
  \centering
  \includegraphics[width=0.7\linewidth]{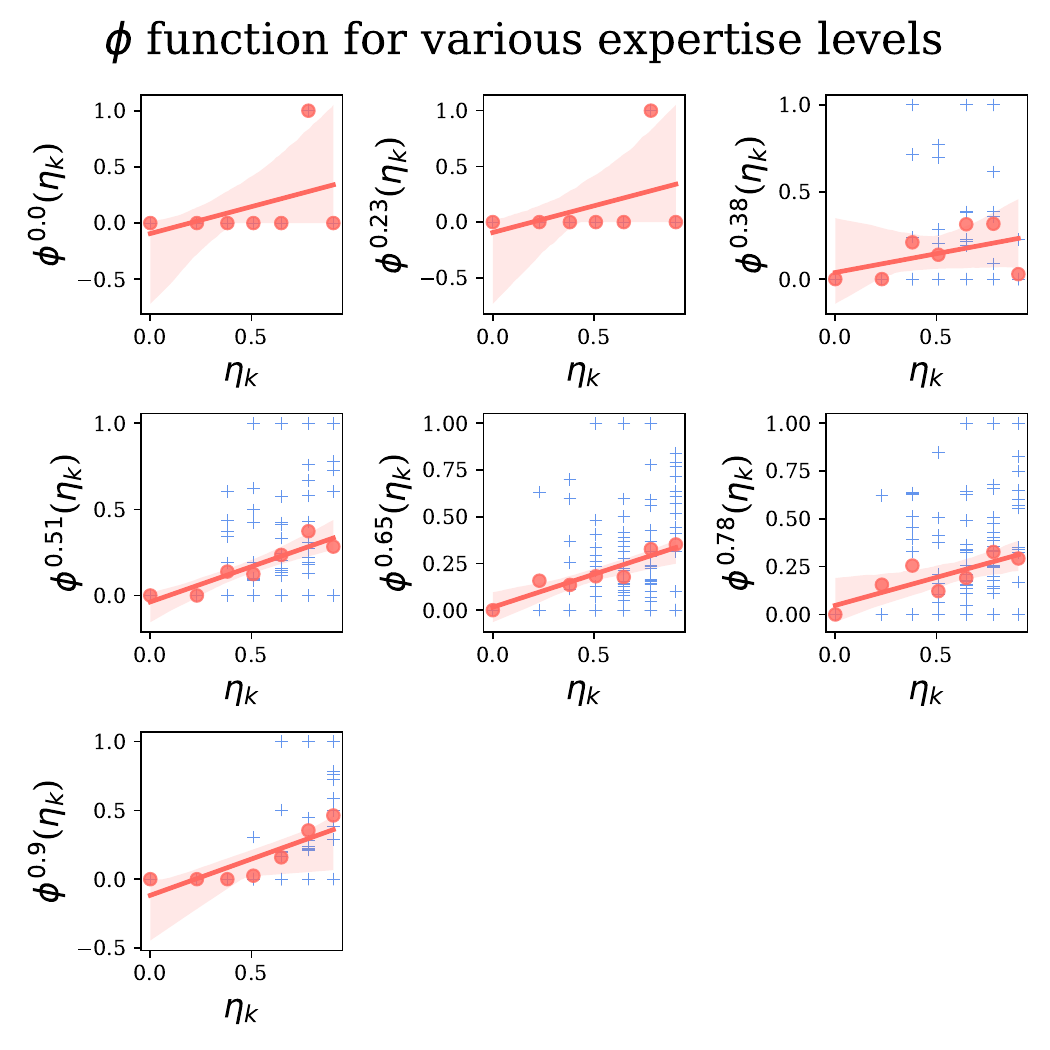}
  \includegraphics[width=0.4\linewidth]{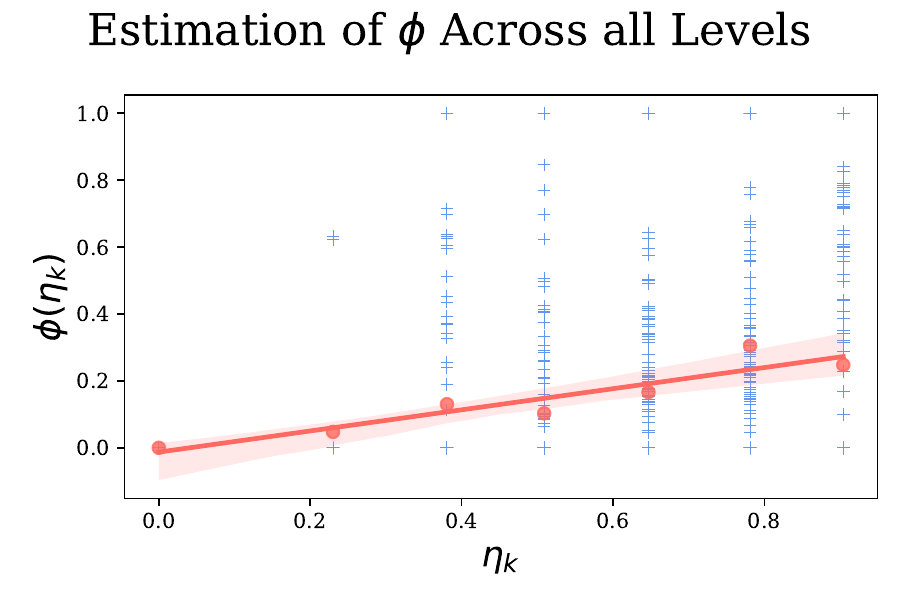}
  \caption{Estimation of $\varphi$ with $B=7$}
  
  \label{Fig:b7}
\end{figure}

\begin{table}[htb]
\caption{Estimation of $\varphi$ with $B=7$}\label{tab:b7}

  \centering
{  \footnotesize\begin{adjustbox}{angle=0}\begin{tabular}{@{\extracolsep{5pt}}lcccccccc} 
\\[-1.8ex]\hline 
\hline \\[-1.8ex] 
& \multicolumn{1}{c@{}}{Overall} & 
\multicolumn{3}{c@{}}{For fixed $\ell$}\\ 
\cmidrule(lr){2-2}
\cmidrule(lr){3-9} \\[-2ex] 
& & $c_1$ & \mc{$c_2$} & \mc{$c_3$} & \mc{$c_4$}& \mc{$c_5$}& \mc{$c_6$}& \mc{$c_7$}\\ 
\midrule
Correlation & $0.23^{****}$ & 
 $0.37$ & $0.36^{**}$ & $0.05$ & $0.18^{***}$ & $0.25^{****}$ & $0.16$& $0.45^{****}$\\ 
  P-value & $10^{-11}$ & 
  $3\times10^{-1}$& $10^{-2}$ & $6\times10^{-1}$ & $10^{-1}$ & $2\times10^{-4}$ & $2\times10^{-2}$ & $3\times10^{-6}$ \\ \hline \\ [-1.8ex] \textit{Note:}  & \multicolumn{8}{r}{$^{*}$p$<$0.1; $^{**}$p$<$0.05; $^{***}$p$<$0.01; $^{****}$p$<$0.0001} 
\end{tabular}\end{adjustbox}}
\end{table}

\begin{figure}[htb]
  \centering
  \includegraphics[valign=c, width=0.65\linewidth]{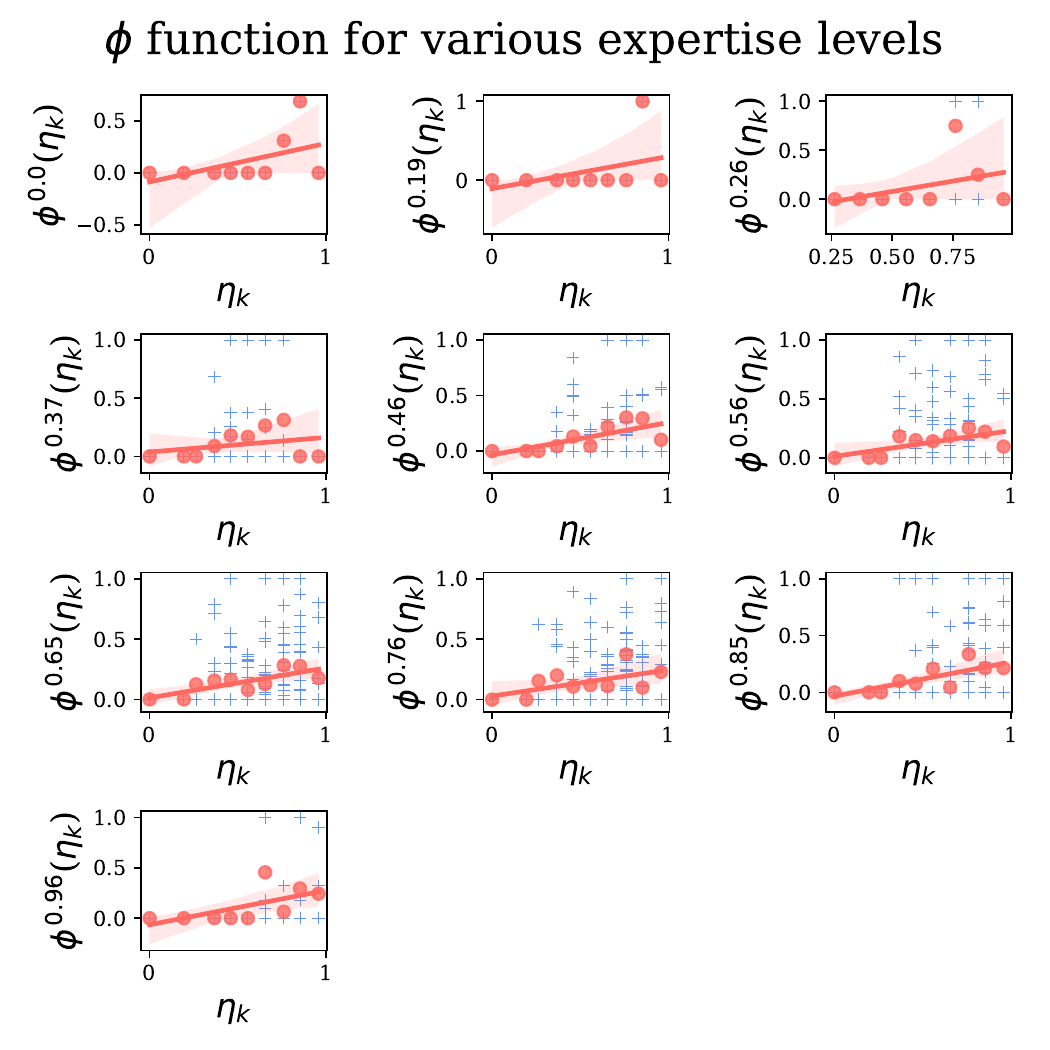}
  \includegraphics[valign=c, width=0.3\linewidth]{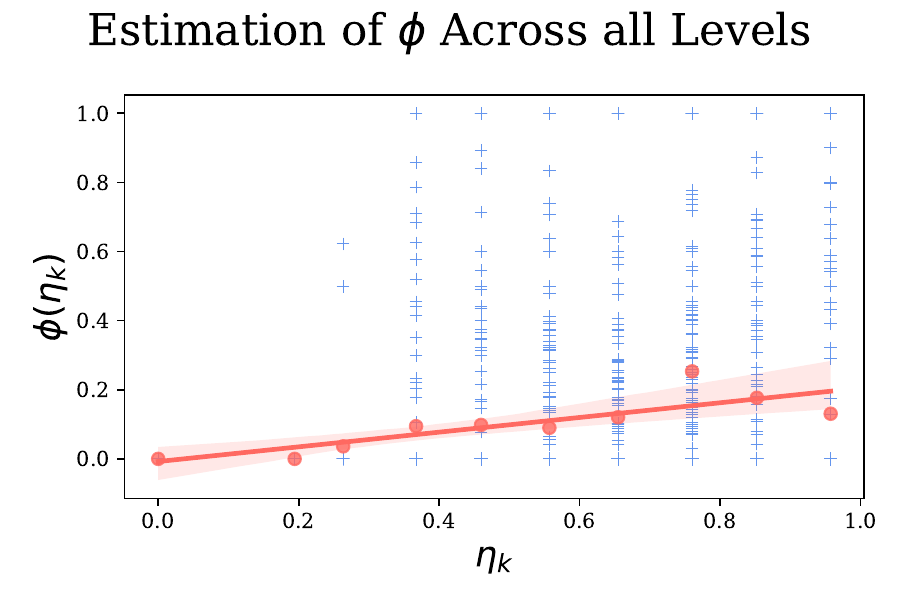}
  \caption{Estimation of $\varphi$ with $B=10$}
  
  \label{Fig:b10}
\end{figure}

\begin{table}[htb]
\caption{Estimation of $\varphi$ with $B=10$}\label{tab:b10}
  \centering
{  \footnotesize\begin{adjustbox}{angle=90}\begin{tabular}{@{\extracolsep{5pt}}lccccccccccc} 
\\[-1.8ex]\hline 
\hline \\[-1.8ex] 
& \multicolumn{1}{c@{}}{Overall} & 
\multicolumn{3}{c@{}}{For fixed $\ell$}\\ 
\cmidrule(lr){2-2}
\cmidrule(lr){3-12} \\[-2ex] 
& & $c_1$ & \mc{$c_2$} & \mc{$c_3$} & \mc{$c_4$}& \mc{$c_5$}& \mc{$c_6$}& \mc{$c_7$}& \mc{$c_8$}& \mc{$c_9$}& \mc{$c_{10}$}\\ 
\midrule
Correlation & $0.14^{****}$ & $0.47$ & $0.35$ & $0.35^{**}$ & $0.00$ & $0.16^{**}$ & $0.03$ & $0.15^{**}$ & $0.10^{*}$ & $0.21^{***}$ & $0.30^{**}$\\ 
  P-value &$1.2073\times 10^{-7}$ & $0.1111$ & $0.2453$ & $0.0388$ & $0.9763$ & $0.0291$ & $0.6754$ & $0.0135$ & $0.0957$ & $0.0026$ & $0.0418$ \\ \hline \\ [-1.8ex] \textit{Note:}  & \multicolumn{11}{r}{$^{*}$p$<$0.1; $^{**}$p$<$0.05; $^{***}$p$<$0.01; $^{****}$p$<$0.0001} 
\end{tabular}\end{adjustbox}}
\end{table}

\begin{figure}[htb]
  \centering
  \includegraphics[valign=c, width=0.7\linewidth]{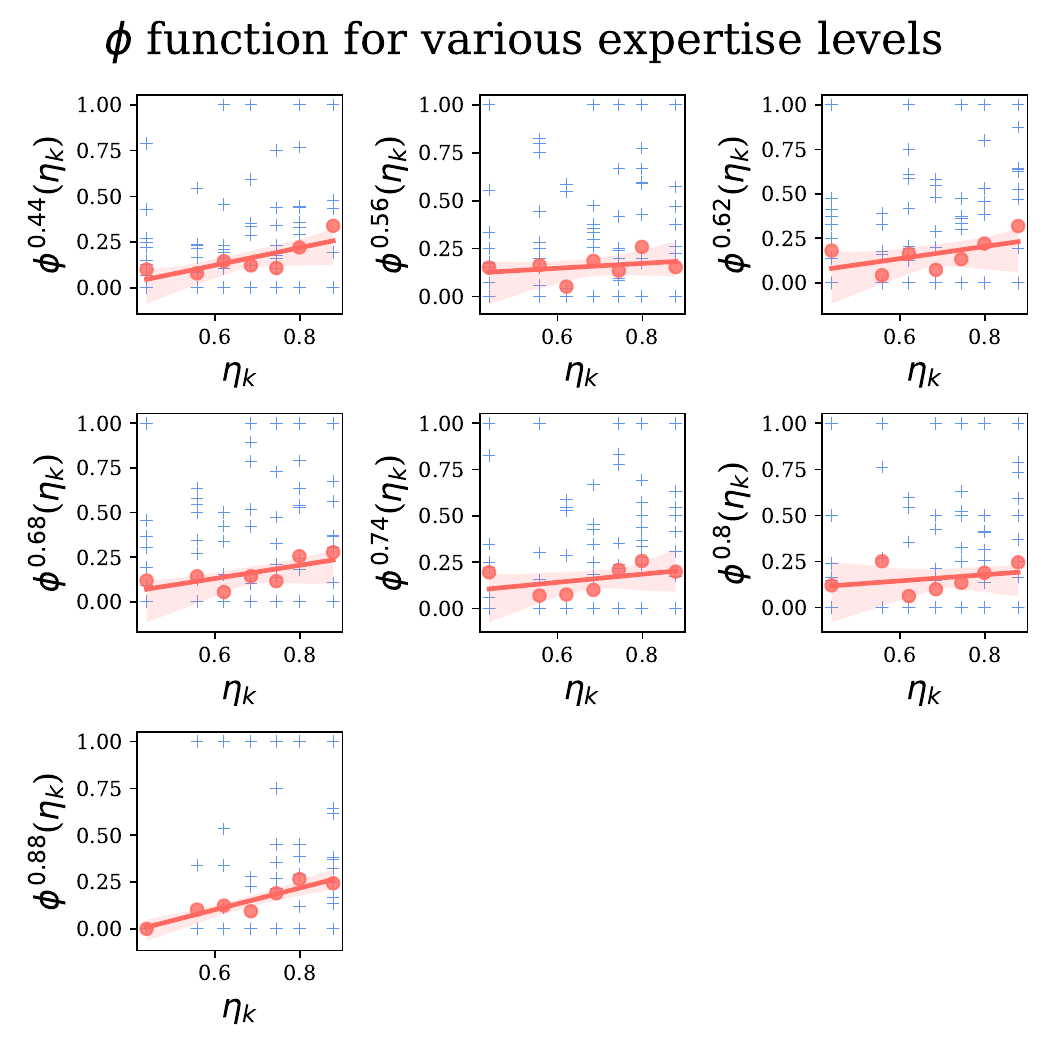}
  \includegraphics[valign=c, width=0.4\linewidth]{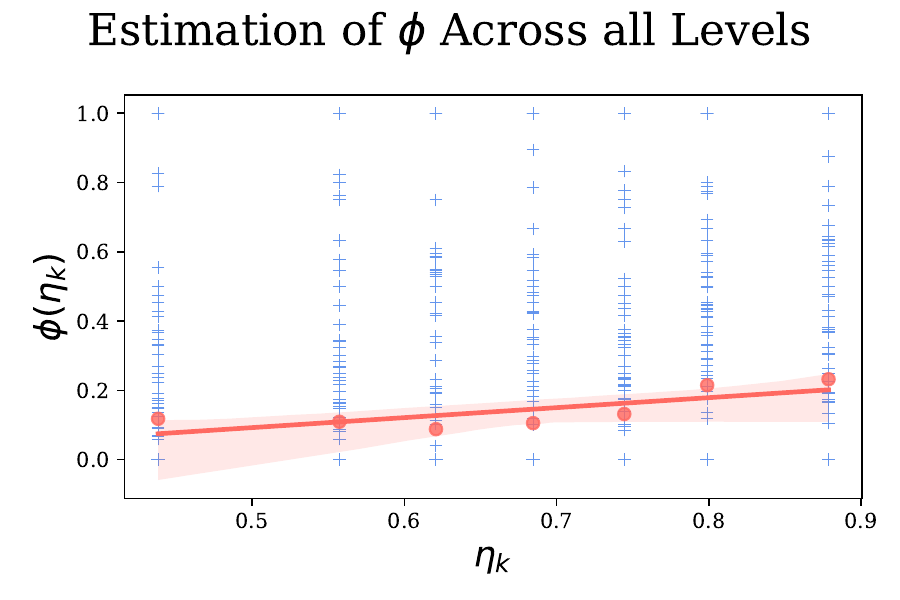}
  \caption{Estimation of $\varphi$ with $B=7$ using quantile bucketting.}
  
  \label{Fig:b7_q}
\end{figure}

\begin{table}[htb]
\caption{Estimation of $\varphi$ with $B=7$ using quantile bucketting.}\label{tab:b7_q}

  \centering
{  \footnotesize\begin{adjustbox}{angle=0}\begin{tabular}{@{\extracolsep{5pt}}lcccccccc} 
\\[-1.8ex]\hline 
\hline \\[-1.8ex] 
& \multicolumn{1}{c@{}}{Overall} & 
\multicolumn{3}{c@{}}{For fixed $\ell$}\\ 
\cmidrule(lr){2-2}
\cmidrule(lr){3-9} \\[-2ex] 
& & $c_1$ & \mc{$c_2$} & \mc{$c_3$} & \mc{$c_4$}& \mc{$c_5$}& \mc{$c_6$}& \mc{$c_7$}\\ 
\midrule
Correlation & $0.10^{****}$ & 
 $0.07$ &  $0.03$ & $0.08$ &  $0.10^{*}$& $0.11^{*}$ & $0.09$ & $0.29^{***}$\\ 
  P-value & $3\times 10^{-5}$ & 
  $0.25$& $0.62$ & $0.16$ & $0.099$ &$0.063$&$0.17$&$0.0003$ \\ \hline \\ [-1.8ex] \textit{Note:}  & \multicolumn{8}{r}{$^{*}$p$<$0.1; $^{**}$p$<$0.05; $^{***}$p$<$0.01; $^{****}$p$<$0.0001} 
\end{tabular}\end{adjustbox}}
\end{table}

\clearpage
\section{Additional Tests and Statistics}\label{app:additional}

\subsection{Robustness Check on Attention}\label{app:rob}
In order to check whether asking respondents to answer the questions associated with the tasks they delegated biased their answers, we run a robustness check on the time spent on the survey as a function of the delegation frequency.

First, note that 
the
design’s key components are that (i) respondents delegate 
in response to a
prompt (e.g., 
historical landmark) without seeing the questions so that (ii) we collect the liquid vote 
along with how they would have voted had they not delegated (this information is also used to assess competence).
Subsequently, participants are 
told
“you may be asked to answer a few additional questions.” Importantly, we do not say why, and participants answer questions they have not seen before. 
We deliberately designed
the survey that way 
in order
to minimize any biasing effects 
that might make a 
participant 
less motivated to answer questions they delegated. 

Second, we 
computed
the time participants spent on the experiment (standardized per experiment to account for the fact that different experiments had different questions) as a function of the proportion of time 
that
they delegate, to see whether they are significantly likely to spend less time on the test if they 
delegate
more (as one could hypothesize 
that
they would be less motivated and then answer the second phase questions more rapidly). We 
found
no relation between time spent on the survey and 
the likelihood delegation
(with a confidence interval for the coefficient of [-0.505, 0.155] and a coefficient of -0.175).

\subsection{Frequency of Correctness}\label{app:frequency}

We begin by directly comparing how frequently liquid and direct democracy were correct. The corresponding data (separated by task) can be found in \Cref{fig:compare-correct}. Although our experiments find evidence for sufficient conditions for DNH and PG, we do not observe statistically significant gains for LD over DD (but we do observe a positive difference in the point estimates). This is still consistent with DNH and PG holding (which both only kick in the limit, and even on these finite samples, LD is doing no worse). However, even if we believe such a gain could exist, it could not appear for a variety of reasons: it’s plausible we simply didn’t have enough statistical power to detect the improvement (we have $32$ comparisons) or that the improvement appears in the large limit and there were not enough group sizes. We leave the study of this to future work.

\begin{figure}[htb]
  \begin{center}
      \includegraphics[width=0.5\linewidth]{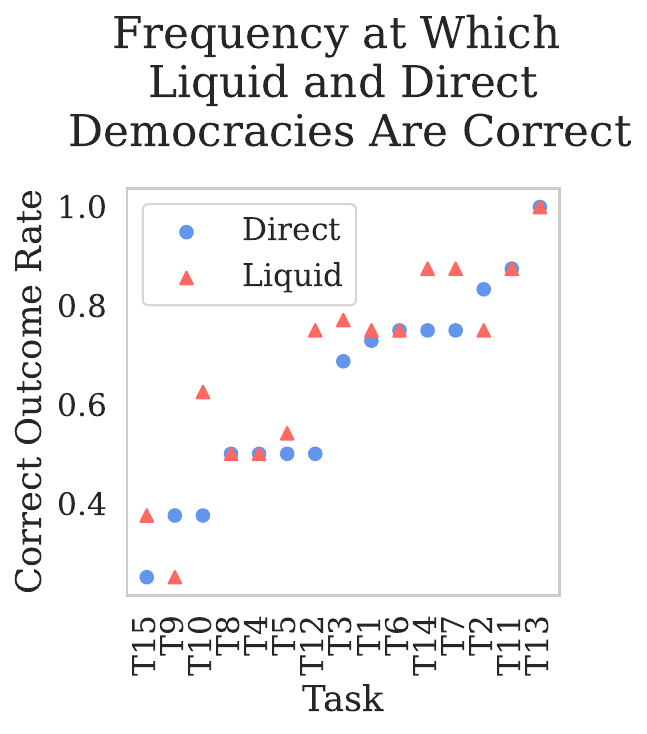}
  \end{center}
  \caption{Frequency of Correctness for Liquid and Direct Democracies}
  \label{fig:compare-correct}
\end{figure}

On average, we see an overall improvement, with a few exceptions for some tasks. To understand the statistical significance of this improvement, we fit the following generalized linear model. We let $o^L_{e, t}$ and $o^D_{e, t}$ be the proportion of questions answered correctly using liquid and direct democracy in experiment $e$ and task $t$. More formally, $o_{e,t}^L = \sum_{r} \mathbb{I}[\sum_{i\in[N_e]} v_{i, r}^\ell > N_e/2]/8$ and $o_{e,t}^D = \sum_{r} \mathbb{I}[\sum_{i\in[N_e]} v_{i, r}^D > N_e/2]/8.$ We then fit the following generalized linear model
\begin{equation*}
    o^j_{e, t} = \alpha_0 + \alpha_e + \alpha_t + \beta^{LvD}\gamma_j + \e_{e, t}
\end{equation*}
where $\gamma_j = \mathbb{I}[j = L]$, i.e., it is an indicator that this proportion came from liquid democracy. Once fitted, $\beta^{LvD}$ here represents the impact of using liquid democracy over direct democracy, while taking into account fixed effects from tasks and experiments (i.e., some tasks may be easier or harder, and some groups may overall perform better than others). We find $\beta^{\text{LvD}}=0.0313$ with $s.e.=0.025,$ $t=1.229$ and $p=0.23.$ In other words, liquid democracy's average proportion of correct answers is $3$ points above that of direct democracy, but this increase is not significant. 

\subsection{Increase in Competence}\label{app:increase}
Next, we test liquid democracy's ability to increase the average competence. We let $\eta_{e,t}^L = \frac{\sum_{i \in [N_e]}w_{i, t}\eta_{i,t}}{N_e}$ (resp. $\eta_{e,t}^D = \frac{\sum_{i \in [N_e]}\eta_{i,t}}{N_e}$) to denote the average competence 
after (resp., before)
delegation. We fit an analogous generalized linear model to the proportion of correctness, namely
\[
    \eta^j_{e, t}  = \alpha_0 + \alpha_e + \alpha_t + \beta^{increase}\gamma_j + \e_{e, t}.
\]
We find $\beta^{\text{increase}}=0.031$ with $s.e.=0.006,$ $t=4.78$ and $p=0.000004.$ In other words, across all tasks and experiments, the mean average competence post delegation is $3\%$ higher than the mean average competence without delegation, and this result is quite statistically significant.

\subsection{Maximum Weight}\label{app:maxweight}
Next, we consider the maximum weight received by any individual. Note that this weight should always be considered with respect to the size of the group considered. Boxplots of the observed maximum weights can be found in \Cref{Fig:weight}. Along with this, we plot a line for $N_e/2$, the ``dictator'' line. Note that any voter with weight more than $N_e/2$ essentially has complete control over the outcome. 

\begin{figure}[htb]
  \begin{center}
      \includegraphics[width=0.9\linewidth]{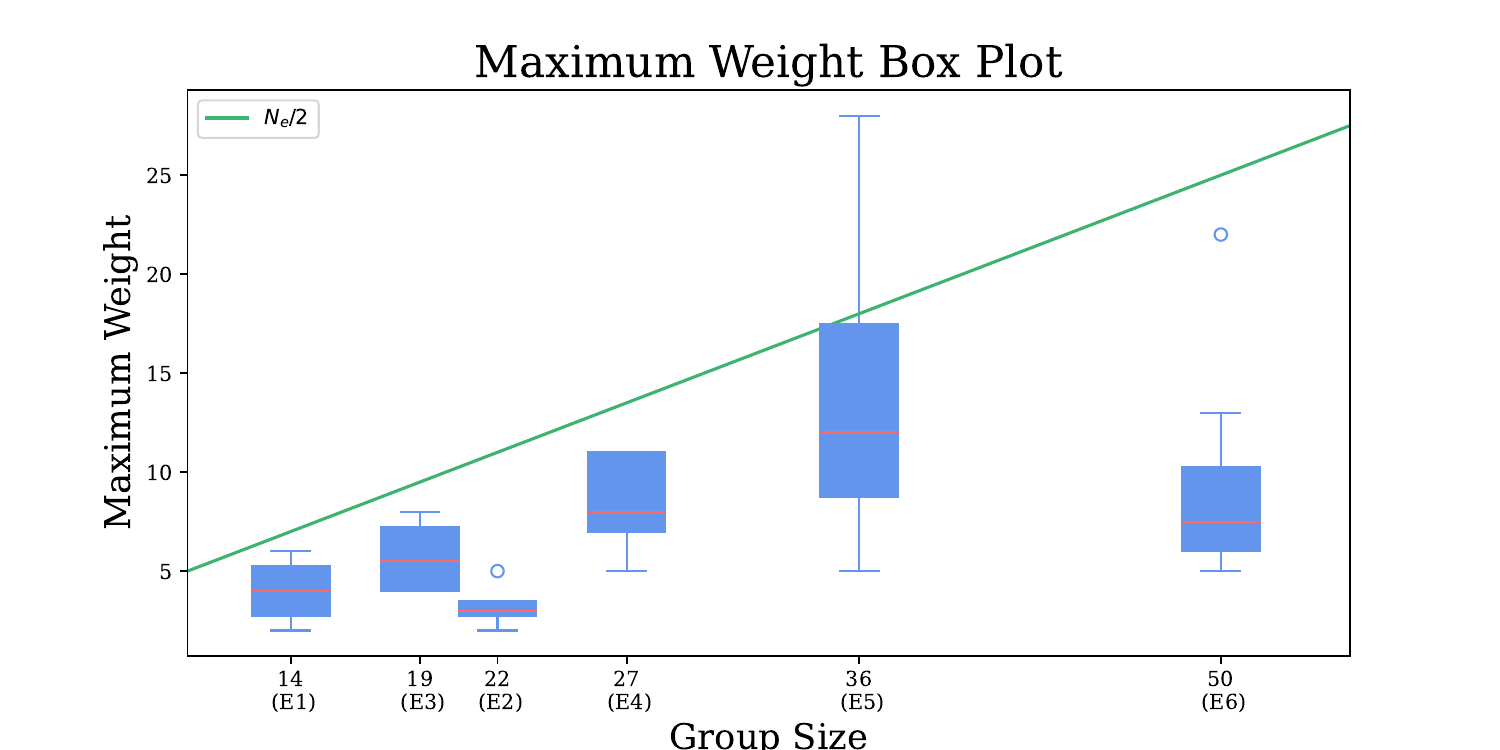}
  \end{center}
  \caption{Distribution of maximum weights received in tasks across experiments.}
  \label{Fig:weight}
\end{figure}

competence after
delegation. For each experiment $e$ and task $t,$ we use $m_{e,t} = \max_{i \in [N_e]} w_{i, t}$ to denote the maximum weight, and $\eta_{e,t}^L = \frac{\sum_{i \in [N_e]}w_{i, t}\eta_{i,t}}{N_e}$ (resp. $\eta_{e,t}^D = \frac{\sum_{i \in [N_e]}\eta_{i,t}}{N_e}$) to denote the average competence post (resp. pre) delegation.

\subsection{Coalition Analysis}\label{app:coalition}
For this section, we will use powers of coalitions as a different way to measure concentration of power. We use the standard terminology of calling voters that do not delegate, \emph{gurus}. Suppose the gurus have weights $w_1 \ge w_2 \ge \cdots$. We will say the \emph{power} of a coalition of gurus $S$ is the number of votes they control, i.e., $\sum_{i \in S} w_i$.  Note that the maximum weight is exactly the most power a single guru holds. However, we can also consider other measures such as the largest amount of power that $k$ gurus hold. This could potentially be as large as $k$ times the max weight, but in certain instances could be much smaller. Of paticular interest is the minimal coalition size that constitutes a majority, i.e., the minimal $k$ such that $\sum_{i = 1}^k w_i \ge n/2$. Note that in the case of a dictatorship, this size is just $1$, while in the case of direct democracy, it is $n/2$. Values among this spectrum give another measure of how concentrated power is.

\Cref{Fig:wcdf} shows this first property, i.e., in each experiment and each task, for any fixed $k$, what proportion of the votes do the top $k$ gurus control. \Cref{Fig:box_guru} shows the second, that is, for each group, what the fewest number of gurus that collectively control half the votes. Note that, anecdotally, the size of the smallest coalition that 
controls half the votes is about the square root of the 
group size. This is based only on the six mean numbers measured in \Cref{Fig:box_guru}. No rigorous conclusions can be drawn for this, but further empirical inquiry on this matter would be 
an interesting direction for future work.

Finally, as a way to understand the unequalness of this power, we plot the smallest majority coalition size against the total number of gurus. This can be found \Cref{Fig:gurus}. Overall, while this data is interesting, it is difficult to draw conclusions about how harmful this level of concentration of power is. Our theoretical results hold asymptotically, and either more data or larger group sizes would be needed to truly understand the implications.

\begin{figure}[htb]
  \centering
  \includegraphics[width=\linewidth]{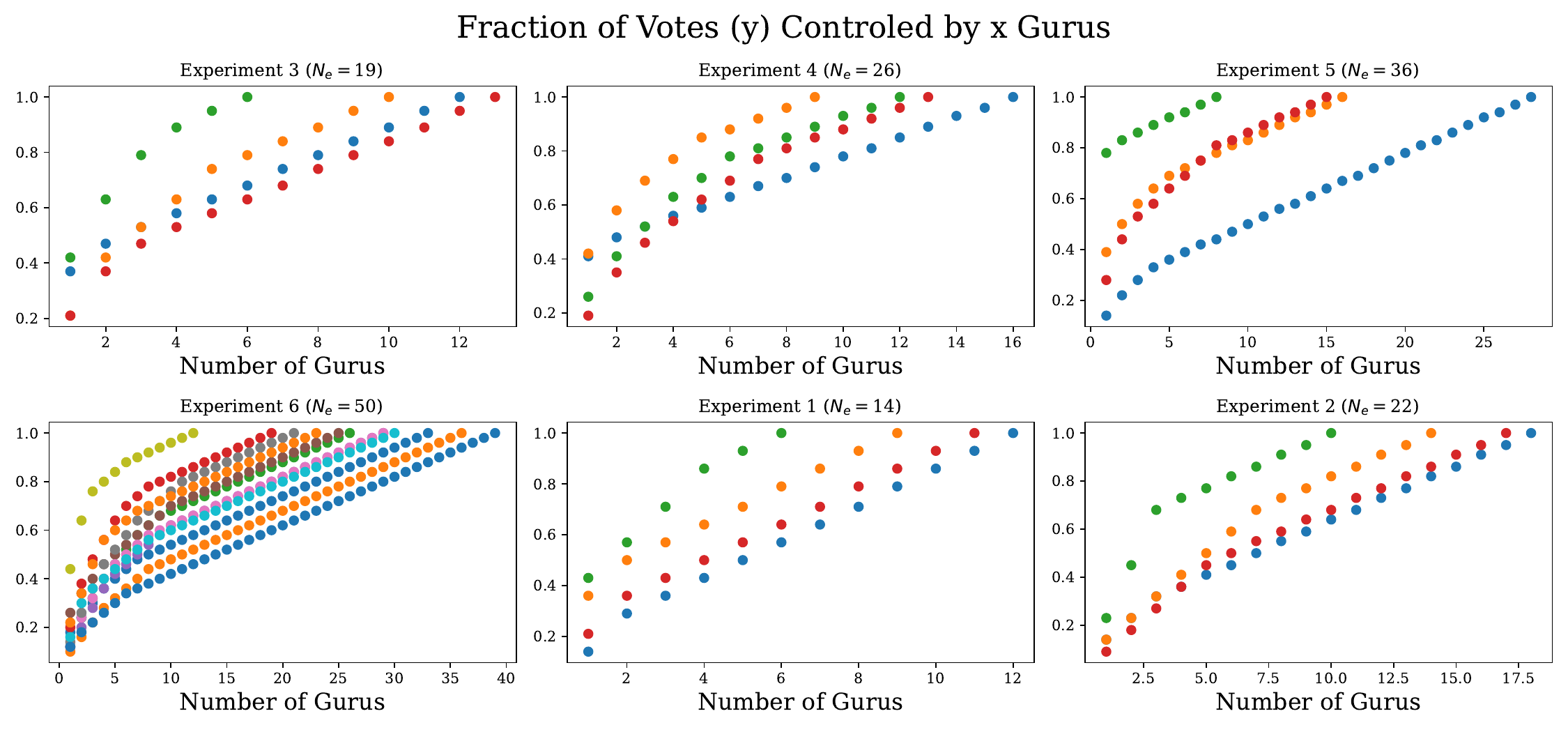}
  \caption{The fraction of votes gathered by 
  the
  smallest coalition that maximizes total weight. The x axis represents the number $k$ of voters with the $k$ largest weights. The y axis represents the fraction of votes they 
gathered altogether. Each color represents a task (the first five experiments comprised $4$ tasks, whil the sixth comprised $13$ tasks.}
  \label{Fig:wcdf}
\end{figure}

\begin{figure}[htb]
  \centering
  \includegraphics[width=0.7\linewidth]{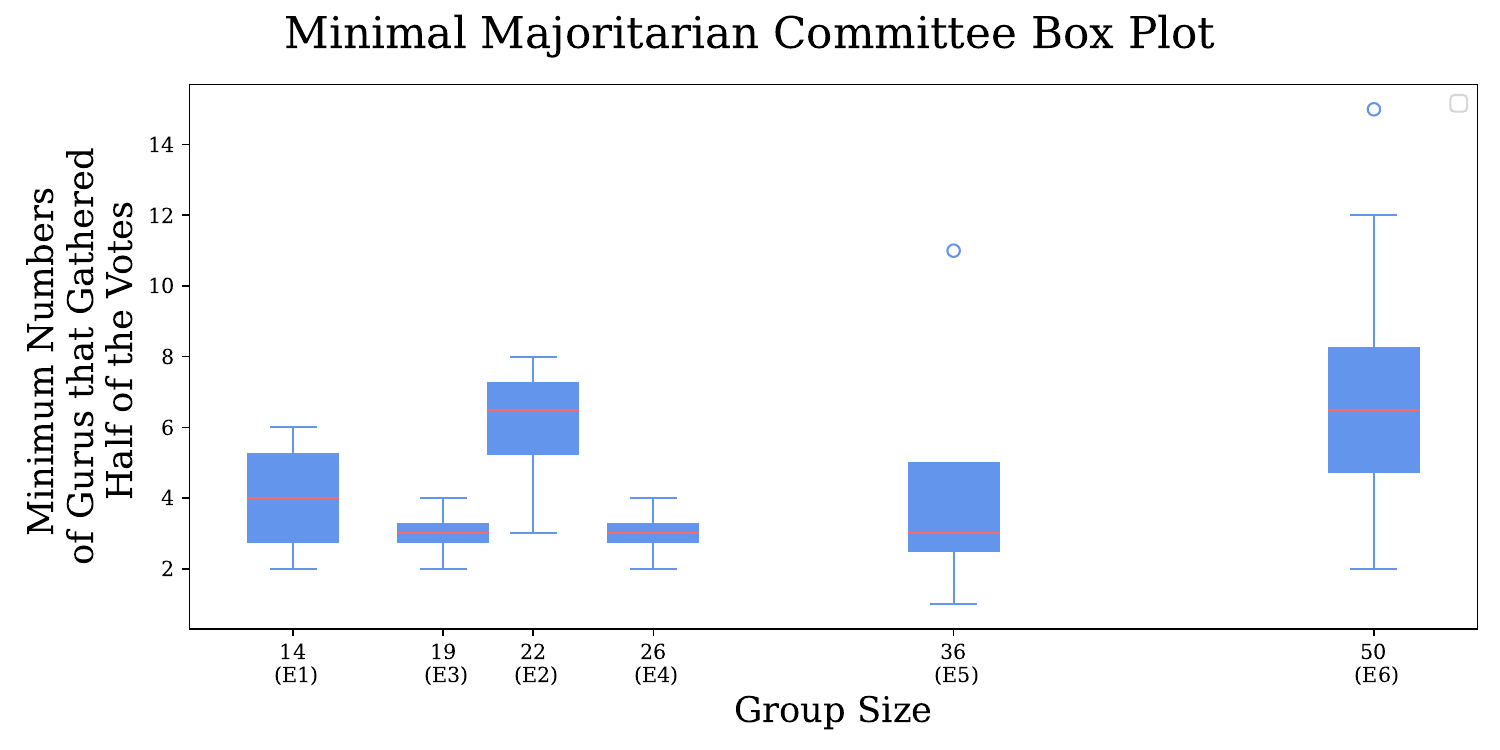}
  \caption{Smallest coalition of gurus that control half the votes. These are box-plots with medians, quantiles and outliers.}
  \label{Fig:box_guru}
\end{figure}

\begin{figure}[htb]
  \centering
  \includegraphics[width=0.7\linewidth]{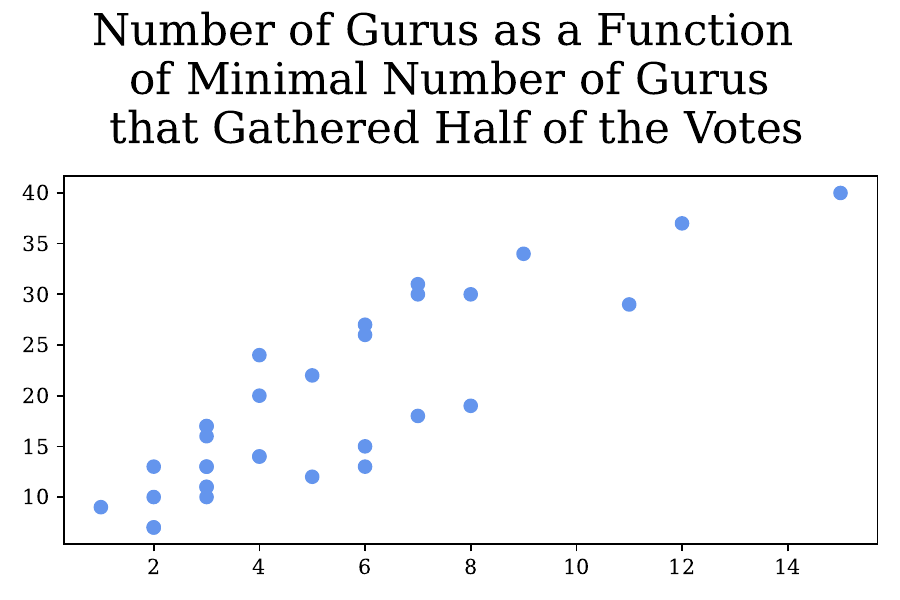}
  \caption{Number of Gurus as a Function of the Size of the Smallest Potentially 
  Majority
  Coalition (SPMC)}
  \label{Fig:gurus}
\end{figure}

\clearpage
\section{Pre-study}\label{app:pre}
We provide in this section the results of the same analyses 
in
the pre-study, that comprised $6$ experiments and a similar design that was used to inform the survey flow and material of the main study.

\subsection{Update in the Design for the Main Study}
In the main study, the spatial reasoning task was dropped as participants were almost always correct. Ambiguous questions (for which multiple answers were true) taken from \citet{simoiu2019studying} were rephrased to enforce a clear objective answer. Questions that were mislabeled in \citet{simoiu2019studying} were removed. The survey flow was also changed so that more questions could be answered in the same amount of time, and all tasks were set to have the same number of questions $|R_t|=8.$

\subsection{Recruitment}
The six groups with which the pre-study was run are described below. There were a total of $102$ participants that participated in the surveys between March 21st and April 5th, 2022. Of the participants across all experiments, 29\% were native English speakers, 16\% were female, 4\% non-binary, and 80\% were male.

\begin{table}[ht]
  \caption{Qualitative group descriptions and sizes from the pre-study}
  \centering
  \begin{tabular}{clc}
    \toprule
    Group ID & Group Description & Group Size \\
    \midrule
    1 & Research Group & 11 \\
    2 & Graduate and Undergraduate Class & 12 \\
    3 & Graduate Class & 32\\
    4 & Sports Team & 14\\
    5 & Financial Association & 18\\
    6 & Group of employees, students and faculty & 15\\
    \bottomrule
  \end{tabular}
  \label{tab:pre-studys}
\end{table}

\subsection{Material}
The tasks used in the pre-study are described below, \Cref{tab:tasks_pre} details the tasks and \Cref{tab:commands_pre} some of the questions within the tasks.

\begin{table}
  \caption{Task Descriptions}
\centering
\small
\begin{tabular}{p{1.7cm}p{1.9cm}p{9cm}p{2.6cm}}
\toprule
\textbf{ID} & \textbf{Number of Questions} &\textbf{Task Prompts} & \textbf{Corresponding Experiment(s)} \\
\midrule
Landmark & 5 & You will be shown images of architectural landmarks from around the world, and asked to select the country where the landmark is located. & $1, 2, 3, 4, 5, 6$ \\
Movie& 5 & You will be provided with short audio files with theme songs from various movies, and asked to select the movie it was featured in. & $1, 2, 3, 4, 5, 6$ \\
English& 5 & You will be given English idioms, and asked to identify their meaning. An idiom is a group of words that have a meaning not deducible from those of the individual words (e.g., rain cats and dogs, see the light). & $1, 2, 3, 4, 5, 6$ \\
Spatial& 3& You will be asked to watch a short video of the Cups and Balls magic trick, and identify the location of the ball at the end of the trick.
 & $1, 2, 3, 4, 5, 6$ \\
\midrule
Sports& 7& You will be given US college basketball teams, and asked to predict which round they will make it to in the NCAA Tournament, taking place in March 2022?& $1,2$
\\
Sports& 7&  You will be given upcoming soccer games, and asked to predict the games' outcome?
 & $3, 4, 5$ \\
Sports& 7& You will be given upcoming sport events (soccer and tennis games), and asked to predict the games' outcome?
 & $6$ \\
\bottomrule
\end{tabular}
\label{tab:tasks_pre}
\end{table}

\begin{table}[htb]
  \caption{Survey Material}
    \centering
    \scriptsize
  \begin{tabular}{lp{10cm}l}
    \toprule
    Task & Prompt & Answer\\
    \midrule
    \multirow{5}*{Landmark} &  This landmark is located in Italy. & False \\
    &  This landmark is located in Turkey. & True \\
    &  This landmark is located in Myanmar. & False \\
    &  This landmark is located in France. & False  \\
    &  This landmark is located in Brazil. & False \\
    \midrule
    \multirow{7}*{Movie} &  This music was featured as a theme song in the movie The Hobbit. & False\\
    &  This music was featured as a theme song in the movie The Empire of Sun. & False \\
    &  This music was featured as a theme song in the movie Gravity. & True \\
    &  This music was featured as a theme song in the movie Goodfellas. & False\\
    &  This music was featured as a theme song in the movie The Pianist. & False \\
    &  This music was featured as a theme song in the movie A Passage through India. & False\\
    &  This music was featured as a theme song in the movie The Schindler's List. & True \\
    \midrule
    \multirow{5}*{English} &  ``A man of straw" means ``A very active person". & False \\
    &  ``To drive home" means ``To emphasize". & True \\
    &  ``To smell a rat" means ``To suspect foul dealings". & True \\
    &  ``To end in smoke" means ``To excite great applause". & False \\
    &  ``To catch a tartar" means ``To deal with a person who is more than one's match". & False \\
    \midrule
    \multirow{5}*{Prediction for Experiments 1-2} &  The US college basketball team West Virginia Mountaineers will make it to the Elite Eight in the 2022 NCAA Tournament. & False \\
    &  The US college basketball team Michigan State Spartans will make it to the First Round in the 2022 NCAA Tournament. & True \\
    &  The US college basketball team Syracuse Orange will win the 2022 NCAA Tournament. & False \\
    &  The US college basketball team Purdue Boilermakers will make it to the 2nd round in the 2022 NCAA Tournament. & True \\
    &  The US college basketball team Arizona Wildcats will make it to the Elite Eight in the 2022 NCAA Tournament. & False \\
    \midrule
    \multirow{5}*{Prediction for Experiments 3-5} &  Galatasaray SK will beat FC Barcelona during the Europa League game on March 17th. & False \\
    &  Olympic de Marseille and OGC Nice will tie during the French League game on March 20th. & False \\
    &  VFL Wolfsburg will beat Bayer 04 Leverkusen during the German League game on March 20th. & False \\
    &  Salernitana will lose against Juventus during the Italian League game on March 20th. & True \\
    &  FC Barcelona and Real Madrid CF will tie during the Spanish League game on March 20th. & False \\
    \midrule
    \multirow{5}*{Prediction for Experiment 6} &  Eintracht Frankfurt will beat FC Barcelona during the Europa League game on April 7th. & False \\
    &  Olympic de Lyon and West Ham United will tie during the Europa League game on April 7th. & True \\
    &  Brazil will lose to Spain during the Women's International Friendly game on April 7th. & False \\
    &  Neither Rafael Nadal nor Novak Djokovic will qualify for the ATP Masters 1000 Monte Carlo Final on April 17th. & NA \\
    &   Stefanos Tsitsipas will win the ATP Masters 1000 Monte Carlo Tournament on April 17th. & NA \\
    \midrule
    \multirow{3}*{Spatial Reasoning} &  The object is located in the middle cup at the end of the trick. & False \\
    &  The object is located in the middle cup at the end of the trick. & False \\
    &  The object is located in the right cup at the end of the trick. & True \\
    \bottomrule
  \end{tabular}
  \label{tab:commands_pre}
\end{table}

\subsection{Assessing competence}
We run the IRT framework and find a correlation of $97\%$ between the naive competence $\eta^{\text{naive}}_{i,t}$ and the competence computed accounting for task difficulty $\eta_{i,t}.$ Note that, while the competence distribution for the main experiment is normal (Appendix~\ref{app:norm}), the competence distribution in the pre-study is not, due to the spatial reasoning task for which almost all participants are correct.
We show the distribution of competence computed with both the naive and IRT frameworks
in \Cref{Fig:expertise_distribution_pre}.
\begin{figure}[htb]
  \centering
    \includegraphics[width=0.4\linewidth]{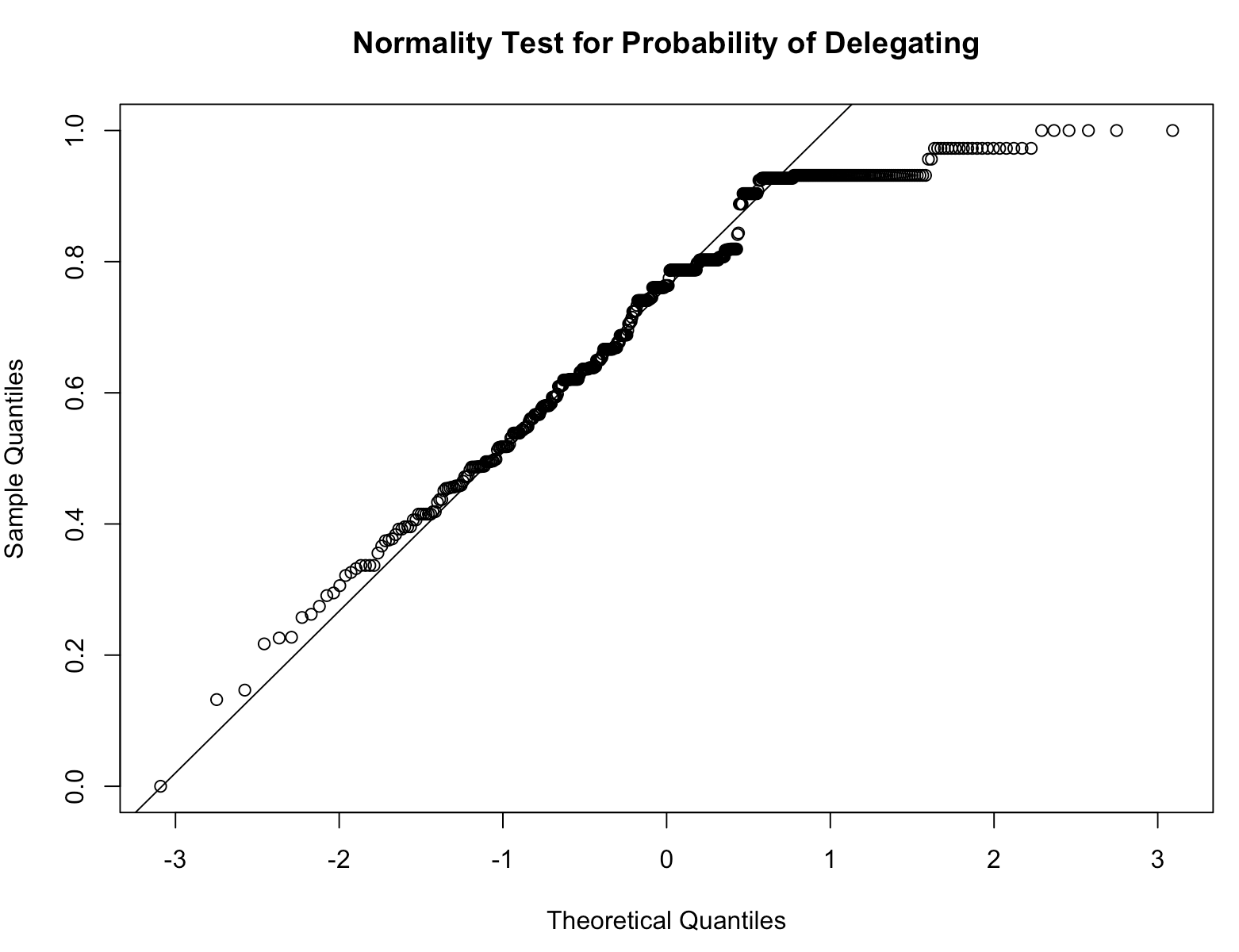}
  \includegraphics[width=0.5\linewidth]{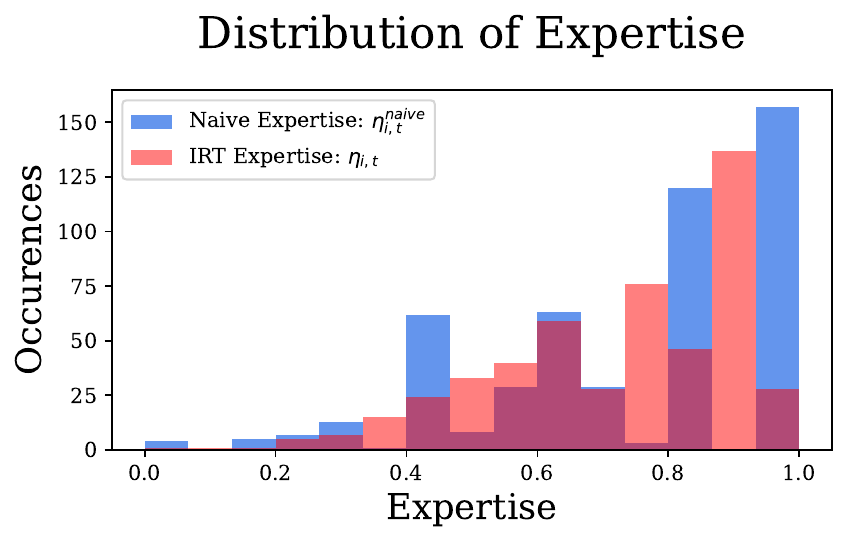}
  \caption{Normality Test for competence in pre-study (left) and Distribution of competence per delegation behavior (right)}
  
  \label{Fig:expertise_distribution_pre}
\end{figure}

\subsection{Delegation Statistics}
We collected $505$ delegation data points, one per participant per task. Of those, $28\%$ were delegations and $57\%$ are direct participants that did not receive any delegation besides their own. Among the delegates, $21\%$ received only one delegation besides their own (hence had weight $2$ in the decision), $11\%$ received two delegations besides their own and just about $1\%$ received five or more delegations besides their own, showing little sign of concentration of power.

\Cref{Fig:graphs} gives examples of some delegation graphs, where nodes labeled by their naive competence $\eta^{naive}_{i, t} = \sum_{i\in [N_e]} v^{D}_{i, e, t} / R_t$. 
The top left plot shows an example of a successful delegation chain to the right, where an expert from experiment $6$ and task \textit{landmarks}, with $\eta_{i, t}=1$ was identified by six other participants either directly or transitively through a local expert $j$ with $\eta_{j, t}=0.8$. On the right, a smaller chain 
shows
two participants delegating to a more competent expert, who in turn delegates to a non-expert. 

Over the course of the experiments, we observed only two delegation cycles of size two (where A delegates to B, who delegates to A), both in Experiment $3$ with $N_3=32.$

Note that in the pre-study, only $27\%$ of the tasks were delegated, which is much less than the $47\%$ delegation rate of the main study.

\begin{figure}[htb]
\begin{center}
      \includegraphics[width=\linewidth]{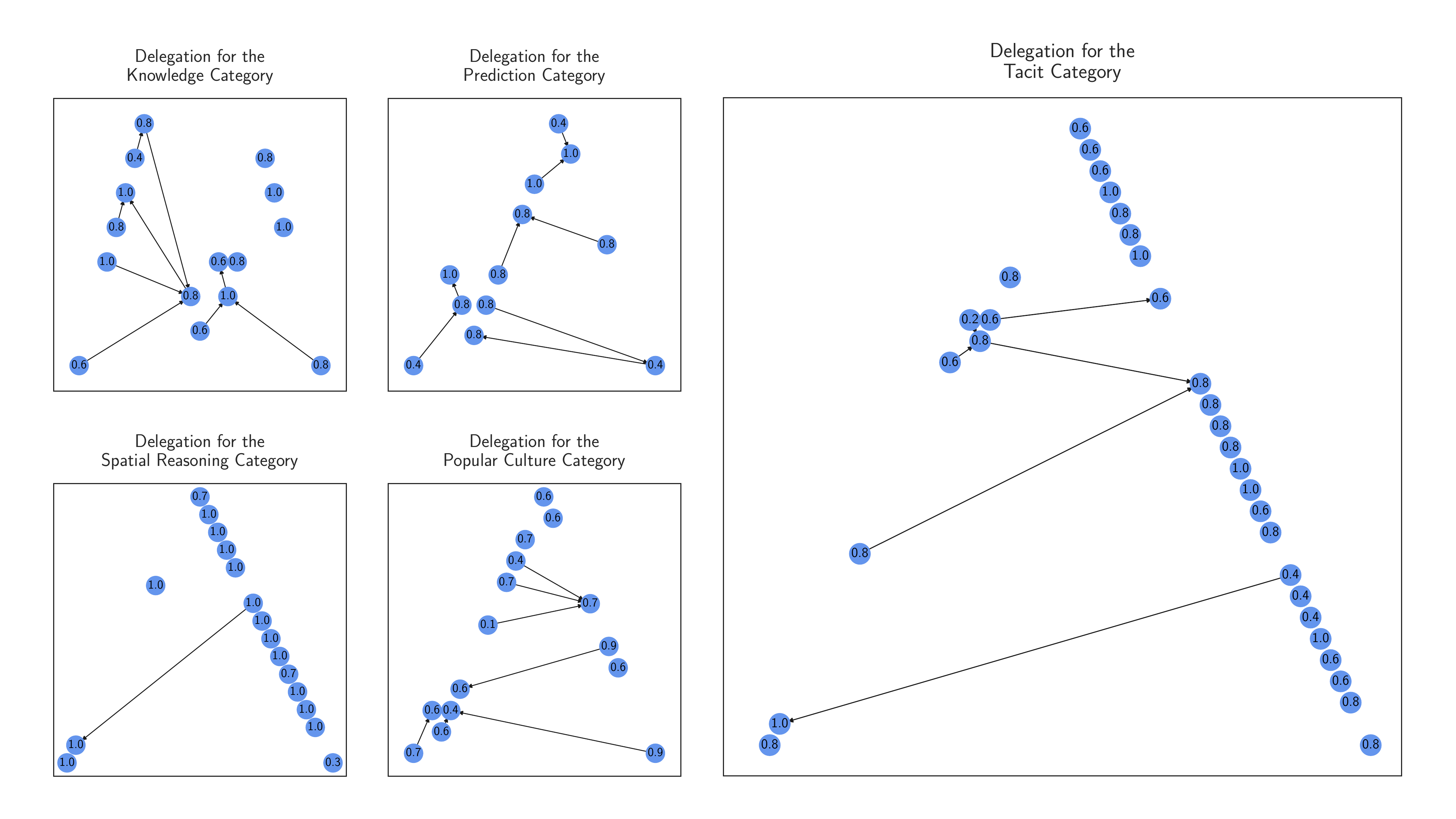}
\end{center}

  \caption{Examples of delegatoin graphs for each of the task categories}
  \label{Fig:graphs}
\end{figure}

\subsection{Estimating $q$ and $\varphi$}
  


Next, we estimate the $q$ and $\varphi$ functions using the same methods applied on the main study. \Cref{table:resultsQ_pre} shows the equivalent results of \Cref{table:resultsQ,tab:task-specific-6} for the pre-study.  We again, using the Kneedle algorithm, find the optimal number of clusters to be $4$. Statistics about the clusters can be found in \Cref{tab:bucket-stats-pre}. \Cref{Fig:phi_batch1} and \Cref{table:resultsPHI_pre_kmeans} show the equivalent information to \Cref{Fig:phiQ} and \Cref{phi:test_table}. We overall find qualitatively very similar results. We can conclude that both $q$ is decreasing and $\varphi$ is increasing in its second coordinate in statistically significant ways.

\begin{table}[htb]
\caption{Estimates of $\beta_q$ \label{table:resultsQ_pre}}
  \centering
  \footnotesize
{\begin{tabular}{@{\extracolsep{5pt}}lcccccccc} 
\\[-1.8ex]\hline 
\hline \\[-1.8ex] 
& \multicolumn{1}{c@{}}{Overall Model} & 
\multicolumn{7}{c@{}}{Tasks Models}\\ 
\cmidrule(lr){2-2}
\cmidrule(lr){3-9}
& \multicolumn{1}{c@{}}{\textit{}}
& \multicolumn{7}{c@{}}{(Tasks)}
\\[1ex] 
& \mc{} & \mc{($T_1$)} & \mc{($T_3$)} & \mc{($T_2$)} & \mc{($T_4$)} & \mc{($T_5$)} & \mc{($T_6$)} & \mc{($T_7$)}\\ 
\midrule
Effect Size $\beta^q$  & $-2.45^{****}$ & $-4.44^{**}$ & $0.02$ & $-5.11^{***}$ & $0.48$ & $-3.52$ & $-3.10^{**}$ & $-3.11$ \\ 
  & (0.50) & (2.24) & (1.69) & (1.50) & (2.28) & (2.19) & (1.34) & (3.81) \\ [1ex]
\midrule
 \begin{tabular}{@{}c@{}}Fixed Effects 
 \end{tabular} & \multicolumn{1}{c}{NA}  & \multicolumn{1}{c}{NA} & \multicolumn{1}{c}{NA} &
  \multicolumn{1}{c}{NA} & \multicolumn{1}{c}{NA} &
  \multicolumn{1}{c}{NA}& \multicolumn{1}{c}{NA} &
  \multicolumn{1}{c}{NA} \\
\midrule
  \begin{tabular}{@{}c@{}} Clustered S.E.\end{tabular} & \multicolumn{1}{c}{\textrm{i}}  & \multicolumn{1}{c}{\textrm{NA}} & \multicolumn{1}{c}{\textrm{NA}} &
  \multicolumn{1}{c}{\textrm{NA}} & \multicolumn{1}{c}{\textrm{NA}} &
  \multicolumn{1}{c}{\textrm{NA}}
   & \multicolumn{1}{c}{\textrm{NA}} &
  \multicolumn{1}{c}{\textrm{NA}}\\ 
\hline \\ [-1.8ex] 
\textit{Note:}  & \multicolumn{8}{r}{$^{*}$p$<$0.1; $^{**}$p$<$0.05; $^{***}$p$<$0.01; $^{****}$p$<$0.0001} \\
\end{tabular}}
\end{table}



\begin{table}[htb]
    \centering
    \caption{Bucket Descriptions}
    \label{tab:bucket-stats-pre}
    \begin{tabular}{cccc}
    \toprule
        Bucket & Interval & Mean competence & Proportion of participants \\
        \midrule
        $c_1$ &$[0.00, 0.50]$ &$0.40$&$15\%$ \\
        $c_2$ &$[0.51, 0.69]$&$0.61$& $ 26\%$\\
        $c_3$ & $[0.70, 0.84]$&$0.78$& $27\%$\\
        $c_4$ & $[0.89, 1.00]$&$0.94$& $33\%$\\
        \bottomrule
    \end{tabular}
\end{table}


\begin{figure}[htb]
  \includegraphics[valign=c, width=0.5\linewidth]{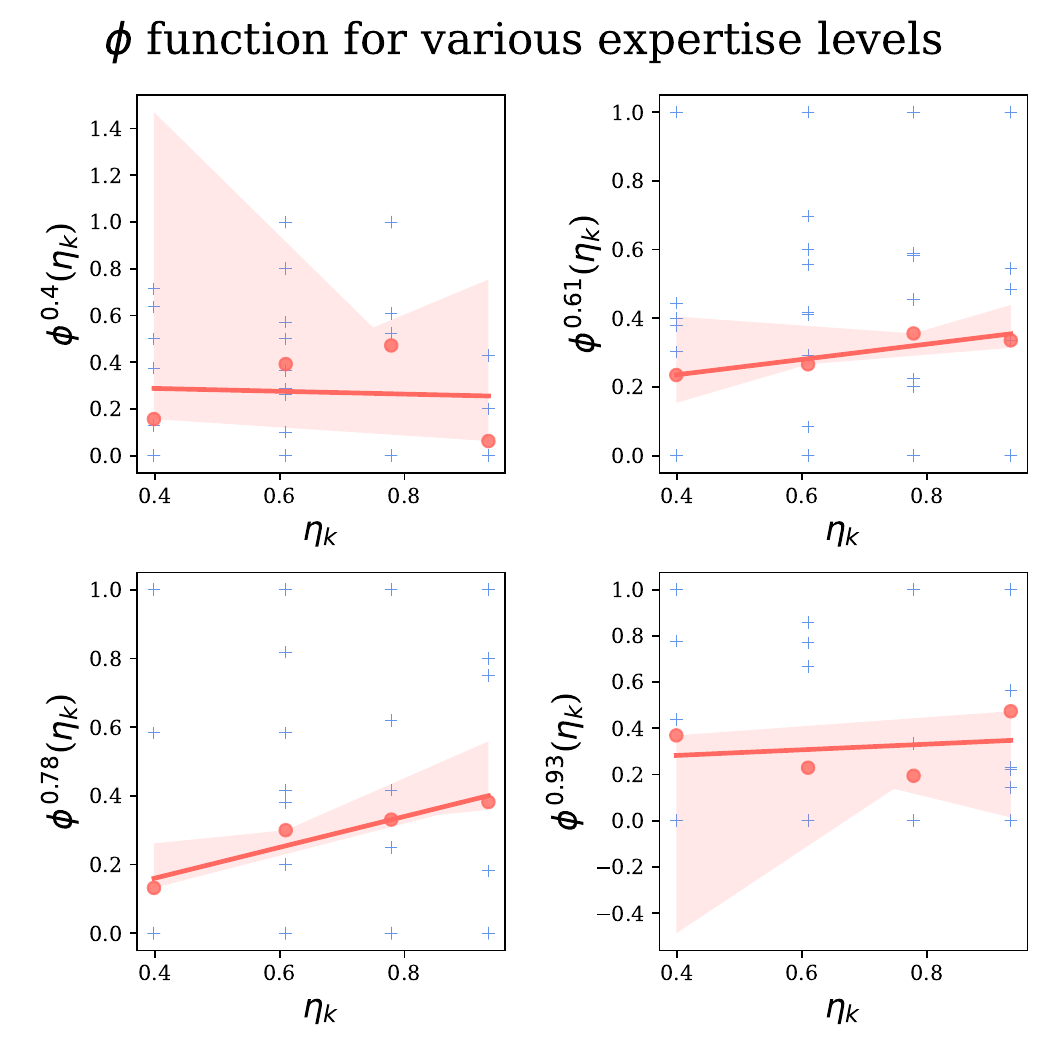}
   \includegraphics[valign=c, width=0.5\linewidth]{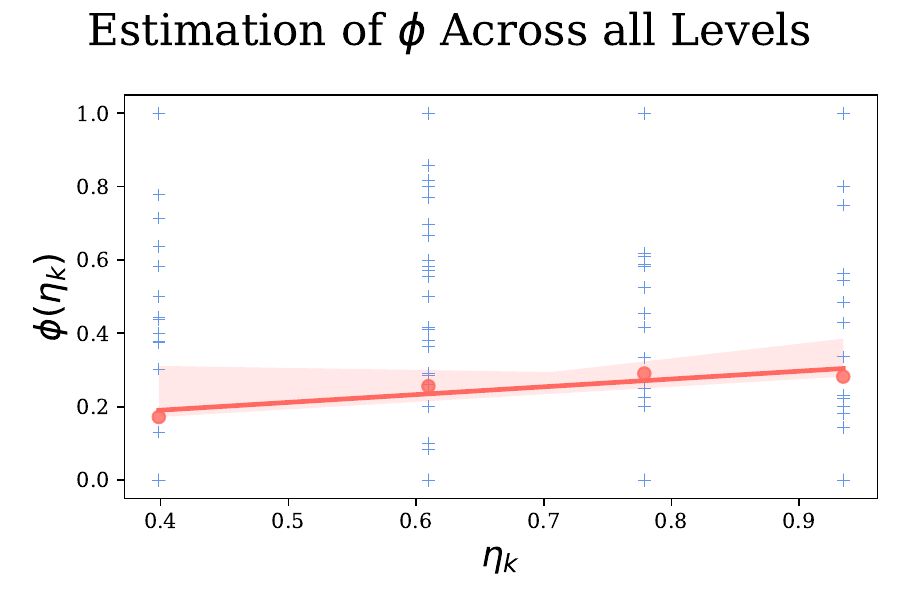}

    \caption{Pooled estimates of $\varphi^{\ell}_{e, t}$, both for each bucket individually, and grouped together. The blue crosses show the values computed for $\varphi_{e,t}^{\ell}(\eta_k)$. The pink dots show the average across all values for that $\eta_k$, and the pink lines correspond to a linear regression over the mean values. We observe increasing trends across the board, with slope (coefficient of determination) being $-0.05 (0.005), 0.18 (0.81), 0.39 (0.91)$ and $0.10 (0.05)$ respectively, for individual buckets, and $0.21 (0.81)$ for the pooled test.}
    
  \label{Fig:phi_batch1}
\end{figure}

\begin{table}[htb]
\caption{Summary of correlation statistics}\label{table:resultsPHI_pre_kmeans}

  \centering
{  \footnotesize\begin{adjustbox}{angle=0}\begin{tabular}{@{\extracolsep{5pt}}lccccc} 
\\[-1.8ex]\hline 
\hline \\[-1.8ex] 
& \multicolumn{1}{c@{}}{Overall} & 
\multicolumn{3}{c@{}}{For fixed $\ell$}\\ 
\cmidrule(lr){2-2}
\cmidrule(lr){3-6} \\[-2ex] 
& & $c_1$ & \mc{$c_2$} & \mc{$c_3$} & \mc{$c_4$}\\ 
\midrule
Correlation & $0.099^{*}$ & 
 $-0.009$ & $0.079$ & $0.20^{*}$ & $0.19$\\ 
  P-value & $8\times 10^{-1}$ & 
  $0.93$& $0.46$ & $9\times10^{-1}$ & $0.11$\\ \hline \\ [-1.8ex] \textit{Note:}  & \multicolumn{5}{r}{$^{*}$p$<$0.1; $^{**}$p$<$0.05; $^{***}$p$<$0.01; $^{****}$p$<$0.0001} 
\end{tabular}\end{adjustbox}}
\end{table}

\end{APPENDIX}
\end{document}